\documentclass[11pt,letterpaper]{article}
\usepackage[margin=1in]{geometry}
\usepackage{amsmath,amssymb,amsthm}
\usepackage{newtxtext}
\usepackage[subscriptcorrection]{newtxmath}
\usepackage{xcolor}
\usepackage{graphicx}
\usepackage{float}
\usepackage{algorithm}
\usepackage{algpseudocode}
\usepackage{multirow}
\usepackage{array}
\usepackage[font=small]{caption}
\usepackage{authblk}
\usepackage[authoryear,round,sort]{natbib}
\usepackage{xurl}
\usepackage{hyperref}
\hypersetup{
  hidelinks,
  pdftitle={Manifold Fitting by Successive Tangent-Space Projection},
  pdfauthor={Yuqing Xia, Bingjie Li, Zhigang Yao}
}
\theoremstyle{plain}
\newtheorem{theorem}{Theorem}
\newtheorem{lemma}{Lemma}
\newtheorem{corollary}{Corollary}
\newtheorem{proposition}{Proposition}
\newtheorem{definition}{Definition}
\theoremstyle{definition}
\newtheorem{assumption}{Assumption}
\theoremstyle{remark}
\newtheorem{remark}{Remark}

\newenvironment{keywords}
  {\par\smallskip\noindent\small\textbf{Keywords: }}
  {\par\medskip}

\newcommand{\bed}{\begin{definition}}
\newcommand{\eed}{\end{definition}}
\newcommand{\beq}{\begin{equation}}
\newcommand{\eeq}{\end{equation}}

\newcommand{\R}{\mathbb{R}}

\newcommand{\E}{\mathbb{E}}

\newcommand{\cov}{\mathrm{cov}}
\newcommand{\reach}{\mathrm{reach}}
\newcommand{\inj}{\mathrm{inj}}
\newcommand{\dist}{\mathrm{dist}}

\graphicspath{{./art/}}

\begin{document}

\title{Manifold Fitting by Successive Tangent-Space Projection}
\author[1]{Yuqing Xia}
\author[2]{Bingjie Li}
\author[3,*]{Zhigang Yao}
\affil[1]{Zhejiang University of Finance and Economics, School of Data
Sciences, Hangzhou, 310018, China}
\affil[2]{Department of Epidemiology and Biostatistics, Ministry of Education
Key Laboratory of Environment and Health, School of Public Health, Tongji
Medical College, Huazhong University of Science and Technology, Wuhan,
Hubei, China}
\affil[3]{Department of Statistics \& Data Science, National University of
Singapore, 21 Lower Kent Ridge Road, Singapore}
\date{}
\maketitle
\begingroup
\renewcommand{\thefootnote}{\fnsymbol{footnote}}
\footnotetext[1]{Corresponding author: \href{mailto:zhigang.yao@nus.edu.sg}{\texttt{zhigang.yao@nus.edu.sg}}.}
\endgroup

\begin{abstract}
{Manifold fitting seeks to recover the geometric structure underlying noisy
ambient observations. We propose Successive Tangent Space Projection
(STSP), an iterative manifold-fitting method that reduces normal noise
while preserving tangential variation. The former improves fitting
accuracy, whereas the latter helps retain geometric coverage. We
characterise the nearby population fixed-point set of STSP as a manifold-fitting
object. Under uniform sampling from a compact smooth manifold with positive
reach and isotropic Gaussian noise, this set lies within $O(\sigma^2)$ of
the underlying manifold, and nearby population orbits converge
geometrically to it. At the finite-sample level, with high probability,
fixed points in the local tube lie within $O(\sigma^2)$ of the underlying
manifold up to sampling error, and empirical orbits initialised within that
tube and generated from a fixed reference sample enter and remain in the same
neighbourhood. We further develop a
multi-scale extension, MS-STSP, designed to reduce curvature bias while
preserving the $O(\sigma^2)$ geometric localization order of STSP at both the population
and finite-sample levels. Numerical experiments show that STSP compares
favourably with competing methods in balancing fitting accuracy and
geometric coverage, and that MS-STSP reduces curvature-induced shrinkage,
particularly under high noise.}
\end{abstract}

\begingroup
\section{Introduction}

Manifold fitting aims to recover a smooth low-dimensional structure from
noisy observations in a high-dimensional ambient space.  It builds on the
manifold hypothesis, which views the apparent high dimensionality of data
as arising from variation around a much simpler underlying structure.
\citet{fefferman} connected this hypothesis to the statistical problem of
determining whether a distribution can be well approximated by a smooth
manifold, while \citet{whiteley2026} explained how such structure can arise
from latent variables, correlation, and stationarity.  The purpose of
manifold fitting is to estimate the geometric structure of possible
signals in their original space, including their local variation and
connectivity \citep{yaosu2026}.

The literature approaches manifold fitting from several directions.
Minimax theory establishes fundamental limits for estimating a manifold
under Hausdorff loss \citep{minmaxME,genovese2012b}.  Implicit
reconstruction methods define the fitted manifold as a ridge or a zero
set constructed from the data \citep{genovese2014,fefferman2018,yaoxiamf};
the definition is global, but the manifold must still be recovered
numerically.  In practice, numerical manifold fitting commonly
relies on neighbourhood information, either to recover an implicitly
defined manifold or to construct fitted points directly.  Two
representative local methods are local PCA and mean shift.  Local PCA
estimates tangent geometry from a neighbourhood covariance matrix and
projects points onto the resulting affine approximation \citep{local_pca}.
Structure-adaptive manifold estimation also uses local geometric
information, iteratively refining local weights to reconstruct a smooth
manifold and denoise the observations \citep{puchkin2022}.  Mean shift is a
mode-seeking procedure based on a kernel density estimate \citep{mean_shift},
while subspace-constrained mean shift (SCMS) adapts the mean-shift update to
estimate density ridges
\citep{ozertem2011locally,genovese2014}.

The differing behaviour of local PCA and mean shift highlights two aspects
of manifold fitting.  Fitting accuracy measures whether the fitted points
lie close to the underlying manifold, whereas geometric coverage measures
whether they represent all parts of it.  Figure~\ref{fig:intro-sphere-motivation}
illustrates how the two methods affect these criteria.  By retaining
tangential positions, local PCA preserves broad geometric coverage but
leaves visible variation normal to the sphere.  Mean shift updates points
along the full mean-shift vector, producing a more concentrated fit but also
uneven density and local gaps.

\begin{figure}[H]
  \centering
  \includegraphics[width=\linewidth]{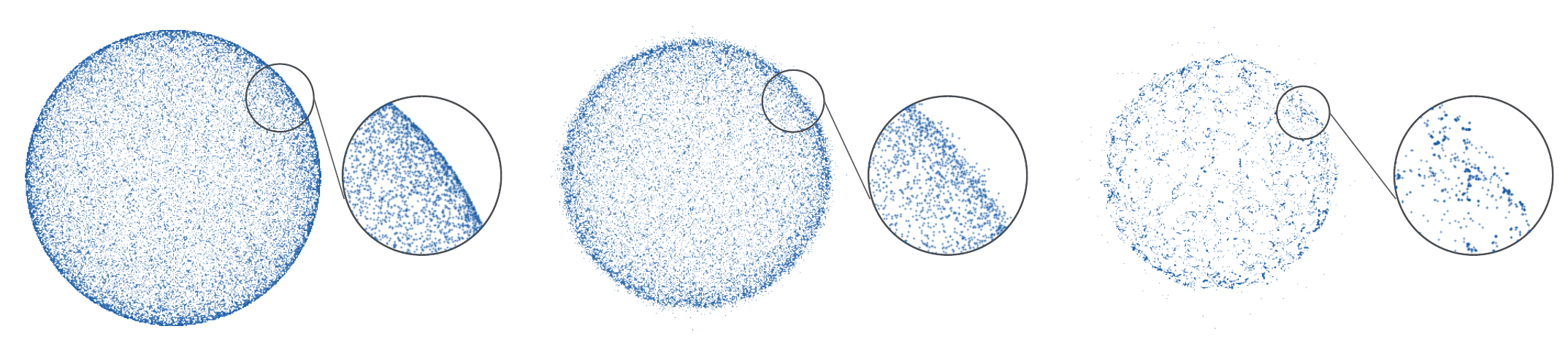}
  \vspace{-5pt}
  {\makebox[\linewidth]{%
    \makebox[.333\linewidth]{(a)}%
    \makebox[.333\linewidth]{(b)}%
    \makebox[.333\linewidth]{(c)}}}
  \caption{(a) The underlying sphere; (b) the local PCA output; and
  (c) the mean-shift output. The same near-boundary region is enlarged in
  each panel.}
  \label{fig:intro-sphere-motivation}
\end{figure}

{These observations motivate Successive Tangent Space Projection (STSP), an
iterative manifold-fitting algorithm whose nearby population fixed-point set
forms a fitted manifold. Each update removes the estimated normal
displacement from a point to its local mean while retaining its tangential
position. The normal correction improves fitting accuracy, whereas the
preservation of tangential position limits drift along the manifold and
helps retain geometric coverage. Repeated application of the update gives
an explicit construction of the fixed-point set. The multi-scale extension,
MS-STSP, combines results obtained at several neighbourhood scales through
squared-scale extrapolation to reduce curvature bias. Our contributions are
as follows.}

\begin{itemize}
  {
  \item \emph{Local geometric refinement.}
  We introduce STSP, which combines local averaging and tangent-space
  projection to correct normal displacement while retaining tangential
  variation.

  \item \emph{Fixed-point geometry and population convergence.}
  We show that the complete nearby population fixed-point set is a
  $d$-dimensional manifold whose nearest-point projection onto $\mathcal M$
  is a diffeomorphism. When $h$ is proportional to $\sigma$, this
  fixed-point manifold lies within $O(\sigma^2)$ of $\mathcal M$, and every
  sufficiently close population orbit converges geometrically to one of its
  points.

  \item \emph{Finite-sample orbit localization.}
  For a fixed reference sample, we derive uniform local-moment and
  tangent-projector bounds and use them to control empirical orbits
  initialised within the local tube. After a geometric transient, each orbit remains within
  $O(\sigma^2+h r_N)$ of the underlying manifold, where $r_N$ depends on the
  intrinsic local sample size $Nh^d$.

  \item \emph{Multi-scale correction.}
  We develop MS-STSP by combining corresponding fixed-reference branches
  through squared-scale extrapolation. The theory shows that stable
  extrapolation preserves the $O(\sigma^2)$ localization order, while the
  experiments investigate its ability to reduce curvature-induced
  shrinkage.}
\end{itemize}
\endgroup

\noindent Section~\ref{sec:STSP} states the model and assumptions and defines the STSP update. Section~\ref{sec:analysis} develops the local-moment expansions, population fixed-point manifold, and finite-sample localization theory. Section~\ref{sec:multiscale} analyses the curvature bias of local averaging and develops the multi-scale procedure. Section~\ref{Sect5} reports numerical studies. Section~\ref{sec:conclusion} concludes. Full proofs are given in the Supplementary Material.

\section{Manifold Fitting via Successive Tangent-Space Projection}
\label{sec:STSP}

{\subsection{Model, assumptions, and notation}}
\label{sec:model-assump}
Let $\mathcal{M}\subset \R^D$ be a compact $d$-dimensional manifold,
{with $1\le d<D$}. We observe
\[
    \mathbf{y}_i = \mathbf{x}_i + \sigma\boldsymbol\xi_i, \qquad i = 1,\dots,N,
\]
where $\mathbf{x}_i \in \mathcal{M}$ is the signal and $\boldsymbol\xi_i \stackrel{\mathrm{i.i.d.}}{\sim} \mathcal{N}(\mathbf0, I_D)$ is standardised ambient noise, independent of $\{\mathbf{x}_i\}_{i=1}^N$.

\begin{assumption}\label{ass:conditions}
The manifold and sampling satisfy the following.
\begin{enumerate}
    \item[\textup{(A1)}] The signals $\{\mathbf{x}_i\}_{i=1}^N$ are drawn independently from the uniform distribution on $\mathcal{M}$.
    \item[\textup{(A2)}] {$\mathcal{M}$ is a $C^5$ embedded submanifold without boundary. In geodesic normal charts of the radius used below, the embedding and its derivatives through order five are bounded by a fixed constant $K_{\mathrm{geo}}$.}
    \item[\textup{(A3)}] $\reach(\mathcal{M}) \ge \tau > 0$, where
    \[
    \reach(\mathcal M)=\sup\{r>0:\text{every $\mathbf y$ with $\dist(\mathbf y,\mathcal M)<r$
    has a unique nearest point in $\mathcal M$}\}.
    \]
    \item[\textup{(A4)}] {Let $\iota_{\mathcal M}=\inf_{x\in\mathcal M}\inj_x(\mathcal M)>0$ be the global injectivity radius of the induced Riemannian metric, and set $r_{\mathrm{geo}}=\min\{\tau,\iota_{\mathcal M}\}$. The dimensions $d,D$ are fixed, $\sigma<1$, and $\sigma/r_{\mathrm{geo}}$ is sufficiently small.}
\end{enumerate}
\end{assumption}

{Conditions (A2) and (A3) control the geometric complexity of $\mathcal{M}$: within a neighbourhood of radius $h \ll \tau$, the manifold deviates from its tangent space by at most order $h^2$, and the positive reach excludes sharp features and near self-intersections that would make recovery ill-posed. The injectivity radius in (A4) is the largest intrinsic scale on which geodesic normal coordinates are single-valued uniformly over $\mathcal M$. Condition (A4) ensures the noise is small relative to both the reach and injectivity scales, so that normal fluctuations can be distinguished from genuine structure.}

{We use the following notation throughout. For a nonempty set
$\mathcal S\subset\R^D$, define
$\dist(\mathbf y,\mathcal S)=\inf_{\mathbf s\in\mathcal S}
\|\mathbf y-\mathbf s\|$. For nonempty compact sets $\mathcal S_1$ and
$\mathcal S_2$, their Hausdorff distance is
$d_H(\mathcal S_1,\mathcal S_2)=
\max\{\sup_{u\in\mathcal S_1}\dist(u,\mathcal S_2),
\sup_{v\in\mathcal S_2}\dist(v,\mathcal S_1)\}$. Whenever the nearest point is unique,
$\Pi_{\mathcal S}(\mathbf y)$ denotes that point. In particular, if
$\dist(\mathbf y,\mathcal M)<\tau$, then (A3) gives the unique projection
$\Pi_{\mathcal M}(\mathbf y)$. For $x\in\mathcal M$, $T_x\mathcal M$ and
$N_x\mathcal M$ denote the tangent and normal spaces, $\mathbf P_x$ is the
orthogonal projector onto $T_x\mathcal M$, and
$\mathbf Q_x:=I_D-\mathbf P_x$ is the orthogonal projector onto
$N_x\mathcal M$. Thus $\Pi_{\mathcal S}$ returns a projected point,
whereas bold capital $\mathbf P$ and $\mathbf Q$ denote linear projector
operators. The symbol $\mathrm{II}_x$ denotes the second fundamental form,
and $I_m$ denotes the $m\times m$ identity matrix. For $\mathbf b\in\R^D$
and $r>0$, $B_{\mathrm E}(\mathbf b,r)=\{\mathbf y\in\R^D:
\|\mathbf y-\mathbf b\|\le r\}$ is the ambient Euclidean ball. For
$x\in\mathcal M$, $B_g(x,r)=\{z\in\mathcal M:d_g(x,z)\le r\}$ is the
geodesic ball, and $d_g$ is geodesic distance. The symbol $dA$ denotes the
induced $d$-dimensional Riemannian volume measure on $\mathcal M$. If
$0<r<\inj_x(\mathcal M)$, the exponential map
$\exp_x$ is one-to-one on $\{u\in T_x\mathcal M:\|u\|<r\}$; writing
$z=\exp_x(u)$ gives geodesic normal coordinates centred at $x$, and then
$d_g(x,z)=\|u\|$. Bold lowercase symbols denote observed, reference, or
query vectors, whereas plain lowercase symbols denote deterministic
manifold points or local-coordinate vectors. Hats denote empirical quantities,
and parenthesized superscripts denote iterates or repeated composition.
Constants may depend on
$d,D,\tau,K_{\mathrm{geo}},\iota_{\mathcal M}$ and $\mathrm{Area}(\mathcal M)$,
but not on $N,\sigma$, or $h$, and their values may change between
occurrences.}

{
\subsection{From an implicit fitting equation to an explicit construction}
\label{sec:implicit-explicit}

\citet{yaoxiamf} define a fitted manifold through the zero set of a
smoothed local normal residual, denoted here by $\mathbf G_h$:
\[
  \widehat{\mathcal M}_{h,{\mathrm{YX}}}
  =\{\mathbf y:\mathbf G_h(\mathbf y)=\mathbf0\}.
\]
The fitted locations have zero estimated normal fitting bias, and their
analysis places this zero-level set close to the underlying manifold.
This construction provides a geometric starting point for manifold fitting:
characterizing fitted locations through local normal residuals. However,
an implicit equation characterizes the fitted set but still requires a
numerical procedure to produce points on it. Their implementation minimizes
$\|\mathbf G_h(\mathbf y)\|_2^2$ by gradient descent with line search.
Membership in the fitted set requires a zero objective, whereas a stationary
numerical solution may have a positive objective. This distinction motivates
an explicit construction whose iterates can be analysed directly relative
to the underlying manifold.

For an affine-space correction, let $\mathbf m(\mathbf y)$ be a weighted
local centre, let $T_{\mathrm{loc}}(\mathbf y)$ be a local tangent space, and let
$\mathbf P_{\mathrm{loc}}(\mathbf y)$ be its orthogonal projector. The associated
local affine space and normal residual are
\[
 \mathcal A(\mathbf y)=\mathbf m(\mathbf y)+T_{\mathrm{loc}}(\mathbf y),
 \qquad
 \mathbf R(\mathbf y)
 =\mathbf y-\Pi_{\mathcal A(\mathbf y)}(\mathbf y)
 =\{I_D-\mathbf P_{\mathrm{loc}}(\mathbf y)\}
   \{\mathbf y-\mathbf m(\mathbf y)\}.
\]
The zero-residual set of this affine-space construction is
\[
  \widehat{\mathcal M}
  =\{\mathbf y:\mathbf R(\mathbf y)=\mathbf0\}.
\]
Its geometric condition is that the displacement from the local centre has
zero projection onto the estimated normal space. Accordingly, for any
affine-space residual of the form above, define the corresponding
explicit correction by
\[
 \mathcal T(\mathbf y)
 =\mathbf m(\mathbf y)+\mathbf P_{\mathrm{loc}}(\mathbf y)
   \{\mathbf y-\mathbf m(\mathbf y)\}
 =\mathbf y-\mathbf R(\mathbf y).
\]
Then
$\widehat{\mathcal M}=\operatorname{Fix}(\mathcal T)$: the implicit
zero-residual equation is equivalently the fixed-point equation of a local
tangent-space correction.

Motivated by the normal-residual characterization, STSP constructs an
explicit iteration from hard-radius conditional moments. Section~\ref{sec:analysis}
studies its fixed-point set and local trajectories. Let $\mathbf Y$ denote a
generic observation from the model and set
$p_h(\mathbf y)=\Pr\{\mathbf Y\in B_{\mathrm E}(\mathbf y,h)\}$. Wherever
$p_h(\mathbf y)>0$, let $\boldsymbol\mu_h(\mathbf y)$ and
$\boldsymbol\Sigma_h(\mathbf y)$ be the conditional mean and covariance of
$\mathbf Y$ given $\mathbf Y\in B_{\mathrm E}(\mathbf y,h)$. Let
$\mathcal D_h$ collect points with positive local probability and a strict
eigengap. For $\mathbf y\in\mathcal D_h$, let $T_h(\mathbf y)$ be the uniquely
defined leading-$d$ eigenspace of $\boldsymbol\Sigma_h(\mathbf y)$, and let
$\mathbf P_h(\mathbf y)$ be its orthogonal projector. The population domain,
local affine space, residual, fitting object, and explicit correction are
\begin{flalign*}
 &\mathcal D_h
 =\bigl\{\mathbf y:p_h(\mathbf y)>0,\ 
   \lambda_d\{\boldsymbol\Sigma_h(\mathbf y)\}>
   \lambda_{d+1}\{\boldsymbol\Sigma_h(\mathbf y)\}\bigr\}, &&\\
 &\mathcal A_h(\mathbf y)
 =\boldsymbol\mu_h(\mathbf y)+T_h(\mathbf y), &&\\
 &\mathbf F_h(\mathbf y)
 =\{I_D-\mathbf P_h(\mathbf y)\}
   \{\mathbf y-\boldsymbol\mu_h(\mathbf y)\},\qquad
   \mathcal S_{h,{\mathrm{pop}}}
 =\{\mathbf y\in\mathcal D_h:\mathbf F_h(\mathbf y)=\mathbf0\}, &&\\
 &\mathcal T_h(\mathbf y)
 =\Pi_{\mathcal A_h(\mathbf y)}(\mathbf y)
 =\boldsymbol\mu_h(\mathbf y)+\mathbf P_h(\mathbf y)
   \{\mathbf y-\boldsymbol\mu_h(\mathbf y)\}
 =\mathbf y-\mathbf F_h(\mathbf y). &&
\end{flalign*}
Consequently,
$\mathcal S_{h,{\mathrm{pop}}}=\operatorname{Fix}(\mathcal T_h)$, and manifold
fitting is implemented by the successive corrections
$\mathbf y^{(t+1)}=\mathcal T_h(\mathbf y^{(t)})$. Figure~\ref{fig:implicit-explicit} illustrates normal-residual fitting
and successive tangent-space corrections.}

\begin{figure}[H]
  \centering
  \begin{minipage}[t]{.485\linewidth}
    \centering
    \includegraphics[width=\linewidth]{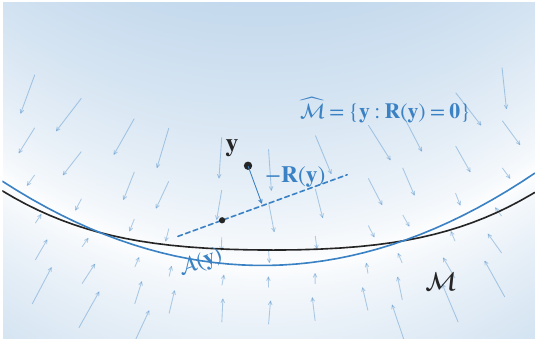}
    \smallskip

    \textup{(a) Residual zero set}
  \end{minipage}%
  \hfill
  \begin{minipage}[t]{.485\linewidth}
    \centering
    \includegraphics[width=\linewidth]{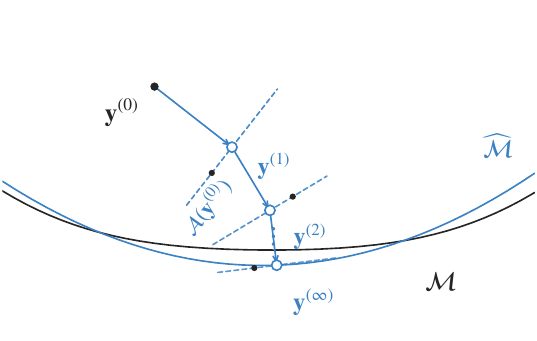}
    \smallskip

    \textup{(b) Successive tangent-space projections}
  \end{minipage}
  \caption[Implicit and explicit manifold fitting]{{Implicit and explicit manifold fitting. (a) Blue
  intensity represents the length $\|\mathbf R(\mathbf y)\|$ of the local
  affine-space residual, and the white locus is the generic fitted set
  $\widehat{\mathcal M}=\{\mathbf y:\mathbf R(\mathbf y)=\mathbf0\}$.
  (b) A schematic sequence of projections onto the query-local affine spaces
  $\mathcal A(\mathbf y)$ approaching the fitted set. For STSP, the residual
  and affine space are $\mathbf R=\mathbf F_h$ and
  $\mathcal A=\mathcal A_h$, respectively.}}
  \label{fig:implicit-explicit}
\end{figure}

\subsection{The STSP update}
\label{sec:algorithm}
{
Set $\mathbf z_j=\mathbf y_j$ and keep
$\mathcal Z=\{\mathbf z_j\}_{j=1}^N$ fixed throughout the algorithm.
At iteration $t$ only the query locations $\widehat{\mathbf y}_i^{(t)}$
change.
For each query, first form the radius-$h$ neighbourhood in the fixed
reference cloud,
\[
  \mathcal I_{\mathcal Z,h}(\mathbf y)
  = \bigl\{ j : \|\mathbf{z}_j-\mathbf y\|\le h \bigr\}.
\]
Let $n_{\mathcal Z,h}(\mathbf y)
=|\mathcal I_{\mathcal Z,h}(\mathbf y)|$ denote the corresponding count.
The radius $h$ must be large enough for stable estimation and small enough
that the manifold is nearly linear within it. Whenever
$n_{\mathcal Z,h}(\mathbf y)>0$, compute the local mean
$\widehat{\boldsymbol\mu}_{\mathcal Z,h}(\mathbf y)
=n_{\mathcal Z,h}(\mathbf y)^{-1}
\sum_{j\in\mathcal I_{\mathcal Z,h}(\mathbf y)}\mathbf z_j$
and covariance
\[
  \widehat{\boldsymbol\Sigma}_{\mathcal Z,h}(\mathbf y)
    = \frac{1}{n_{\mathcal Z,h}(\mathbf y)}
      \sum_{j\in\mathcal I_{\mathcal Z,h}(\mathbf y)}
      \{\mathbf z_j-\widehat{\boldsymbol\mu}_{\mathcal Z,h}(\mathbf y)\}
      \{\mathbf z_j-\widehat{\boldsymbol\mu}_{\mathcal Z,h}(\mathbf y)\}^{\top}.
\]
If, in addition, the $d$th covariance eigenvalue is strictly greater than the
$(d+1)$st, let $\widehat T_{\mathcal Z,h}(\mathbf y)$ be its uniquely defined
leading-$d$ eigenspace, and let
$\widehat{\mathbf P}_{\mathcal Z,h}(\mathbf y)$ be the orthogonal projector
onto $\widehat T_{\mathcal Z,h}(\mathbf y)$. On this empirical domain, define
\[
 \widehat{\mathcal T}_{\mathcal Z,h}(\mathbf y)
 =\widehat{\boldsymbol\mu}_{\mathcal Z,h}(\mathbf y)
  +\widehat{\mathbf P}_{\mathcal Z,h}(\mathbf y)
   \{\mathbf y-\widehat{\boldsymbol\mu}_{\mathcal Z,h}(\mathbf y)\}.
\]
All iterations apply this same deterministic map through
$\widehat{\mathbf y}_i^{(t+1)}=\widehat{\mathcal T}_{\mathcal Z,h}
(\widehat{\mathbf y}_i^{(t)})$. Conditional on
$\mathcal Z$, the $N$ query orbits are uncoupled.  The update retains the
estimated tangential component relative to the local mean and discards the
estimated normal component.
The hard-radius operator is well defined along a trajectory when every
visited neighbourhood is nonempty and its $d$th and $(d+1)$st covariance
eigenvalues are separated. The finite-sample theory below supplies uniform
conditions for both properties.

\begin{algorithm}[!ht]

\caption{Successive Tangent-Space Projection (STSP)}
\label{alg:stsp}
\begin{algorithmic}[1]
\Require Reference points
  \(\mathcal Z=\{\mathbf z_j\}_{j=1}^N\subset\R^D\), initial queries
  \(\widehat{\mathbf y}_i^{(0)}=\mathbf z_i\), dimension \(d\), radius
  \(h\), and iteration horizon \(L\).
\Ensure Fitted queries \(\{\widehat{\mathbf y}_i^{(L)}\}_{i=1}^N\).
\For{\(t=0\) to \(L-1\)}
  \For{\(i=1\) to \(N\)}
    \State
      \(\mathcal I_i^{(t)}
       :=\mathcal I_{\mathcal Z,h}(\widehat{\mathbf y}_i^{(t)})\).
    \State Set
      \(\widehat{\boldsymbol\mu}_i^{(t)}
      :=\widehat{\boldsymbol\mu}_{\mathcal Z,h}(\widehat{\mathbf y}_i^{(t)})\) and
      \(\widehat{\boldsymbol\Sigma}_i^{(t)}
      :=\widehat{\boldsymbol\Sigma}_{\mathcal Z,h}(\widehat{\mathbf y}_i^{(t)})\),
      computed from \(\{\mathbf z_j:j\in\mathcal I_i^{(t)}\}\).
    \State Set \(\widehat{\mathbf P}_i^{(t)}
      :=\widehat{\mathbf P}_{\mathcal Z,h}(\widehat{\mathbf y}_i^{(t)})\), the
      leading-\(d\) projector of \(\widehat{\boldsymbol\Sigma}_i^{(t)}\).
    \State
      \(\widehat{\mathbf y}_i^{(t+1)}
       =\widehat{\boldsymbol\mu}_i^{(t)}
        +\widehat{\mathbf P}_i^{(t)}
          (\widehat{\mathbf y}_i^{(t)}-\widehat{\boldsymbol\mu}_i^{(t)})\).
  \EndFor
\EndFor
\State \Return \(\{\widehat{\mathbf y}_i^{(L)}\}_{i=1}^N\).
\end{algorithmic}
\end{algorithm}
}

{
\subsection{Relation to local PCA, mean shift, and SCMS}
\label{sec:relationship}

Write $\mathbf Q_h(\mathbf y)=I_D-\mathbf P_h(\mathbf y)$. The population
STSP update has the projected-shift form
\begin{equation}
\label{eq:stsp-projected-shift}
 \mathcal T_h(\mathbf y)
 =\mathbf y+\mathbf Q_h(\mathbf y)
   \{\boldsymbol\mu_h(\mathbf y)-\mathbf y\}.
\end{equation}
It therefore applies the estimated normal component of the displacement
from the query to its hard-radius conditional mean while preserving
the query's tangential component.

The local-PCA comparator used in Section~\ref{Sect5} applies one
projection onto an estimated affine tangent space. STSP recomputes the
query-local conditional moments as the query moves and repeatedly applies
the normal correction in \eqref{eq:stsp-projected-shift}, while the
distribution or sample that supplies those moments remains fixed.
Mean shift instead takes the full displacement to the local mean.

This algebraic form is closely related to SCMS\@. A standard SCMS update can
be written schematically as
\[
 \mathcal T_{{\mathrm{SCMS}},h}(\mathbf y)
 =\mathbf y+\widetilde{\mathbf Q}_h(\mathbf y)
   \{\widetilde{\boldsymbol\mu}_h(\mathbf y)-\mathbf y\},
\]
where $\widetilde{\boldsymbol\mu}_h(\mathbf y)-\mathbf y$ is a kernel
mean-shift vector and $\widetilde{\mathbf Q}_h(\mathbf y)$ projects onto the
estimated normal eigenspace of a kernel-density Hessian. \begingroup For Gaussian
weights, write $\boldsymbol\mu_{G,h}$ and $\boldsymbol\Sigma_{G,h}$ for the
weighted mean and covariance, and $f_{G,h}$ for the corresponding Gaussian
kernel density estimate. Then
$\nabla\log f_{G,h}=(\boldsymbol\mu_{G,h}-\mathbf y)/h^2$ and
$\nabla^2\log f_{G,h}=\boldsymbol\Sigma_{G,h}/h^4-I_D/h^2$.
Thus a projected update using Gaussian moments and the leading covariance
eigenspace has the same local form as SCMS-log-KDE. The STSP studied here
instead uses hard-radius conditional moments. Our contribution concerns the
geometry of its nearby fixed-point set and finite-sample control of its
iterates, rather than the projected mean-shift form alone.\endgroup\ 
Section~\ref{sec:analysis} establishes these local results.
Table~\ref{tab:projected-local-updates} summarizes the local quantities,
updates, and targets of these methods.

\begin{table}[H]
\centering

\caption{Schematic projection-based comparison of the local updates
considered here.}
\label{tab:projected-local-updates}
\small
\setlength{\tabcolsep}{4pt}
\begin{tabular}{>{\raggedright\arraybackslash}p{.14\linewidth}
                >{\raggedright\arraybackslash}p{.25\linewidth}
                >{\raggedright\arraybackslash}p{.27\linewidth}
                >{\raggedright\arraybackslash}p{.24\linewidth}}
\hline
Method & Local quantities & Update & Target or role \\
\hline
Mean shift & Local mean & Full displacement toward the local mean & Density modes \\
SCMS & Kernel mean-shift vector and normal eigenspace of the density Hessian & Normal projection of the kernel mean-shift vector & Density ridge \\
Local PCA & Estimated affine tangent space & Projection onto the affine tangent space & Local affine approximation \\
STSP & Local mean and covariance within a fixed radius & Repeated normal correction using fixed reference data & Fixed-point manifold near the underlying manifold \\
\hline
\end{tabular}
\end{table}

}

\section{Theoretical Analysis}
\label{sec:analysis}

Fix $c>0$ and set $h=c\sigma$, with
$c\sigma\le r_{\mathrm{geo}}/2$.  The population results are stated as
$\sigma\to0$.  In the finite-sample results, $N\to\infty$ jointly with
$\sigma\to0$. With $d,D,c,A$ fixed, the condition
$N\sigma^{d+2}/\{D\log(1/\sigma)+\log(1/\delta)\}\ge C_0$ for a sufficiently
large constant $C_0$ gives uniform empirical consistency and preserves the
population $O(\sigma^2)$ localization order.

\subsection{Flat reference model and uniform local expansions}
\label{sec:flat-expansion}
Proofs of the results in this subsection are given in Supplementary
Sections~S2.1--S2.2.  Let
$\Phi_\rho=\{\mathbf y:\dist(\mathbf y,\mathcal M)\le\rho\}$. For
$\mathbf y\in\Phi_\rho$, write $x=\Pi_{\mathcal M}(\mathbf y)$ and
$\boldsymbol\nu(\mathbf y)=\mathbf y-x\in N_x\mathcal M$. We analyse the
population operator on $\Phi_{A\sigma}$ for a fixed $A<\infty$.

Put $q=D-d$.  For $a\in\R^q$, let
$W=(W_{\mathrm T},W_{\mathrm N})\in\R^d\times\R^q$ have
density
\begin{equation}
\label{eq:flat-law}
 p_{c,a}(w)=
 \frac{\mathbf 1\{\|w\|\le c\}
       \exp\{-\|a+w_{\mathrm N}\|^2/2\}}
      {\int_{\|v\|\le c}\exp\{-\|a+v_{\mathrm N}\|^2/2\}\,dv}.
\end{equation}
Let $\E_{c,a}$ and $\cov_{c,a}$ denote expectation and covariance under
$p_{c,a}$, respectively. Define
\[
 m_c(a)=a+\E_{c,a}W_{\mathrm N},\qquad S_c(a)=\cov_{c,a}(W).
\]
The law in \eqref{eq:flat-law} is the exact scaled conditional law for an
affine $d$-plane. We call $(c,A)$ \emph{spectrally admissible} when there
is $g_0>0$ such that, for every $\|a\|\le A$,
\begin{equation}
\label{eq:admissible}
 \lambda_{\min}\{\cov_{c,a}(W_{\mathrm T})\}
 -\lambda_{\max}\{\cov_{c,a}(W_{\mathrm N})\}\ge g_0.
\end{equation}
The two ingredients have distinct roles. Spectral admissibility identifies
the leading-$d$ covariance eigenspace with the tangent space, whereas the
score identity
\begin{equation}
\label{eq:normal-score-identity}
 \mathrm Dm_c(a)=I_q-\cov_{c,a}(W_{\mathrm N})
\end{equation}
and the Brascamp--Lieb inequality imply
$\sup_{\|a\|\le A}\|\mathrm Dm_c(a)\|_{\mathrm{op}}\le\kappa_0<1$.

\begin{remark}[Scale-free spectral condition]
The normalised gap in \eqref{eq:admissible} depends only on
$(d,q,c,A)$, where $q=D-d$, and not on $\sigma$, the reach, or other
geometric features of $\mathcal M$. It can therefore be checked once for
each dimension pair, radius ratio, and normalised tube width. Supplementary
Section~S2.1 reports the minimum gaps for the values used in the numerical
studies.
\end{remark}

\begin{lemma}[Uniform local-moment expansion]
\label{lem:pop-mean-cov}
Under Assumption~\ref{ass:conditions} and spectral admissibility, let
$\mathbf y\in\Phi_{A\sigma}$, $x=\Pi_{\mathcal M}(\mathbf y)$, and write
$\boldsymbol\nu(\mathbf y)=\sigma a$ in orthonormal tangent--normal frames at
$x$. Uniformly over $\Phi_{A\sigma}$,
\[
 \boldsymbol\mu_h(\mathbf y)
 =x+\sigma(\mathbf0_d,m_c(a))+O(\sigma^2),\qquad
 \boldsymbol\Sigma_h(\mathbf y)=\sigma^2S_c(a)+O(\sigma^3),
\]
and
\[
 \|\mathbf P_h(\mathbf y)-\mathbf P_x\|_{\mathrm{op}}\le C\sigma.
\]
Consequently,
\[
 \mathbf Q_x\{\mathcal T_h(\mathbf y)-x\}
   =\sigma m_c(a)+O(\sigma^2),\qquad
 \|\mathbf P_x\{\mathcal T_h(\mathbf y)-x\}\|\le C\sigma^2.
\]
\end{lemma}

\subsection{Population fixed-point manifold and attraction}
\label{sec:population-attraction}
Proofs of the population recursion and fixed-point manifold theorem are
given in Supplementary Section~S2.2. Throughout this subsection,
$(c,A)$ is spectrally admissible and $c\sigma\le r_{\mathrm{geo}}/2$.

\begin{theorem}[Local normal-distance recursion]
\label{thm:contraction}
Under Assumption~\ref{ass:conditions} and the spectral admissibility
condition above, there are $\sigma_{\mathrm{con}}>0$, $C<\infty$, and
$\kappa\in(\kappa_0,1)$ such that, for $\sigma\le\sigma_{\mathrm{con}}$ and
$\mathbf y\in\Phi_{A\sigma}$,
\[
 \dist\bigl(\mathcal T_h(\mathbf y),\mathcal M\bigr)
 \le \kappa\,\dist(\mathbf y,\mathcal M)+C\sigma^2.
\]
For sufficiently small $\sigma$, $\Phi_{A\sigma}$ is forward invariant,
and any population orbit starting there obeys
\[
 \dist\bigl(\mathcal T_h^{(t)}(\mathbf y),\mathcal M\bigr)
 \le\kappa^t\dist(\mathbf y,\mathcal M)
 +\frac{C\sigma^2}{1-\kappa}.
\]
\end{theorem}

To lift normal localization to ambient convergence, we identify the
complete nearby zero set of the update residual and control the map
transverse to that set. Recall $\mathbf F_h$ and
$\mathcal S_{h,{\mathrm{pop}}}$ from Section~\ref{sec:implicit-explicit}.

\begin{theorem}[Population fixed-point manifold and local orbit convergence]
\label{thm:orbit-convergence}
Fix $0<A_0<A$.  Under the conditions of
Theorem~\ref{thm:contraction}, there are
$C,C_{\mathrm{step}},\rho_0>0$ and $\bar\kappa\in(0,1)$ such that, for all
sufficiently small $\sigma$, the
fixed-point set in $\Phi_{A_0\sigma}$ is the compact $C^1$ normal graph
\[
 \mathcal M_h=\{x+b_h(x):x\in\mathcal M,\ b_h(x)\in N_x\mathcal M\},
 \qquad \|b_h\|_{C^1}\le C\sigma^2.
\]
Its nearest-point projection onto $\mathcal M$ is a $C^1$ diffeomorphism
and
\[
 d_H(\mathcal M_h,\mathcal M)\le C\sigma^2.
\]
There is a neighbourhood
$\mathcal U_h=\{\mathbf y:\dist(\mathbf y,\mathcal M_h)<\rho_0\sigma\}
\subset\Phi_{A\sigma}$ such that every
$\mathbf y^{(0)}\in\mathcal U_h$ generates an orbit
$\mathbf y^{(t+1)}=\mathcal T_h(\mathbf y^{(t)})$ that converges to a
fixed point $\mathbf y^{(\infty)}\in\mathcal M_h$.  More precisely,
\[
 \dist(\mathbf y^{(t)},\mathcal M_h)
 \le\bar\kappa^t\dist(\mathbf y^{(0)},\mathcal M_h)
\]
and
\[
 \|\mathbf y^{(t)}-\mathbf y^{(\infty)}\|
 \le
 \frac{C_{\mathrm{step}}\bar\kappa^t}{1-\bar\kappa}
 \dist(\mathbf y^{(0)},\mathcal M_h).
\]
In particular,
$\dist(\mathbf y^{(\infty)},\mathcal M)\le C\sigma^2$.
\end{theorem}

\subsection{Finite-sample localization}
\label{sec:sample-fixpt}
Proofs of the uniform empirical approximation, fixed-point localization,
and orbit localization results are given in Supplementary Section~S3.1.
We analyse the empirical operator $\widehat{\mathcal T}_{\mathcal Z,h}$
constructed from the fixed i.i.d.\ reference sample $\mathcal Z$.

Let $0<\delta<1$ and define
\[
 r_N=
 \sqrt{\frac{D\log(1/h)+\log(1/\delta)}{Nh^d}}
 +\frac{D\log(1/h)+\log(1/\delta)}{Nh^d}.
\]

\begin{lemma}[Uniform empirical approximation]
\label{lem:empirical-approximation}
Under the conditions of Lemma~\ref{lem:pop-mean-cov}, there are constants
$C,r_0,c_{\mathrm{gap}}>0$ such that, if $r_N\le r_0$, then with
probability at least $1-\delta$,
\[
 \sup_{\Phi_{A\sigma}}\|\widehat{\boldsymbol\mu}_{\mathcal Z,h}
       -\boldsymbol\mu_h\|\le Chr_N,\qquad
 \sup_{\Phi_{A\sigma}}\|\widehat{\boldsymbol\Sigma}_{\mathcal Z,h}
       -\boldsymbol\Sigma_h\|_{\mathrm{op}}\le Ch^2r_N,
\]
\[
 \sup_{\Phi_{A\sigma}}\|\widehat{\mathbf P}_{\mathcal Z,h}
       -\mathbf P_h\|_{\mathrm{op}}\le Cr_N,\qquad
 \inf_{\Phi_{A\sigma}}
 \{\widehat\lambda_d-\widehat\lambda_{d+1}\}
 \ge c_{\mathrm{gap}}h^2.
\]
\end{lemma}

\begin{corollary}[Localization of empirical fixed points]
\label{thm:sample-fixpt-dist}
Under the conditions of Theorem~\ref{thm:contraction}, there are constants
$C,r_{\mathrm{fp}}>0$ such that, if $r_N\le r_{\mathrm{fp}}$, then with
probability at least $1-\delta$, every fixed point
$\widehat{\mathbf y}_N\in\Phi_{A\sigma}$ of
$\widehat{\mathcal T}_{\mathcal Z,h}$, whenever one exists, satisfies
\[
 \dist(\widehat{\mathbf y}_N,\mathcal M)
 \le C\{\sigma^2+h r_N\}.
\]
Consequently, for fixed $d,D,c,A$, $h=c\sigma$ and
\[
 N\ge C\sigma^{-(d+2)}
 \{D\log(1/\sigma)+\log(1/\delta)\}
\]
suffice for the right-hand side to be $O(\sigma^2)$.
\end{corollary}

\begin{theorem}[Finite-sample orbit localization]
\label{thm:sample-orbit-localization}
For the same fixed $0<A_0<A$, under the conditions of
Theorem~\ref{thm:contraction} and Lemma~\ref{lem:empirical-approximation}, there are
$\sigma_{\mathrm{orb}},r_{\mathrm{orb}},C>0$ such that, if $\sigma\le\sigma_{\mathrm{orb}}$ and
$r_N\le r_{\mathrm{orb}}$, then with probability at least $1-\delta$, every
empirical orbit
\[
 \widehat{\mathbf y}^{(t+1)}
 =\widehat{\mathcal T}_{\mathcal Z,h}(\widehat{\mathbf y}^{(t)}),
 \qquad \widehat{\mathbf y}^{(0)}\in\Phi_{A_0\sigma},
\]
remains in $\Phi_{A\sigma}$ and satisfies, for every $t\ge0$,
\[
 \dist(\widehat{\mathbf y}^{(t)},\mathcal M)
 \le
 \kappa^t\dist(\widehat{\mathbf y}^{(0)},\mathcal M)
 +\frac{C(\sigma^2+h r_N)}{1-\kappa}.
\]
In particular, after the geometric transient the entire empirical orbit
lies in an $O(\sigma^2+h r_N)$ neighbourhood of $\mathcal M$; if
$r_N=O(\sigma)$ and $h=c\sigma$, this neighbourhood has width
$O(\sigma^2)$.
\end{theorem}

Theorem~\ref{thm:sample-orbit-localization} is uniform over all initial
queries in $\Phi_{A_0\sigma}$, including data-dependent observations that
fall in this tube. Supplementary Section~S3.2 quantifies the corresponding
fraction of observed initializations. It also gives the theoretical
iteration complexity. In particular, if
\[
 L\ge
 \frac{\left[\log\left\{(1-\kappa)d_0/
       [C(\sigma^2+hr_N)]\right\}\right]_+}{|\log\kappa|},
 \qquad d_0=\dist(\widehat{\mathbf y}^{(0)},\mathcal M),
\]
the transient term is no larger than the statistical error radius. Thus
$L=O\{\log(1/\sigma)\}$ when $d_0=O(\sigma)$ and
$\sigma^2+hr_N=O(\sigma^2)$.

\section{Multi-scale STSP}
\label{sec:multiscale}

{
The radius $h$ controls a bias--variance trade-off: large neighbourhoods
stabilize local moments but increase curvature bias, whereas small
neighbourhoods can give unstable tangent estimates. MS-STSP combines a
fixed finite collection of radii. The numerical studies use three radii;
the theory below allows any fixed finite number of radii.

\subsection{Motivation, algorithm, and localization results}
Proofs of the quadratic bias expansion and the two MS-STSP localization
theorems are given in Supplementary Sections~S4.1, S4.4, and~S4.5;
Supplementary Sections~S4.2--S4.3 provide the supporting weight and
aggregation results.
For a sufficiently small fixed geometric constant $c_0\in(0,1)$ and
$0<h\le c_0r_{\mathrm{geo}}$, define the clean local-average map
\[
 f_h(x)=
 \frac{1}{\mathrm{Area}\{B_{\mathrm E}(x,h)\cap\mathcal M\}}
 \int_{B_{\mathrm E}(x,h)\cap\mathcal M}y\,dA(y),
 \qquad x\in\mathcal M.
\]

\begin{lemma}[Quadratic structure of local averaging bias]
\label{lem:local_avg_expansion}
Under conditions (A2)--(A4), uniformly in $x\in\mathcal M$,
\[
 f_h(x)=x+\tfrac12 h^2\mathbf B(x)+O(h^4).
\]
Here, for any orthonormal basis $\{e_1,\ldots,e_d\}$ of $T_x\mathcal M$,
\begin{equation}
\label{eq:Bx-def}
 \mathbf B(x)=\frac{1}{d+2}\sum_{k=1}^d
 \mathrm{II}_x(e_k,e_k)\in N_x\mathcal M,
\end{equation}
and the remainder is uniform in $x$.
\end{lemma}

For the same $x$, the vector $\mathbf B(x)$ is independent of $h$.
Consequently, the ordinary least-squares intercept weights for a regression
on $s_j=h_j^2$ satisfy
\[
 \sum_{j=1}^M w_j=1,
 \qquad
 \sum_{j=1}^M w_jh_j^2=0,
\]
and cancel the displayed $h^2$ term in this ideal clean local-average
model. The clean local-average expansion motivates extrapolation in the
squared neighbourhood radius. For noisy STSP branches, the results below
establish that stable aggregation preserves the single-scale localization order.

\begin{algorithm}[H]
\caption{Multi-Scale STSP (MS-STSP)}
\label{alg:multiscale}
\begin{algorithmic}[1]
\Require Original cloud $\mathcal Z=\{\mathbf z_i\}_{i=1}^N$, dimension $d$, radii
$\{h_j\}_{j=1}^M$ with $M\ge2$ and at least two distinct values, and
iteration horizon $L\ge1$.
\Ensure Multi-scale estimates $\{\widehat{\mathbf y}_{i,{\mathrm{MS}}}\}_{i=1}^N$.
\For{$j=1,\ldots,M$}
  \State Apply Algorithm~\ref{alg:stsp} with reference $\mathcal Z$,
  initial queries $\widehat{\mathbf y}_{i,j}^{(0)}=\mathbf z_i$,
  dimension $d$, radius $h_j$, and horizon $L$, obtaining
  $\{\widehat{\mathbf y}_{i,j}^{(L)}\}_{i=1}^N$.
\EndFor
\State Compute the OLS intercept weights $w_1,\ldots,w_M$ for the
regressor $h_j^2$.
\For{$i=1,\ldots,N$}
  \State $\widehat{\mathbf y}_{i,{\mathrm{MS}}}\leftarrow
  \sum_{j=1}^M w_j\widehat{\mathbf y}_{i,j}^{(L)}$.
\EndFor
\State \Return $\{\widehat{\mathbf y}_{i,{\mathrm{MS}}}\}_{i=1}^N$.
\end{algorithmic}
\end{algorithm}

The following theorems establish population and finite-sample localization
for the fixed-reference multi-scale branches in
Algorithm~\ref{alg:multiscale}. The finite-sample result permits
scale-specific horizons $L_j$; Algorithm~\ref{alg:multiscale} is the common
case $L_j=L$.

\begin{theorem}[Population MS-STSP localization]
\label{thm:ms-population}
Let $M$ be fixed and $h_j=c_j\sigma$. Suppose every $(c_j,A)$ is
spectrally admissible, let the common initial point satisfy
$\mathbf y^{(0)}\in\cap_{j=1}^M\mathcal U_{h_j}$, and suppose the weights
satisfy
\[
 \sum_jw_j=1,\qquad \sum_j|w_j|\le C_w,
\]
where $C_w$ is a fixed design bound. Let
$\mathbf y_j^{(\infty)}=\lim_{t\to\infty}\mathcal T_{h_j}^{(t)}(\mathbf y^{(0)})$ and
$\mathbf y_{\mathrm{MS}}^{(\infty)}=\sum_jw_j\mathbf y_j^{(\infty)}$. Then all projected points
$\Pi_{\mathcal M}(\mathbf y_j^{(\infty)})$ lie in one $O(\sigma)$ geodesic patch and
\[
 \dist(\mathbf y_{\mathrm{MS}}^{(\infty)},\mathcal M)\le C C_w\sigma^2.
\]
Thus population MS-STSP retains the $O(\sigma^2)$ localization order of
the single-scale branches.
\end{theorem}

\begin{theorem}[Finite-sample MS-STSP localization]
\label{thm:ms-sample}
Let $0<\delta<1$. For $j=1,\ldots,M$, define
\[
 r_{N,j}=\sqrt{\frac{D\log(1/h_j)+\log(M/\delta)}{Nh_j^d}}
 +\frac{D\log(1/h_j)+\log(M/\delta)}{Nh_j^d}.
\]
Assume the conditions of Theorem~\ref{thm:ms-population} and the
initialization, sampling-error, and finite-horizon conditions of
Supplementary Theorem~S5 in Section~S4.5. These conditions keep each branch in its
attraction tube, reduce its transient error to the localization scale,
and bound accumulated tangential drift up to its prescribed horizon $L_j$.
For $t=0,\ldots,L_j$, define the empirical branch orbit by
$\widehat{\mathbf y}_j^{(t)}=
\widehat{\mathcal T}_{\mathcal Z,h_j}^{(t)}(\mathbf y^{(0)})$, and set
$\widehat{\mathbf y}_{\mathrm{MS}}=\sum_jw_j\widehat{\mathbf y}_j^{(L_j)}$.
Then, with probability at least $1-\delta$,
\[
 \dist(\widehat{\mathbf y}_{\mathrm{MS}},\mathcal M)
 \le C C_w
 \left\{\sigma^2+\max_j h_jr_{N,j}\right\}.
\]
In particular, if $h_j=c_j\sigma$ and
\[
 N\ge C\sigma^{-(d+2)}
 \{D\log(1/\sigma)+\log(M/\delta)\},
\]
the right-hand side is $O(\sigma^2)$.
\end{theorem}

}

\section{Numerical Studies}\label{Sect5}

{
\raggedbottom
\subsection{Experimental setup}
\label{sec:metrics}
We evaluate STSP and MS-STSP on synthetic manifolds and mouse dentate-gyrus
single-cell data. The synthetic study uses four curves (circle, tennis curve,
trefoil knot, and sixfolds) and two surfaces (sphere and torus), shown in
Supplementary Figure~S1, with ten independently generated noisy clouds per
setting. Let $\mathcal Y=\{\mathbf y_i\}_{i=1}^N$ denote the noisy input and
$\widehat{\mathcal Y}=\{\widehat{\mathbf y}_i\}_{i=1}^N$ a fitted output;
data-generation and implementation details are given in Supplementary
Section~S5. Let
$\mathcal X_{\mathrm{eval}}=\{\mathbf x_{{\mathrm{eval}},\ell}\}_{\ell=1}^{N_{\mathrm{eval}}}
\subset\mathcal M$ be an independent dense evaluation cloud, fixed across
repetitions within a setting. Within each repetition, all methods use the
same $\mathcal Y$ and $\mathcal X_{\mathrm{eval}}$.
\par
\begin{samepage}
\noindent We use the following two metrics to evaluate
$\widehat{\mathcal Y}$.
\[
 e_{\mathrm{fit}}=\frac1N\sum_{i=1}^N
 \dist(\widehat{\mathbf y}_i,\mathcal X_{\mathrm{eval}}),
 \qquad
 e_{\mathrm{cov}}=\frac1{N_{\mathrm{eval}}}\sum_{\ell=1}^{N_{\mathrm{eval}}}
 \dist(\mathbf x_{{\mathrm{eval}},\ell},\widehat{\mathcal Y}).
\]
\end{samepage}
Both metrics are means over all points in their respective clouds.
The first approximates the distance from the fitted cloud to the underlying
manifold using the dense evaluation cloud, whereas the second measures how
well the fitted cloud covers that evaluation cloud. Smaller values are
better.

We compare MS-STSP, STSP, YY-adapted, MF25, local PCA \citep{local_pca}, and
mean shift \citep{mean_shift}, and two SCMS variants \citep{ozertem2011locally,zhang2023linear}.
SCMS-KDE uses the Hessian of a Gaussian kernel density estimate to define
its projection subspace; SCMS-log-KDE uses the Hessian of the log density.
Both project the Gaussian mean-shift vector onto the selected normal
subspace. We refer to the estimator of \citet{yao2023manifold} as the
Yao--Yau (YY) estimator, following the terminology of \citet{yaosu2026}.
YY-adapted denotes our dimension-aware adaptation of its local
mean-displacement construction, using a $d$-dimensional tangent estimate
to construct its averaging neighbourhood. MF25 numerically approximates the residual
zero set of the method described in Section~\ref{sec:implicit-explicit}.
{The curve experiments use $N=5000$ and $\sigma=0.08$ for all four curves,
with an additional circle experiment at
$\sigma\in\{0.1,0.2,0.3\}$. The surface experiments use $N=20000$ and
$\sigma=0.08$ for both surfaces.} Noise is isotropic Gaussian in $\R^3$.
For each method with tunable parameters, the synthetic experiments select
the candidate minimizing
$J_{\mathrm{mean}}=\max(e_{\mathrm{fit}},e_{\mathrm{cov}})$.
The implementation code is available at
\url{https://github.com/zhigang-yao/STSP}.

{\subsection{Fitting and coverage errors on synthetic manifolds}}
\label{sec:one-d-exp}
We first examine how the competing methods balance fitting accuracy and
geometric coverage across the six synthetic manifolds; Table~\ref{tab:comparison-main}
summarises the corresponding errors, while Figure~\ref{fig:two-d-fits}
illustrates the fitted clouds on the torus.

\begin{table}[!htbp]
\centering
\caption{Fitting and coverage errors over ten independent realisations
on six synthetic manifolds (mean $\pm$ standard deviation, in units of
$10^{-3}$). For each
manifold and error measure, the smallest mean is bold and the second
smallest is underlined; rankings use unrounded means.}
\label{tab:comparison-main}
\footnotesize
\setlength{\tabcolsep}{2pt}
\begin{tabular}{llcccccc}
\hline
Method & Error & Circle & Tennis curve & Trefoil knot & Sixfolds & Sphere & Torus \\
\hline
\multirow{2}{*}{Local PCA} & $e_{\mathrm{fit}}$ & $18.90\pm0.60$ & $22.56\pm0.87$ & $15.39\pm0.63$ & $33.11\pm0.84$ & $24.67\pm0.30$ & $32.55\pm0.37$ \\
 & $e_{\mathrm{cov}}$ & $\underline{8.07}\pm0.45$ & $\underline{11.19}\pm0.57$ & $\underline{12.76}\pm0.44$ & $\underline{14.31}\pm0.28$ & $20.41\pm0.19$ & $\underline{21.56}\pm0.20$ \\
\hline
\multirow{2}{*}{Mean shift} & $e_{\mathrm{fit}}$ & $15.05\pm3.99$ & $23.07\pm1.40$ & $24.51\pm6.75$ & $25.61\pm3.44$ & $21.72\pm0.57$ & $\underline{23.31}\pm0.76$ \\
 & $e_{\mathrm{cov}}$ & $24.66\pm2.24$ & $29.54\pm1.74$ & $39.42\pm2.22$ & $31.26\pm2.27$ & $42.13\pm1.07$ & $43.19\pm0.67$ \\
\hline
\multirow{2}{*}{SCMS-KDE} & $e_{\mathrm{fit}}$ & $25.56\pm1.00$ & $30.94\pm1.01$ & $15.91\pm0.96$ & $47.04\pm1.94$ & $29.32\pm0.56$ & $40.03\pm0.75$ \\
 & $e_{\mathrm{cov}}$ & $17.27\pm0.78$ & $22.81\pm1.21$ & $14.87\pm0.71$ & $24.19\pm1.44$ & $22.06\pm0.36$ & $24.86\pm0.38$ \\
\hline
\multirow{2}{*}{SCMS-log-KDE} & $e_{\mathrm{fit}}$ & $10.68\pm0.71$ & $14.11\pm1.21$ & $12.86\pm0.78$ & $22.24\pm1.43$ & $18.33\pm0.46$ & $27.66\pm0.55$ \\
 & $e_{\mathrm{cov}}$ & $10.66\pm0.75$ & $14.19\pm1.22$ & $13.58\pm0.70$ & $20.44\pm1.54$ & $19.71\pm0.35$ & $28.13\pm0.53$ \\
\hline
\multirow{2}{*}{YY-adapted} & $e_{\mathrm{fit}}$ & $\underline{8.91}\pm0.83$ & $\underline{11.44}\pm0.80$ & $15.32\pm0.58$ & $\underline{18.41}\pm0.86$ & $\underline{17.23}\pm0.33$ & $25.86\pm0.47$ \\
 & $e_{\mathrm{cov}}$ & $8.95\pm0.71$ & $11.75\pm0.84$ & $17.33\pm0.58$ & $18.56\pm0.85$ & $23.51\pm0.23$ & $29.85\pm0.46$ \\
\hline
\multirow{2}{*}{MF25} & $e_{\mathrm{fit}}$ & $10.15\pm0.81$ & $13.37\pm1.01$ & $\underline{12.54}\pm0.77$ & $21.50\pm1.45$ & $17.62\pm0.44$ & $30.90\pm0.98$ \\
 & $e_{\mathrm{cov}}$ & $10.12\pm0.78$ & $13.33\pm1.07$ & $13.26\pm0.69$ & $20.76\pm1.48$ & $19.22\pm0.36$ & $23.97\pm0.50$ \\
\hline
\multirow{2}{*}{STSP} & $e_{\mathrm{fit}}$ & $10.27\pm0.80$ & $13.31\pm1.00$ & $12.81\pm0.69$ & $20.48\pm1.13$ & $18.57\pm0.38$ & $24.01\pm0.42$ \\
 & $e_{\mathrm{cov}}$ & $9.85\pm0.81$ & $12.77\pm1.00$ & $13.27\pm0.61$ & $17.53\pm1.53$ & $\underline{19.16}\pm0.25$ & $23.77\pm0.37$ \\
\hline
\multirow{2}{*}{MS-STSP} & $e_{\mathrm{fit}}$ & $\mathbf{5.98}\pm0.69$ & $\mathbf{8.92}\pm0.83$ & $\mathbf{11.56}\pm0.51$ & $\mathbf{15.04}\pm0.40$ & $\mathbf{11.49}\pm0.20$ & $\mathbf{16.67}\pm0.20$ \\
 & $e_{\mathrm{cov}}$ & $\mathbf{5.83}\pm0.67$ & $\mathbf{8.53}\pm0.81$ & $\mathbf{12.09}\pm0.51$ & $\mathbf{12.49}\pm0.43$ & $\mathbf{13.65}\pm0.19$ & $\mathbf{17.33}\pm0.16$ \\
\hline
\end{tabular}
\end{table}

The results reveal a consistent trade-off between fitting and coverage
errors. Local PCA provides particularly small coverage
errors, ranking second on five of the six manifolds, but has substantially
larger fitting errors. Mean shift sometimes fits more closely, including the
second-smallest torus fitting error, but has the largest coverage error on
every manifold. SCMS-log-KDE improves markedly on SCMS-KDE and gives a more
balanced fit, while YY-adapted has the second-smallest fitting error on
four manifolds at the cost of larger coverage errors, especially on the
trefoil knot, sphere, and torus. MF25 is also competitive, attaining the
second-smallest trefoil fitting error and coverage errors on the sphere and
torus close to those of STSP. This broadly similar performance is consistent
with MF25 and STSP targeting related manifold-fitting objects, both defined
as zero sets of local normal residuals, although their constructions and
numerical procedures differ. Across all six manifolds, STSP reduces fitting error by
17--46\% relative to local PCA and coverage error by 44--66\% relative to
mean shift, thereby performing well on both criteria. MS-STSP has the
smallest fitting and coverage errors in every setting, reducing the
corresponding STSP errors by 9--42\%. For example, on the torus its fitting
and coverage errors are 0.01667 and 0.01733, compared with 0.02401 and
0.02377 for STSP.

\begin{figure}[!htbp]
  \centering
  \includegraphics[width=.9\linewidth]{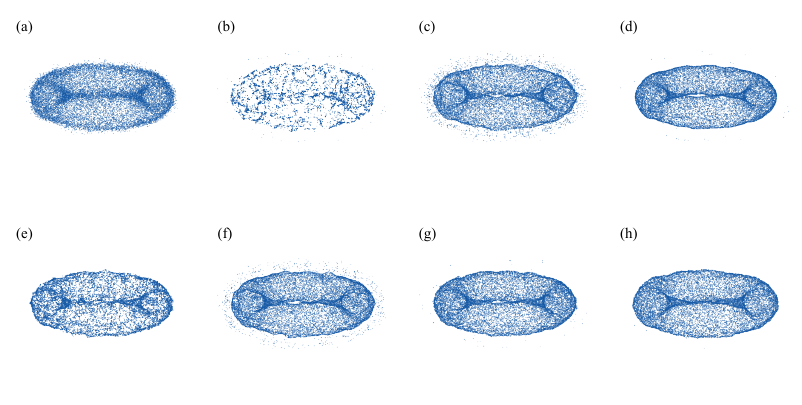}
  \caption{Torus fits for the first of ten independent realisations.
  Panels (a)--(h) show local PCA, mean shift, SCMS-KDE, SCMS-log-KDE,
  YY-adapted, MF25, STSP, and MS-STSP, respectively.}
  \label{fig:two-d-fits}
\end{figure}

Figure~\ref{fig:two-d-fits} uses the torus as a representative example.
The local-PCA fit in panel (a) covers the surface but has a visibly rough,
spiked boundary. The mean-shift and YY-adapted fits in panels (b) and (e)
contain clear gaps, consistent with their larger coverage errors. SCMS-KDE
and MF25 in panels (c) and (f) recover smoother boundaries but retain
scattered points away from the surface. SCMS-log-KDE and STSP in panels
(d) and (g) give coherent, well-covered surfaces with only a small number
of such points. MS-STSP in panel (h) gives the cleanest reconstruction,
with a smooth boundary, few visible outliers, and even coverage. These
comparisons are consistent with the intended effects of successive normal
correction and its multi-scale extension. Corresponding fitted-cloud
comparisons for the other manifolds are provided in Supplementary
Figures~S4--S9.

\subsection{Multi-scale fitting under high noise}
\label{sec:high-noise}
To assess robustness beyond the baseline noise level, Table~\ref{tab:high-noise-main}
tracks the two errors on the circle as $\sigma$ increases. Both errors increase
for every method, but at markedly different rates. Local PCA continues to
favour coverage, attaining the smallest coverage error at $\sigma=0.3$ while
its fitting error becomes large; mean shift shows the converse pattern, with
the second-smallest fitting error at $\sigma=0.2$ and $0.3$ but substantially
larger coverage errors. SCMS-log-KDE remains more stable than SCMS-KDE, while
YY-adapted preserves coverage better than fitting and MF25 deteriorates
sharply at the highest noise level. STSP remains competitive on fitting but
loses coverage as noise increases. MS-STSP provides the most consistent
balance: it has the smallest fitting error at all three noise levels and the
smallest or second-smallest coverage error. At $\sigma=0.3$, its fitting and
coverage errors are 0.03323 and 0.02706, reductions of approximately 57\% and
60\% from the corresponding STSP errors of 0.07739 and 0.06815.

\begin{table}[!htbp]
\centering
\caption{Fitting and coverage errors over ten independent realisations
of the high-noise circle experiment (mean $\pm$ standard deviation, in units
of $10^{-3}$).
Within each noise level and for each error measure separately, the smallest
mean is marked in bold and the second smallest is underlined. Rankings use
unrounded means.}
\label{tab:high-noise-main}

\footnotesize
\setlength{\tabcolsep}{3pt}
\begin{tabular}{llccc}
\hline
Method & Error & $\sigma=0.1$ & $\sigma=0.2$ & $\sigma=0.3$ \\
\hline
\multirow{2}{*}{Local PCA} & $e_{\mathrm{fit}}$ & $26.27\pm0.76$ & $78.86\pm3.16$ & $150.48\pm5.16$ \\
 & $e_{\mathrm{cov}}$ & $\underline{10.27}\pm0.46$ & $\underline{18.74}\pm0.82$ & $\mathbf{26.51}\pm1.28$ \\
\hline
\multirow{2}{*}{Mean shift} & $e_{\mathrm{fit}}$ & $16.94\pm6.26$ & $\underline{37.22}\pm5.81$ & $\underline{72.96}\pm9.49$ \\
 & $e_{\mathrm{cov}}$ & $26.92\pm2.38$ & $46.24\pm2.75$ & $77.62\pm13.68$ \\
\hline
\multirow{2}{*}{SCMS-KDE} & $e_{\mathrm{fit}}$ & $38.13\pm1.08$ & $125.19\pm5.35$ & $230.17\pm10.26$ \\
 & $e_{\mathrm{cov}}$ & $26.13\pm1.10$ & $63.66\pm7.25$ & $66.47\pm12.18$ \\
\hline
\multirow{2}{*}{SCMS-log-KDE} & $e_{\mathrm{fit}}$ & $14.22\pm1.36$ & $43.25\pm3.32$ & $91.09\pm12.33$ \\
 & $e_{\mathrm{cov}}$ & $13.86\pm1.39$ & $39.34\pm3.90$ & $83.21\pm11.31$ \\
\hline
\multirow{2}{*}{YY-adapted} & $e_{\mathrm{fit}}$ & $\underline{11.05}\pm1.39$ & $37.52\pm4.44$ & $96.51\pm6.03$ \\
 & $e_{\mathrm{cov}}$ & $10.84\pm1.20$ & $26.51\pm2.36$ & $30.92\pm1.39$ \\
\hline
\multirow{2}{*}{MF25} & $e_{\mathrm{fit}}$ & $13.51\pm1.31$ & $43.50\pm3.80$ & $154.62\pm13.85$ \\
 & $e_{\mathrm{cov}}$ & $13.27\pm1.34$ & $38.69\pm3.65$ & $78.32\pm11.98$ \\
\hline
\multirow{2}{*}{STSP} & $e_{\mathrm{fit}}$ & $13.34\pm1.00$ & $38.04\pm3.18$ & $77.39\pm12.10$ \\
 & $e_{\mathrm{cov}}$ & $12.55\pm1.22$ & $33.05\pm3.65$ & $68.15\pm10.72$ \\
\hline
\multirow{2}{*}{MS-STSP} & $e_{\mathrm{fit}}$ & $\mathbf{11.03}\pm1.27$ & $\mathbf{17.20}\pm2.35$ & $\mathbf{33.23}\pm3.33$ \\
 & $e_{\mathrm{cov}}$ & $\mathbf{9.36}\pm1.15$ & $\mathbf{14.58}\pm2.19$ & $\underline{27.06}\pm3.83$ \\
\hline
\end{tabular}

\end{table}

\begin{figure}[H]
  \centering
  \includegraphics[width=.9\linewidth]{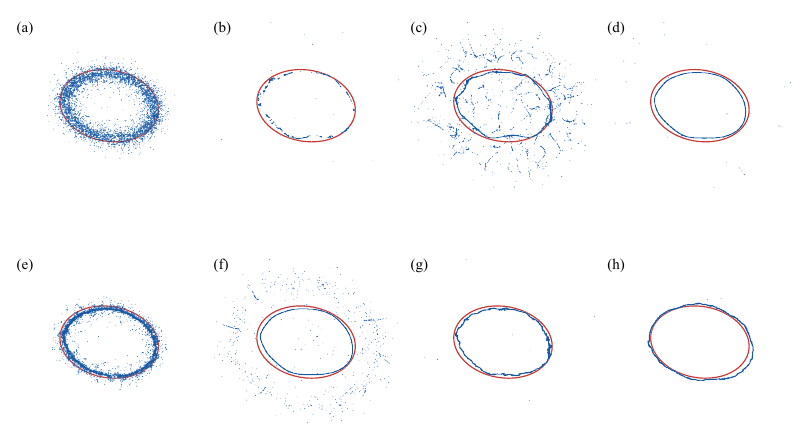}
  \caption{Circle fitting for the first of ten independent realisations
  at $\sigma=0.3$. Panels (a)--(h) show local PCA, mean shift, SCMS-KDE,
  SCMS-log-KDE, YY-adapted, MF25, STSP, and MS-STSP, respectively.
  Blue points are fitted observations and the thin red curve is the
  underlying circle.}
  \label{fig:circle-high-noise}
\end{figure}

Figure~\ref{fig:circle-high-noise} shows the corresponding fits at
$\sigma=0.3$; fits at the lower noise levels are provided in Supplementary
Section~S5. Local PCA
and YY-adapted in panels (a) and (e) form rough, diffuse bands; the
mean-shift fit in panel (b) is sparse and uneven along the circle; and
SCMS-KDE and MF25 in panels (c) and (f) retain many points away from the
underlying curve. More importantly, some apparently smooth fits lie
systematically inside the red reference circle, most visibly for
SCMS-log-KDE and MF25 in panels (d) and (f), illustrating curvature bias
under high noise. STSP in panel (g) removes most outlying points but retains
some irregularity and inward displacement. MS-STSP in panel (h) follows the
reference circle more closely while maintaining nearly complete coverage,
consistent with multi-scale extrapolation reducing curvature bias without
reintroducing substantial fitting variability.

\subsection{Empirical convergence and radius sensitivity}
We next examine agreement between the empirical behaviour of STSP and the
theory in three respects: the scaling of fitting error with noise, geometric
convergence of the iterates to an error floor, and sensitivity to the
neighbourhood radius.

Figure~\ref{fig:rate-vs-sigma} gives a hard-radius finite-sample diagnostic
for the circle and sphere over a wide noise range. The dashed guide has the
theoretical $O(\sigma^2)$ slope; its vertical position is chosen for visual
comparison with the observed errors, rather than calculated from a theoretical
constant.

\begin{figure}[!htbp]
\centering
\includegraphics[width=.78\linewidth]{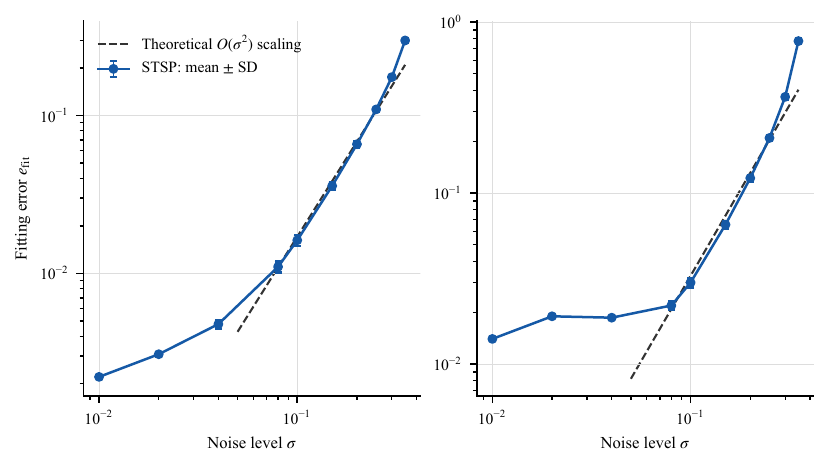}
\vspace{-5pt}
{\makebox[\linewidth]{\makebox[.5\linewidth]{(a)}\makebox[.5\linewidth]{(b)}}}
\caption{STSP fitting error versus noise on log--log axes for
(a) the circle and (b) the sphere. Points show means over ten independent repetitions;
error bars show one standard deviation on either side of the mean.
Dashed lines indicate $O(\sigma^2)$ scaling.}
\label{fig:rate-vs-sigma}
\end{figure}

Figure~\ref{fig:iteration-convergence} shows the first few iterations for
the circle and sphere. The finite-sample bound in
Theorem~\ref{thm:sample-orbit-localization} has the form
$\dist(\widehat{\mathbf y}^{(t)},\mathcal M)
\le\kappa^t\dist(\widehat{\mathbf y}^{(0)},\mathcal M)
+C(\sigma^2+hr_N)/(1-\kappa)$: a geometrically decreasing term followed
by a finite-sample error floor. The observed fitting errors
fall rapidly over the first few iterations and then level off, qualitatively
consistent with this form.

\begin{figure}[!htbp]
\centering
\includegraphics[width=.78\linewidth]{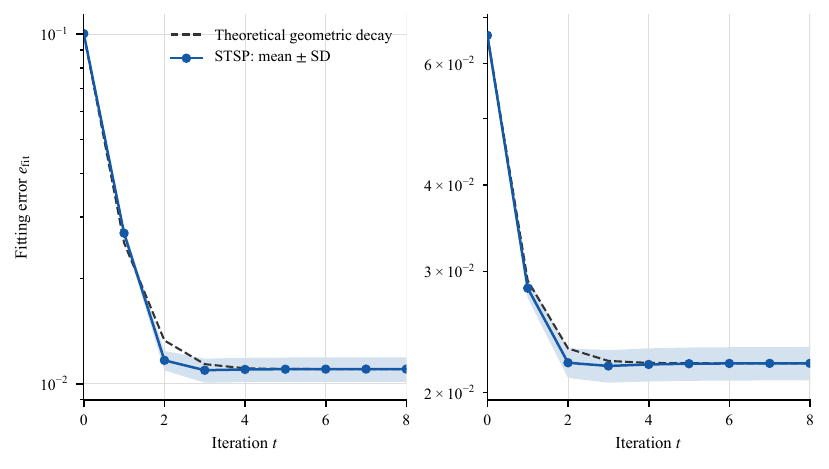}
\vspace{-5pt}
{\makebox[\linewidth]{\makebox[.5\linewidth]{(a)}\makebox[.5\linewidth]{(b)}}}
\caption{STSP fitting error versus iteration for (a) the circle
and (b) the sphere. Solid curves show means over ten independent repetitions;
shaded bands show one standard deviation on either side of the mean.
Dashed curves schematically illustrate the geometric-decay-plus-floor form
of Theorem~\ref{thm:sample-orbit-localization}.}
\label{fig:iteration-convergence}
\end{figure}

Figure~\ref{fig:h-sensitivity} examines sensitivity to the neighbourhood
radius and exhibits the expected bias--variance pattern. Too-small radii
give unstable local moments, whereas too-large radii introduce curvature
bias. The minimizing ratio changes with the noise level.

\begin{figure}[!htbp]
\centering
\includegraphics[width=.62\linewidth]{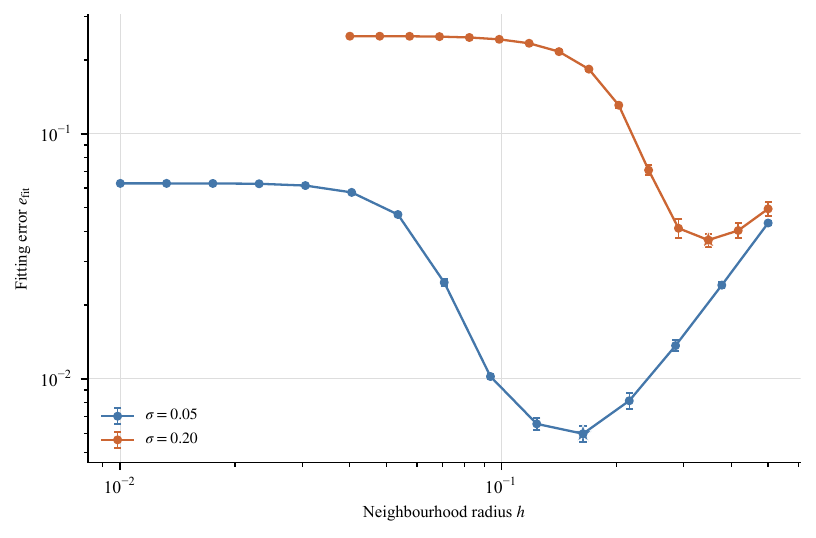}
\caption{STSP fitting error versus neighbourhood radius on a
circle at two noise levels. Curves show means over ten independent
repetitions; error bars show one standard deviation on either side of the mean.
Stars mark minima of the mean errors.}
\label{fig:h-sensitivity}
\end{figure}

\flushbottom
}

\subsection{Dentate gyrus single-cell data}
\label{sec:real-clustering}
\raggedbottom

{We analyse single-cell RNA-sequencing data from the mouse dentate gyrus
\citep{hochgerner2018}.
The data are available at
\url{https://raw.githubusercontent.com/theislab/scvelo_notebooks/master/data/DentateGyrus/10X43_1.h5ad}, which contain 2,930 cells assigned to 14
populations. Five annotations describe granule-cell neurogenesis: radial
glia-like cells, neuronal intermediate progenitor cells (nIPC),
neuroblasts, immature granule cells, and mature granule cells. The remaining
annotations distinguish astrocytes,
oligodendrocyte precursor cells (OPC), oligodendrocytes (OL), microglia,
endothelial cells, and the Cajal--Retzius, Cck--Tox, GABAergic, and mossy
neuronal populations.}

{Counts are library-normalised, log-transformed, and restricted to the
2,000 most variable genes before coordinatewise standardisation and PCA.
All methods use the same unwhitened representation at a given dimension.
The scores are divided by their root-mean-square cloud radius. For each
method, candidate parameters are selected by the average Euclidean
silhouette score (SIL) of the 14 supplied annotations. The principal analysis uses
$p=30$ and $d=1$; candidate sets and iteration settings are given in
Supplementary Section~S5.}

\begin{table}[H]
\centering
\caption{{Results for the dentate gyrus data using 30 principal
components. SIL uses the 14 supplied annotations. Accuracy uses optimal
one-to-one matching of cluster and annotation labels. Clustering uses 14
clusters; entries are means over 10 shared $k$-means seeds, with five
$k$-means replicates per
seed. Larger values are better.}}
\label{tab:dentategyrus-clustering}
\begin{tabular}{lrrrr}
\hline
Method & Average SIL & ARI & Accuracy & NMI \\
\hline
Input & 0.1615 & 0.4439 & 0.5747 & 0.6848 \\
Mean shift & 0.2504 & 0.4575 & 0.6514 & 0.7009 \\
SCMS-KDE & 0.3086 & 0.4701 & 0.6282 & 0.6332 \\
SCMS-log-KDE & 0.3841 & 0.4621 & 0.6174 & 0.7100 \\
Local PCA & 0.3643 & 0.4628 & 0.6230 & 0.7058 \\
YY-adapted & 0.3742 & 0.5153 & 0.6804 & 0.7227 \\
MF25 & 0.3466 & 0.5061 & 0.6788 & 0.7302 \\
STSP & 0.4224 & 0.5048 & 0.6659 & 0.7275 \\
MS-STSP & \textbf{0.4300} & \textbf{0.5474} & \textbf{0.6905} & \textbf{0.7338} \\
\hline
\end{tabular}

\end{table}

{Table~\ref{tab:dentategyrus-clustering} shows that STSP raises average
SIL from 0.1615 in the input and 0.2504 after mean shift to 0.4224.
Its ARI, accuracy, and NMI also exceed those of the input, mean shift,
local PCA, and both SCMS variants. Three-scale MS-STSP achieves the
largest values on all four measures: SIL 0.4300, ARI 0.5474, accuracy
0.6905, and NMI 0.7338.}

{Figure~\ref{fig:dentategyrus-stsp-ms-tsne} shows t-SNE plots computed
separately from the STSP and MS-STSP fitted representations, coloured by
the supplied cell annotations. In each plot,
neuroblasts, immature granule cells, and mature granule cells occupy
neighbouring regions, while the smaller radial glia-like and nIPC
populations remain visible. Several glial, vascular, and other neuronal
populations form separate labelled regions; OPC and OL lie near one another.}

\begin{figure}[!htbp]
\centering
\includegraphics[width=.84\linewidth]{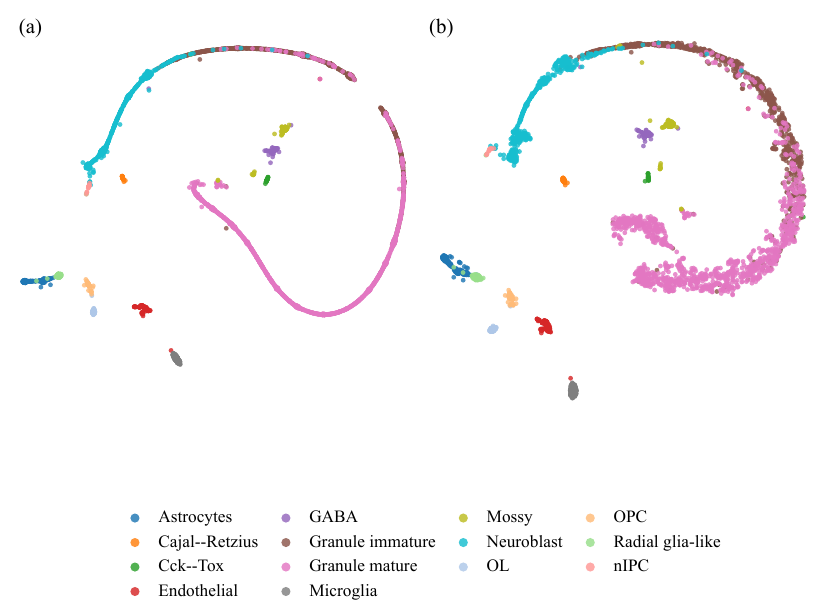}
\caption{Dentate gyrus cell populations in t-SNE plots of (a) the
STSP-fitted and (b) the three-scale MS-STSP-fitted representations using
30 principal components. Colours indicate the 14 supplied annotations.}
\label{fig:dentategyrus-stsp-ms-tsne}
\end{figure}

{Figure~\ref{fig:dentategyrus-pca} compares average SIL across
$p=20,30,40,50$. At each dimension, parameters are selected separately
for each method using the same full-cell annotation-based SIL criterion;
the plot therefore shows selected performance at each $p$, not the
behaviour of one fixed configuration across dimensions. The input SIL
decreases as more components are retained.
All fitted representations improve upon the input at each displayed
dimension. MS-STSP has the highest SIL throughout. STSP remains above
0.40 at $p=20,30,40$ and is second only to MS-STSP at all four
dimensions; local PCA and SCMS-log-KDE also retain substantially higher
SIL than the input as $p$ increases.}

\begin{figure}[H]
\centering
\includegraphics[width=.68\linewidth]{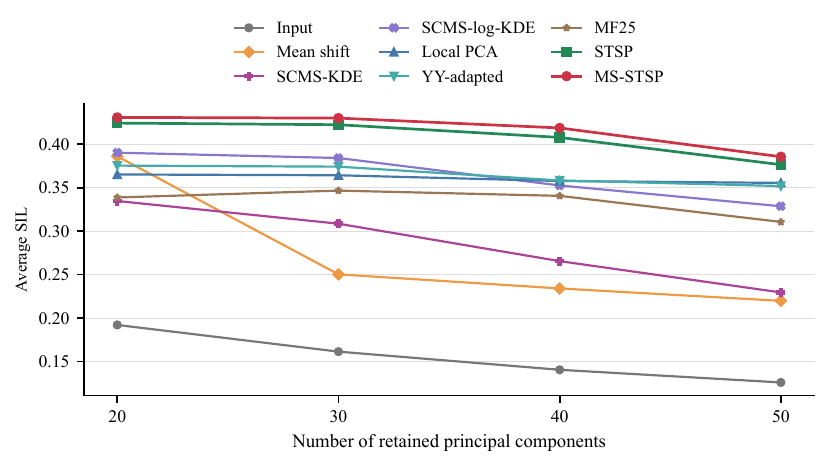}
\caption{Average annotation-based silhouette versus the number of retained
principal components for the dentate gyrus data.}
\label{fig:dentategyrus-pca}
\end{figure}

\section{Conclusion}\label{sec:conclusion}
STSP connects successive local tangent-space corrections to a fitted
manifold through the geometry of its nearby population fixed points. Under
the stated small-noise conditions, this fixed-point set is a $C^1$ normal
graph over the underlying manifold, diffeomorphic to it and within
$O(\sigma^2)$ in Hausdorff distance. Population orbits initialised
sufficiently close to this set converge geometrically to individual fixed
points. With high probability, fixed-reference empirical orbits initialised
within the local tube enter and remain in an $O(\sigma^2+hr_N)$ neighbourhood
of the underlying manifold after a geometric transient.

The experiments assess fitting accuracy and geometric coverage jointly;
at higher noise levels, multi-scale aggregation reduces observed
curvature-induced shrinkage. The theory establishes that stable aggregation
preserves the single-scale localisation order. Further work could establish
empirical-orbit convergence and joint guarantees for coverage and
data-driven scale selection.

\section*{Conflict of interest}
The authors declare no competing interests.

\section*{Acknowledgments}
\begingroup Zhigang Yao has received support from the Singapore Ministry of Education
through Tier~1 grants A-8004146-00-00 and A-8002931-00-00 at the National
University of Singapore. Yuqing Xia was supported by the National Natural
Science Foundation of China (grant No.~12401500). The authors are grateful for the resources and
facilities provided by the Shanghai Institute for Mathematics and
Interdisciplinary Sciences (SIMIS), which were essential for the completion
of this work.\endgroup

\section*{Supplementary material}
The Supplementary Material contains all proofs, supporting theoretical
results, manifold definitions, and detailed parameter settings for the
numerical experiments.

\clearpage
\raggedbottom
\setcounter{topnumber}{3}
\setcounter{totalnumber}{4}
\renewcommand{\topfraction}{0.9}
\renewcommand{\textfraction}{0.05}
\renewcommand{\floatpagefraction}{0.7}
\section*{Supplementary Material: Proofs and Numerical Details}
\setcounter{section}{0}
\setcounter{subsection}{0}
\setcounter{equation}{0}
\setcounter{table}{0}
\setcounter{figure}{0}
\setcounter{theorem}{0}
\setcounter{lemma}{0}
\setcounter{corollary}{0}
\setcounter{proposition}{0}
\setcounter{definition}{0}
\setcounter{assumption}{0}
\setcounter{remark}{0}
\renewcommand{\thesection}{S\arabic{section}}
\renewcommand{\theequation}{S\arabic{equation}}
\renewcommand{\thetable}{S\arabic{table}}
\renewcommand{\thefigure}{S\arabic{figure}}
\renewcommand{\thetheorem}{S\arabic{theorem}}
\renewcommand{\thelemma}{S\arabic{lemma}}
\renewcommand{\thecorollary}{S\arabic{corollary}}
\renewcommand{\theproposition}{S\arabic{proposition}}
\renewcommand{\thedefinition}{S\arabic{definition}}
\renewcommand{\theassumption}{S\arabic{assumption}}
\renewcommand{\theremark}{S\arabic{remark}}
\def\theHsection{appendix.\arabic{section}}
\def\theHsubsection{appendix.\arabic{section}.\arabic{subsection}}
\def\theHequation{appendix.\arabic{equation}}
\def\theHtable{appendix.\arabic{table}}
\def\theHfigure{appendix.\arabic{figure}}
\def\theHtheorem{appendix.\arabic{theorem}}
\def\theHlemma{appendix.\arabic{lemma}}
\def\theHcorollary{appendix.\arabic{corollary}}
\def\theHproposition{appendix.\arabic{proposition}}
\def\theHdefinition{appendix.\arabic{definition}}
\def\theHassumption{appendix.\arabic{assumption}}
\def\theHremark{appendix.\arabic{remark}}
\def\theHalgorithm{appendix.\arabic{algorithm}}
\section{Notation and geometric preliminaries}
\label{app:sec:analysis}

The supplement is organized by the three levels of theory in the main
paper.  Section~S1 records notation and geometric estimates.  Section~S2
establishes the population fixed-point and orbit results for STSP, Section~S3
gives the finite-sample results, and Section~S4 develops the localization
theory for MS-STSP\@. Numerical definitions and implementation details are
collected in Section~S5.

Notation follows the main paper. In particular, $\mathrm{II}_x$ denotes
the second fundamental form at $x$, and $K_{\mathrm{geo}}$ is the geometric
derivative bound in Assumption~1(A2).

Fix $c>0$ and set $h=c\sigma$, with
$c\sigma\le r_{\mathrm{geo}}/2$. Population asymptotics are taken as
$\sigma\to0$; in Section~S3, $N\to\infty$ jointly with $\sigma\to0$.
The normalised curvature remainder in the flat-law expansion is
$O(c^2\sigma/r_{\mathrm{geo}})$, which tends to zero for fixed $c$.

We first record two standard consequences of bounded curvature and positive reach, used repeatedly below.

\begin{lemma}[Tangent approximation]
\label{app:lem:tangent-approx}
For $x\in\mathcal{M}$ and $z\in\mathcal{M}$ with $\|z-x\|\le h\le\tau/2$, write $z-x = \mathbf{P}_x(z-x) + \mathbf{n}(z)$ with $\mathbf{n}(z)\in N_x\mathcal{M}$. Then
\[
  \|\mathbf{n}(z)\| \le \frac{\|z-x\|^2}{2\tau}, \qquad
  {0\le \|z-x\|-\|\mathbf{P}_x(z-x)\| \le \frac{\|z-x\|^3}{4\tau^2}.}
\]
\end{lemma}

\begin{proof}
{By Proposition~6.3 of \citet{Niyogi2008} and the assumption
$\reach(\mathcal{M})\ge\tau$,
$\dist(z,x+T_x\mathcal M)\le \|z-x\|^2/(2\tau)$ for
$z\in\mathcal M$. Since
$\dist(z,x+T_x\mathcal M)=\|(I_D-\mathbf P_x)(z-x)\|
=\|\mathbf n(z)\|$, this gives the first inequality. For the second
inequality, set $r=\|z-x\|$,
$r_{\mathrm T}=\|\mathbf P_x(z-x)\|$, and
$r_\perp=\|\mathbf n(z)\|$. Because the tangent and normal components are
orthogonal, $r^2=r_{\mathrm T}^2+r_\perp^2$, so $r_{\mathrm T}\le r$.
If $r=0$ the claim is trivial. Otherwise,}
\[
{
  0\le r-r_{\mathrm T}
  =\frac{r^2-r_{\mathrm T}^2}{r+r_{\mathrm T}}
  =\frac{r_\perp^2}{r+r_{\mathrm T}}
  \le \frac{r_\perp^2}{r}
  \le \frac{r^3}{4\tau^2},
}
\]
{where the last step uses $r_\perp\le r^2/(2\tau)$ from the first inequality. This proves the stated bound.}
\end{proof}

\begin{lemma}[Volume comparison]
\label{app:lem:volume}
{There exist constants $c_0\in(0,1)$ and $C>0$ such that, for $x\in\mathcal{M}$ and $h\le c_0 r_{\mathrm{geo}}$,}
\[
{
  B_g\bigl(x,h(1-Ch^2)\bigr)
  \subset B_{\mathrm E}(x,h)\cap\mathcal{M}
  \subset B_g\bigl(x,h(1+Ch^2)\bigr).
}
\]
{Moreover,}
\[
{
  \mathrm{Area}\{B_{\mathrm E}(x,h)\cap\mathcal{M}\} = \omega_d h^d\bigl(1 + O(h^2)\bigr),
}
\]
{where $\omega_d$ is the volume of the unit $d$-ball and the $O(h^2)$ term is uniform over $\mathcal{M}$ and controlled by the curvature bound $K_{\mathrm{geo}}$ in (A2).}
\end{lemma}

\begin{proof}
{Because $h\le c_0r_{\mathrm{geo}}\le \iota_{\mathcal M}$, the exponential map at $x$ gives geodesic normal coordinates on the intrinsic balls considered below, while positive reach excludes distant intrinsic sheets from the corresponding Euclidean ball.  A vector $u\in T_x\mathcal{M}$ in the injectivity ball represents $\exp_x(u)$ and $d_g(x,\exp_x(u))=\|u\|$. The standard volume expansion in normal coordinates, valid under a uniform curvature bound \cite[Chapter~3]{doCarmo1992}, says that the area element has density $1+O(\|u\|^2)$ with respect to Lebesgue measure $du$ on $T_x\mathcal{M}$. Hence}
\[
{
  \mathrm{Area}\{B_g(x,r)\}
  = \int_{\|u\|\le r}\{1+O(\|u\|^2)\}\,du
  = \omega_d r^d\{1+O(r^2)\}.
}
\]
{We now compare the geodesic ball with the ambient Euclidean ball $B_{\mathrm E}(x,h)\cap\mathcal{M}$. In the same exponential coordinates, the embedded geodesic has the Taylor expansion $\exp_x(u)-x=u+\tfrac12\mathrm{II}_x(u,u)+O(\|u\|^3)$, where $\mathrm{II}_x$ is the second fundamental form and the quadratic term is normal to $T_x\mathcal{M}$. Hence $\|\exp_x(u)-x\|^2=\|u\|^2+O(\|u\|^4)$. Choose a constant $C_0>0$, independent of $x$ and $h$, such that}
\[
{
  \left|
  \frac{\|\exp_x(u)-x\|}{\|u\|}-1
  \right|
  \le C_0\|u\|^2
  \qquad (0<\|u\|\le r_{\mathrm{geo}}/2).
}
\]
{Equivalently, whenever $\|u\|\le r_{\mathrm{geo}}/2$,}
\[
{
  \|u\|(1-C_0\|u\|^2)
  \le \|\exp_x(u)-x\|
  \le \|u\|(1+C_0\|u\|^2).
}
\]
{Choose}
\[
{
  c_0 \le \min\left\{\frac{1}{4},\,\frac{1}{4r_{\mathrm{geo}}\sqrt{C_0}}\right\}.
}
\]
{Set $C=8C_0$. Then $h\le c_0r_{\mathrm{geo}}$ implies $h\le r_{\mathrm{geo}}/4$, $Ch^2\le1/2$, and $C_0(2h)^2\le1/4<1/2$. To prove the first inclusion, take $z\in B_g\bigl(x,h(1-Ch^2)\bigr)$ and write $z=\exp_x(u)$. Since $d_g(x,z)=\|u\|$, we have $\|u\|\le h(1-Ch^2)\le h\le r_{\mathrm{geo}}/4$, so the distance comparison above is applicable. Moreover,}
\[
{
  \|z-x\|
  = \|\exp_x(u)-x\|
  \le \|u\|(1+C_0\|u\|^2)
  \le h(1-Ch^2)(1+C_0h^2)
  \le h,
}
\]
{where the last inequality uses $C\ge C_0$. Hence $z\in B_{\mathrm E}(x,h)\cap\mathcal M$. Conversely, take $z\in B_{\mathrm E}(x,h)\cap\mathcal M$ and put $r=\|z-x\|$.  Proposition~6.3 of \citet{Niyogi2008}, applied with $\reach(\mathcal M)\ge\tau$ and $r\le h\le\tau/4$, gives the global comparison}
\[
{
 d_g(x,z)
 \le \tau\{1-\sqrt{1-2r/\tau}\}
 \le 2r\le2h.
}
\]
{Thus the reach condition excludes every intrinsically distant sheet from the Euclidean ball.  Since $2h\le r_{\mathrm{geo}}/2\le\iota_{\mathcal M}$, there is now a unique $u$ in the injectivity ball with $z=\exp_x(u)$ and $\|u\|=d_g(x,z)\le2h$.  The local lower bound is therefore applicable and gives}
\[
{
  \|u\|
  \le \frac{h}{1-C_0(2h)^2}
  \le h\{1+2C_0(2h)^2\}
  \le h(1+Ch^2).
}
\]
{Since $d_g(x,\exp_x(u))=\|u\|$, these two implications yield}
\[
{
  B_g\bigl(x,h(1-Ch^2)\bigr)
  \subset B_{\mathrm E}(x,h)\cap\mathcal{M}
  \subset B_g\bigl(x,h(1+Ch^2)\bigr).
}
\]
{Applying the preceding geodesic-ball volume expansion to the two bounding balls yields $\mathrm{Area}\{B_{\mathrm E}(x,h)\cap\mathcal{M}\}=\omega_d h^d\{1+O(h^2)\}$, with constants uniform over the compact manifold.}
\end{proof}

\section{Population theory for STSP}

\subsection{Flat reference model and uniform local expansions}
Recall from Section~2.2 of the main paper that $\mathbf Y$ is a generic
observation, $\boldsymbol\mu_h(\mathbf y)$ and
$\boldsymbol\Sigma_h(\mathbf y)$ are its conditional local moments,
$T_h(\mathbf y)$ is the corresponding leading-$d$ population local
tangent space, and $\mathbf P_h(\mathbf y)$ is its orthogonal projector.
The population update is
\begin{equation}
\label{app:eq:population-update}
  \mathcal{T}_h(\mathbf{y}) = \boldsymbol{\mu}_h(\mathbf{y}) + \mathbf{P}_h(\mathbf{y})\bigl[\mathbf{y} - \boldsymbol{\mu}_h(\mathbf{y})\bigr].
\end{equation}

The fixed-point identity used in the main paper follows immediately from
\eqref{app:eq:population-update}: wherever the leading-$d$ eigenspace is unique,
\[
 \mathcal T_h(\mathbf y)=\mathbf y
 \quad\Longleftrightarrow\quad
 (I_D-\mathbf P_h(\mathbf y))
 \{\mathbf y-\boldsymbol\mu_h(\mathbf y)\}=\mathbf0.
\]

Use the tube $\Phi_\rho$, projection point
$x=\Pi_{\mathcal M}(\mathbf y)$, and normal offset
$\boldsymbol\nu(\mathbf y)=\mathbf y-x\in N_x\mathcal M$ defined in the main paper. When
$h=c\sigma$, the required uniform conditional-moment expansion holds on
the shrinking tube $\Phi_{A\sigma}$, where $A<\infty$ is fixed.

Use $q$, $W=(W_{\mathrm T},W_{\mathrm N})$, $p_{c,a}$,
$\E_{c,a}$, $\cov_{c,a}$, $m_c$, and $S_c$ as defined in the main paper.
For reference, the density is
\begin{equation}
\label{app:eq:flat-law}
 p_{c,a}(w)=
 \frac{\mathbf 1\{\|w\|\le c\}
       \exp\{-\|a+w_{\mathrm N}\|^2/2\}}
      {\int_{\|v\|\le c}\exp\{-\|a+v_{\mathrm N}\|^2/2\}\,dv}.
\end{equation}
The corresponding normal map and full covariance are
\[
 m_c(a)=a+\E_{c,a}W_{\mathrm N},\qquad S_c(a)=\cov_{c,a}(W).
\]
The law in \eqref{app:eq:flat-law} is the exact scaled conditional law for an
affine $d$-plane.  {We call $(c,A)$ \emph{spectrally
admissible} when there is $g_0>0$ such that, for every $\|a\|\le A$,}
\begin{equation}
\label{app:eq:admissible}
 \lambda_{\min}\{\cov_{c,a}(W_{\mathrm T})\}
 -\lambda_{\max}\{\cov_{c,a}(W_{\mathrm N})\}\ge g_0.
\end{equation}
{This finite-dimensional integral condition is sufficient because the
required normal contraction follows from the next lemma.}

\begin{lemma}[Automatic flat-model normal contraction]
\label{app:lem:auto-normal-contraction}
For every fixed $c>0$ and $A<\infty$, the flat-model normal mean satisfies
\[
 \mathrm Dm_c(a)=I_q-S_{\mathrm N}(c,a),
 \qquad S_{\mathrm N}(c,a)=\cov_{c,a}(W_{\mathrm N}).
\]
Moreover, there is $\kappa_0=\kappa_0(c,A,d,q)<1$ such that
\[
 \sup_{\|a\|\le A}\|\mathrm Dm_c(a)\|_{\mathrm{op}}\le\kappa_0.
\]
Thus only the spectral separation in \eqref{app:eq:admissible} needs to be
assumed or verified.
\end{lemma}

\begin{proof}
Differentiating the normalised density \eqref{app:eq:flat-law} under the
integral sign gives the score identity
\[
 \frac{\partial}{\partial a_j}\log p_{c,a}(w)
 =-\{w_{{\mathrm N},j}-\E_{c,a}W_{{\mathrm N},j}\}.
\]
Consequently,
\[
 \frac{\partial}{\partial a_j}\E_{c,a}W_{{\mathrm N},i}
 =-\cov_{c,a}(W_{{\mathrm N},i},W_{{\mathrm N},j}),
\]
which proves $\mathrm Dm_c(a)=I_q-S_{\mathrm N}(c,a)$.

Integrating out the tangent coordinates shows that $W_{\mathrm N}$ has density,
on the open ball $\|n\|<c$,
\[
 p_{\mathrm N}(n)\propto
 (c^2-\|n\|^2)^{d/2}\exp\{-\|a+n\|^2/2\}
 =\exp\{-V_a(n)\},
\]
where
\[
 V_a(n)=\tfrac12\|a+n\|^2
        -\tfrac d2\log(c^2-\|n\|^2).
\]
Its Hessian is
\[
 \nabla^2V_a(n)
 =I_q+\frac{d}{c^2-\|n\|^2}I_q
   +\frac{2d}{(c^2-\|n\|^2)^2}nn^{\top}
 \succeq \left(1+\frac d{c^2}\right)I_q.
\]
The Brascamp--Lieb covariance inequality therefore yields
\[
 0\prec S_{\mathrm N}(c,a)
 \preceq \E\{\nabla^2V_a(W_{\mathrm N})^{-1}\}
 \preceq \frac{c^2}{c^2+d}I_q\prec I_q.
\]
The lower inequality is strict because $p_{\mathrm N}$ is positive on the full
$q$-dimensional open ball.  The matrix $S_{\mathrm N}(c,a)$ depends continuously
on $a$; compactness of $\{\|a\|\le A\}$ therefore gives
\[
 \underline s_{c,A}
 =\inf_{\|a\|\le A}\lambda_{\min}\{S_{\mathrm N}(c,a)\}>0.
\]
All eigenvalues of $\mathrm Dm_c(a)=I_q-S_{\mathrm N}(c,a)$ lie in $(0,1)$, and hence
\[
 \|\mathrm Dm_c(a)\|_{\mathrm{op}}
 =1-\lambda_{\min}\{S_{\mathrm N}(c,a)\}
 \le1-\underline s_{c,A}=: \kappa_0<1.
\]
This proves the uniform contraction.  The covariance inequality on the
open ball follows by the standard smooth convex-domain approximation in
the Brascamp--Lieb theorem~\citep{BrascampLieb1976}.
\end{proof}

Throughout Theorems~\ref{app:thm:contraction}--%
\ref{app:thm:sample-orbit-localization}, $h=c\sigma$, $(c,A)$ is {spectrally
admissible}, and
$c\sigma\le r_{\mathrm{geo}}/2$.

\begin{lemma}[Scaled population moments near $\mathcal M$]
\label{app:lem:pop-mean-cov}
For $\mathbf y\in\Phi_{A\sigma}$, choose orthonormal tangent and normal
frames at $x=\Pi_{\mathcal M}(\mathbf y)$ and write
$\boldsymbol\nu(\mathbf y)=\sigma a$ in the normal frame.  Uniformly over this tube,
\[
 \boldsymbol\mu_h(\mathbf y)
 =x+\sigma(\mathbf0_d,m_c(a))+O(\sigma^2),\qquad
 \boldsymbol\Sigma_h(\mathbf y)
 =\sigma^2 S_c(a)+O(\sigma^3),
\]
where the displayed vectors and matrices are transported to the ambient
frames.  Under the eigengap condition in \eqref{app:eq:admissible},
\[
 \|\mathbf P_h(\mathbf y)-\mathbf P_x\|_{\mathrm{op}}
 \le C\sigma .
\]
Consequently,
\[
 \mathbf Q_x\{\mathcal T_h(\mathbf y)-x\}
 =\sigma m_c(a)+O(\sigma^2),\qquad
 \|\mathbf P_x\{\mathcal T_h(\mathbf y)-x\}\|\le C\sigma^2.
\]
{Moreover, after writing
$\mathbf y=x+\sigma a$ in any smooth local orthonormal frame of the normal
bundle, the displayed expansions and the projector bound hold uniformly
in the $C^1$ norm with respect to $(x,a)$ on compact sub-tubes
$\|a\|\le A_{\mathrm{sub}}<A$.  {The $O(D-d)$-equivariance of
$m_c$ and $S_c$ makes these local expansions consistent on overlapping
normal-bundle charts.}}
\end{lemma}

\begin{proof}
{Let $p_\sigma$ be the ambient density of one observation.
For $z=x+\sigma b$, the Gaussian-convolution formula and uniform
sampling on $\mathcal M$ give
\[
 p_\sigma(z)
 =\frac{(2\pi\sigma^2)^{-D/2}}{\mathrm{Area}(\mathcal M)}
   \int_{\mathcal M}
   \exp\{-\|z-\xi\|^2/(2\sigma^2)\}\,dA(\xi).
\]
Use a smooth normal-bundle chart at $x$ and put
$\xi=\exp_x(\sigma s)$.  Uniformly for bounded $b$,
\[
 \frac{\|x+\sigma b-\exp_x(\sigma s)\|^2}{\sigma^2}
 =\|b_{\mathrm T}-s\|^2+\|b_{\mathrm N}\|^2
   -\sigma\langle b_{\mathrm N},\mathrm{II}_x(s,s)\rangle
   +O\{\sigma^2(1+\|s\|^6)\},
\]
on a fixed chart, and
$dA\{\exp_x(\sigma s)\}=\sigma^d
\{1+O(\sigma^2\|s\|^2)\}\,ds$.  After restricting first to
$\|s\|\le C\sqrt{\log(1/\sigma)}$, Taylor expansion under the integral
is dominated by a polynomial times a Gaussian.  The remaining part of
that chart, and the part of $\mathcal M$ outside it, are respectively
Gaussian-tail and $O\{\exp(-c_{\mathrm{tail}}/\sigma^2)\}$ terms for some
$c_{\mathrm{tail}}>0$.  The same domination
holds after one derivative in $b$ or in the chart centre and frame.
Consequently, uniformly on every bounded bundle chart,
\[
 \sigma^{D-d}p_\sigma(x+\sigma b)
 =C_{\mathrm{dens}}\exp(-\|b_{\mathrm N}\|^2/2)+O_{C^1}(\sigma).
\]
Here the leading integral over $s\in\R^d$ removes $b_{\mathrm T}$, and
$C_{\mathrm{dens}}=(2\pi)^{-(D-d)/2}/\mathrm{Area}(\mathcal M)$.

Now write $\mathbf y=x+\sigma a$, $a\in N_x\mathcal M$, and
$w=(\mathbf Y-\mathbf y)/\sigma$. This query-centred change of variables
turns every hard neighbourhood into the fixed ball $\{\|w\|\le c\}$.
Thus differentiation with respect to the query changes the integrand but
introduces no moving-boundary term. On this fixed ball, the preceding
expansion is applicable with $b=a+w$ and shows that the unnormalized
conditional density is
\[
 C_{\mathrm{dens}}\exp\{-\|a+w_{\mathrm N}\|^2/2\}+O_{C^1}(\sigma).
\]
Its integral is uniformly bounded away from zero for $\|a\|\le A$.
Division by this integral therefore gives
\[
 p_{c,a}(w)+O_{C^1}(\sigma)
 \quad\text{on }\{\|w\|\le c\}.
\]
The leading law is exactly \eqref{app:eq:flat-law}; tangent reflection gives
$\E_{c,a}W_{\mathrm T}=\mathbf0_d$.  Since
$\mathbf Y=\mathbf y+\sigma w$, integration of $w$ and $ww^{\top}$ now yields
\[
 \boldsymbol\mu_h(\mathbf y)
 =x+\sigma\{\mathbf0_d,a+\E_{c,a}W_{\mathrm N}\}+O_{C^1}(\sigma^2),
 \qquad
 \boldsymbol\Sigma_h(\mathbf y)=\sigma^2S_c(a)+O_{C^1}(\sigma^3).
\]
This proves both moment expansions, including their asserted uniform
$C^1$ form on compact sub-tubes.  By \eqref{app:eq:admissible}, the
normalised leading covariance $S_c(a)$ has eigengap at least $g_0$.
For small $\sigma$, the normalised population covariance has gap at
least $g_0/2$.  The spectral-projector map is $C^1$ with uniformly
bounded derivative on matrices with this gap, so its standard resolvent
formula (equivalently, the $C^1$ Davis--Kahan perturbation bound) gives
$\|\mathbf P_h(\mathbf y)-\mathbf P_x\|_{C^1}\le C\sigma$.

Finally, the leading term of
$\mathbf y-\boldsymbol\mu_h(\mathbf y)$ is normal and has size
$O(\sigma)$. Substitution in \eqref{app:eq:population-update} shows that the
$O(\sigma)$ projector error contributes only $O(\sigma^2)$. Projection
onto $N_x\mathcal M$ and $T_x\mathcal M$ gives the two stated
consequences.}
\end{proof}

{
\begin{lemma}[Flat-model eigengap criterion]
\label{app:lem:gap}
For the law \eqref{app:eq:flat-law}, define
\[
 \alpha_{\mathrm T}(c,\|a\|)
 =d^{-1}\operatorname{tr}\{\cov_{c,a}(W_{\mathrm T})\}.
\]
Then
\[
 S_c(a)=
 \begin{pmatrix}
   \alpha_{\mathrm T}(c,\|a\|)I_d&0\\
   0&S_{\mathrm N}(c,a)
 \end{pmatrix}.
\]
Hence the leading $d$ eigenspace is tangent uniformly on $\|a\|\le A$ if
and only if
\[
 \inf_{\|a\|\le A}
 \{\alpha_{\mathrm T}(c,\|a\|)
 -\lambda_{\max}(S_{\mathrm N}(c,a))\}>0.
\]
This is exactly the spectral admissibility condition used above; normal
contraction is automatic by Lemma~\ref{app:lem:auto-normal-contraction}. Both
covariance terms are continuous and are given by bounded integrals over the
fixed ball $\{\|w\|\le c\}$.
\end{lemma}

\begin{proof}
Reflection in each tangent coordinate makes every tangent--normal
cross-moment zero, and rotation invariance in the tangent coordinates
makes the tangent block a scalar multiple of $I_d$. The spectral statement
follows. Continuity follows by differentiation under the integral sign on
the bounded domain.
\end{proof}

\begin{remark}[Eigenspace separation]
\label{app:rem:honest}
Conditioning on the neighbourhood event couples signal and noise, so the
conditional eigengap is determined by the limiting integrals in
Lemma~\ref{app:lem:gap}, evaluated for the chosen $d,D,c$ and $A$.
As $c\downarrow0$ the conditioned law approaches the uniform law on a
small ball, the two covariance blocks become isotropic, and their gap
vanishes at $O(c^4)$. The localization results assume spectral
admissibility at every constituent scale.
\end{remark}

\begin{table}[ht]
\centering
\caption{Numerically estimated minimum normalised spectral gaps for the
flat reference law at normalised tube width $A=3$.}
\label{app:tab:flat-spectral-gaps}
\begin{tabular}{c@{\qquad}cc}
\hline
$c=h/\sigma$ & $d=1,\ q=2$ & $d=2,\ q=1$ \\
\hline
1.0 & 0.0207 & 0.0267 \\
1.25 & 0.0448 & 0.0628 \\
1.5 & 0.0824 & 0.1243 \\
2.0 & 0.2106 & 0.3496 \\
2.5 & 0.4331 & 0.6030 \\
3.0 & 0.7855 & 0.9267 \\
3.5 & 1.3115 & 1.3578 \\
4.0 & 2.0575 & 1.9269 \\
5.0 & 4.3429 & 3.5847 \\
6.0 & 7.6804 & 5.9915 \\
8.0 & 16.8168 & 12.7253 \\
\hline
\end{tabular}
\end{table}

For each row, the reported quantity is
\[
 \inf_{\|a\|\le3}
 \left[\lambda_{\min}\{\cov_{c,a}(W_{\mathrm T})\}
 -\lambda_{\max}\{\cov_{c,a}(W_{\mathrm N})\}\right].
\]
The positive values provide a numerical check of spectral separation for
these dimension pairs and radius ratios. The theoretical results assume
spectral admissibility on $\Phi_{3\sigma}$.
}

\begin{remark}[Effective expansion parameter]
{The normalised remainder has effective size
$c^2\sigma/r_{\mathrm{geo}}$, as is visible from the curvature term in
the preceding exponent. Thus a large calibrated $c$ requires
$c^2\sigma\ll r_{\mathrm{geo}}$. Under the uniform sampling condition,
the leading local law has no additional tangential density-score term.}
\end{remark}

\subsection{Population fixed-point manifold and attraction}

\begin{theorem}[Local normal-distance recursion]
\label{app:thm:contraction}
Under Assumption~1 of the main paper and the {spectral
admissibility condition}
above, there are $\sigma_{\mathrm{con}}>0$, $C<\infty$, and
$\kappa\in(\kappa_0,1)$ such that, for $\sigma\le\sigma_{\mathrm{con}}$ and
$\mathbf y\in\Phi_{A\sigma}$,
\[
 \dist\bigl(\mathcal T_h(\mathbf y),\mathcal M\bigr)
 \le \kappa\,\dist(\mathbf y,\mathcal M)+C\sigma^2 .
\]
For sufficiently small $\sigma$, $\Phi_{A\sigma}$ is forward invariant,
and any population orbit starting there obeys
\[
 \dist\bigl(\mathcal T_h^{(t)}(\mathbf y),\mathcal M\bigr)
 \le\kappa^t\dist(\mathbf y,\mathcal M)
 +\frac{C\sigma^2}{1-\kappa}.
\]
{This result controls normal localization; complete ambient
convergence and fixed-point existence are established below.}
\end{theorem}

\begin{proof}
With the notation of Lemma~\ref{app:lem:pop-mean-cov},
$\mathbf y-x=\sigma(\mathbf0_d,a)$ and
\[
 \mathcal T_h(\mathbf y)-x
 =\sigma(\mathbf0_d,m_c(a))+O(\sigma^2).
\]
Indeed, the leading displacement
$\mathbf y-\boldsymbol\mu_h(\mathbf y)$ is normal, the leading projection
is tangent, and multiplying its $O(\sigma)$ projection error by an
$O(\sigma)$ displacement contributes only $O(\sigma^2)$.  Put
$\mathbf v_0=x+\sigma(\mathbf0_d,m_c(a))$.  It lies on the normal fibre inside the
reach tube, so $\dist(\mathbf v_0,\mathcal M)=\sigma\|m_c(a)\|$.  Since distance
to a closed set is one-Lipschitz and
$\|\mathcal T_h(\mathbf y)-\mathbf v_0\|=O(\sigma^2)$, we have
\[
 \dist\bigl(\mathcal T_h(\mathbf y),\mathcal M\bigr)
 \le \sigma\|m_c(a)\|+C\sigma^2.
\]
Tangent and normal reflection symmetry gives
$m_c(\mathbf0_q)=\mathbf0_q$; the mean-value
theorem and Lemma~\ref{app:lem:auto-normal-contraction} give
$\|m_c(a)\|\le\kappa_0\|a\|$.  {The preceding displays give
$\dist\bigl(\mathcal T_h(\mathbf y),\mathcal M\bigr)
\le\kappa_0\dist(\mathbf y,\mathcal M)+C\sigma^2$, which immediately implies
the stated inequality for every $\kappa>\kappa_0$.}  Iteration yields the affine
recursion.  Forward invariance follows from
$\kappa A\sigma+C\sigma^2\le A\sigma$ after decreasing $\sigma_{\mathrm{con}}$.
\end{proof}

{
Ambient convergence additionally requires control of tangential motion.
We obtain it by identifying the complete nearby zero set of the update
residual and controlling the map transverse to that set. Recall the
population residual $\mathbf F_h$ and its zero set
$\mathcal S_{h,{\mathrm{pop}}}$ from Section~2.2 of the main paper.

\begin{lemma}[Fixed-point manifold and transverse contraction]
\label{app:lem:fixed-manifold}
Fix $0<A_0<A$.  Under the conditions of
Theorem~\ref{app:thm:contraction}, there are
$\sigma_{\mathrm{fix}},\rho_0,C,C_{\mathrm{step}}>0$ and $\bar\kappa\in(0,1)$ such that, for
$\sigma\le\sigma_{\mathrm{fix}}$:
\begin{enumerate}
\item[(i)] $\mathcal M_h:=\mathcal S_{h,{\mathrm{pop}}}\cap\Phi_{A_0\sigma}$
is a compact $C^1$ $d$-dimensional normal graph
\[
 \mathcal M_h=\{x+b_h(x):x\in\mathcal M,\ b_h(x)\in N_x\mathcal M\},
 \qquad \|b_h\|_{C^1}\le C\sigma^2,
\]
and
$\Pi_{\mathcal M}:\mathcal M_h\to\mathcal M$ is a $C^1$
diffeomorphism;
\item[(ii)]
\[
 d_H(\mathcal M_h,\mathcal M)\le C\sigma^2;
\]
\item[(iii)] on
$\mathcal U_h=\{\mathbf y:\dist(\mathbf y,\mathcal M_h)<\rho_0\sigma\}$
contained in $\Phi_{A\sigma}$,
\[
 \dist\bigl(\mathcal T_h(\mathbf y),\mathcal M_h\bigr)
 \le\bar\kappa\,\dist(\mathbf y,\mathcal M_h),
 \qquad
 \|\mathcal T_h(\mathbf y)-\mathbf y\|
 \le C_{\mathrm{step}}\,\dist(\mathbf y,\mathcal M_h).
\]
\end{enumerate}
\end{lemma}

\begin{proof}
Work first in one smooth normal-bundle chart and write
$\mathbf y=x+\sigma a$, with $x\in\mathcal M$ and $\|a\|\le A_0$.
The residual $\mathbf F_h(\mathbf y)$ lies in the estimated normal space
$T_h(\mathbf y)^\perp$.  By
Lemma~\ref{app:lem:pop-mean-cov}, this estimated normal space makes angle
$O(\sigma)$ with $N_x\mathcal M$; hence orthogonal projection from it
onto $N_x\mathcal M$ is invertible for small $\sigma$.  Use this
projection to identify the residual with a vector in $N_x\mathcal M$ and
define its scaled coordinate representation $\mathbf H_{\sigma,x}(a)$.
More explicitly, using $\mathbf Q_x=I_D-\mathbf P_x$, set
\[
 \mathbf H_{\sigma,x}(a)
 =\sigma^{-1}\mathbf Q_x\mathbf F_h(x+\sigma a).
\]
The restriction of $\mathbf Q_x$ to
$T_h(x+\sigma a)^\perp$ is one-to-one, so
$\mathbf H_{\sigma,x}(a)=\mathbf0$ is equivalent to
$\mathbf F_h(x+\sigma a)=\mathbf0$.
The uniform $C^1$ expansion in Lemma~\ref{app:lem:pop-mean-cov} gives
\[
 \mathbf H_{\sigma,x}(a)
 =a-m_c(a)+\mathbf r_\sigma(x,a),
 \qquad
 \|\mathbf r_\sigma\|_{C^1}\le C\sigma .
\]
In particular,
\[
 \mathrm D_a\mathbf H_{\sigma,x}(a)
 =I_q-\mathrm Dm_c(a)+O(\sigma)
 =\cov_{c,a}(W_{\mathrm N})+O(\sigma).
\]
The smallest eigenvalue of $\cov_{c,a}(W_{\mathrm N})$ is bounded away
from zero uniformly on $\|a\|\le A_0$; hence the normal Jacobian is
uniformly invertible for sufficiently small $\sigma$.
Thus $\mathbf F_h(x+\sigma a)=\mathbf0$ is equivalent to
\[
 a=m_c(a)-\mathbf r_\sigma(x,a).
\]
The right-hand side maps the closed $A_0$-ball into itself and is a
contraction in $a$, uniformly in $x$, because
$\|\mathrm Dm_c(a)\|_{\mathrm{op}}\le\kappa_0<1$ and
$\|\mathrm D_a\mathbf r_\sigma\|\le C\sigma$; on the boundary its norm is at
most $\kappa_0A_0+C\sigma<A_0$.  It has a unique solution
$a_h(x)$ for every $x$, and $a_h(x)=O(\sigma)$ because
$m_c(\mathbf0_q)=\mathbf0_q$ and
$\mathbf r_\sigma(x,\mathbf0_q)=O(\sigma)$.  The parameter-dependent contraction
theorem, equivalently the implicit-function theorem applied to
$\mathbf H_{\sigma,x}$, makes $a_h$ a $C^1$ section of the normal
bundle, with $\|a_h\|_{C^1}\le C\sigma$.  On overlaps uniqueness makes
these local sections agree.
Set $b_h(x)=\sigma a_h(x)$. Consequently
\[
 \mathcal M_h=\{x+b_h(x):x\in\mathcal M\},
 \qquad \|b_h\|_{C^1}\le C\sigma^2,
\]
is a compact $C^1$ manifold of fixed points, its nearest-point projection
onto $\mathcal M$ is a diffeomorphism, and its graph representation gives
$d_H(\mathcal M_h,\mathcal M)\le C\sigma^2$. The same contraction
argument excludes any other residual zero in $\Phi_{A_0\sigma}$.
Choose $\rho_0<(A-A_0)/2$. After decreasing $\sigma_{\mathrm{fix}}$ so
that the graph displacement is smaller than
$(A-A_0)\sigma/2$, the resulting attraction neighbourhood
$\mathcal U_h$ is contained in $\Phi_{A\sigma}$.

It remains to control points off $\mathcal M_h$.  In bundle coordinates
centred at the graph $a_h(x)$, write
\[
 \mathcal T_h(x+\sigma a)
 =X_\sigma(x,a)+\sigma A_\sigma(x,a),
 \qquad
 A_\sigma(x,a)\in N_{X_\sigma(x,a)}\mathcal M .
\]
The $C^1$ moment and projector expansions, together with the fact that
the leading displacement is normal, give
\[
 X_\sigma(x,a)=x+O_{C^1}(\sigma^2),\qquad
 A_\sigma(x,a)=m_c(a)+O_{C^1}(\sigma).
\]
Because the graph consists of pointwise fixed points,
\[
 X_\sigma\{x,a_h(x)\}=x,\qquad
 A_\sigma\{x,a_h(x)\}=a_h(x).
\]
For the transverse coordinate $\boldsymbol\Delta=a-a_h(x)$, the mean-value theorem
therefore gives
\[
 A_\sigma(x,a)-a_h\{X_\sigma(x,a)\}
 =
 \left[\mathrm Dm_c\{a_h(x)\}+O(\sigma)+O(\|\boldsymbol\Delta\|)\right]
 \boldsymbol\Delta .
\]
The term from differentiating $a_h\{X_\sigma(x,a)\}$ is included in
$O(\sigma)$ because $\|a_h\|_{C^1}=O(\sigma)$ and
$\mathrm D_aX_\sigma=O(\sigma^2)$.  Choose
$\kappa_0<\kappa_1<\bar\kappa<1$. If $C_\ast$ bounds the remainder
derivatives and coordinate distortion, choose
$C_\ast(\sigma_{\mathrm{fix}}+\rho_0)\le
\tfrac12\min\{\kappa_1-\kappa_0,\bar\kappa-\kappa_1\}$.
Then the bracket has norm at most $\kappa_1$; in particular, the
attraction radius deteriorates as $\kappa_0\uparrow1$. Uniformly in this tube,
\[
 \dist(x+\sigma a,\mathcal M_h)
 =\sigma\|\boldsymbol\Delta\|
 \{1+O(\sigma)+O(\|\boldsymbol\Delta\|)\}.
\]
Shrinking $\sigma_{\mathrm{fix}}$ and $\rho_0$ once more absorbs this coordinate
distortion and gives the first inequality in (iii).  Finally,
$\mathbf F_h$ is $C^1$, vanishes on $\mathcal M_h$, and has uniformly
bounded derivative there.  Another application of the mean-value theorem
gives
\[
 \|\mathcal T_h(\mathbf y)-\mathbf y\|
 =\|\mathbf F_h(\mathbf y)\|
 \le C_{\mathrm{step}}\,\dist(\mathbf y,\mathcal M_h),
\]
which completes the proof.
\end{proof}

\begin{theorem}[Population fixed-point manifold and local orbit convergence]
\label{app:thm:orbit-convergence}
Under the conditions of Lemma~\ref{app:lem:fixed-manifold}, the fixed-point
set in $\Phi_{A_0\sigma}$ is the normal graph $\mathcal M_h$ described
there, with $\|b_h\|_{C^1}\le C\sigma^2$ and
$d_H(\mathcal M_h,\mathcal M)\le C\sigma^2$. Every
$\mathbf y^{(0)}\in\mathcal U_h$ generates an orbit
$\mathbf y^{(t+1)}=\mathcal T_h(\mathbf y^{(t)})$ that converges to a
fixed point $\mathbf y^{(\infty)}\in\mathcal M_h$.  More precisely,
\[
 \dist(\mathbf y^{(t)},\mathcal M_h)
 \le\bar\kappa^t\dist(\mathbf y^{(0)},\mathcal M_h)
\]
and
\[
 \|\mathbf y^{(t)}-\mathbf y^{(\infty)}\|
 \le
 \frac{C_{\mathrm{step}}\bar\kappa^t}{1-\bar\kappa}
 \dist(\mathbf y^{(0)},\mathcal M_h).
\]
In particular,
$\dist(\mathbf y^{(\infty)},\mathcal M)\le C\sigma^2$.
\end{theorem}

\begin{proof}
Lemma~\ref{app:lem:fixed-manifold}(iii) makes $\mathcal U_h$ forward
invariant and gives the first display by induction.  It also gives
\[
 \|\mathbf y^{(t+1)}-\mathbf y^{(t)}\|
 \le
 C_{\mathrm{step}}\bar\kappa^t\dist(\mathbf y^{(0)},\mathcal M_h).
\]
The right-hand side is summable.  Hence the orbit is Cauchy and converges
in $\R^D$ to a point $\mathbf y^{(\infty)}$; summing the tail from $t$
onward gives the stated error bound.  Since the distance to the closed
set $\mathcal M_h$ tends to zero,
$\mathbf y^{(\infty)}\in\mathcal M_h$, and every point of
$\mathcal M_h$ is fixed.  Part (ii) of the lemma yields the final
distance bound.
\end{proof}
}

\section{Finite-sample theory for STSP}
\subsection{Empirical fixed points and orbit localization}
\label{app:sec:sample-fixpt}
We analyse the empirical operator $\widehat{\mathcal T}_{\mathcal Z,h}$
constructed from the fixed i.i.d.\ reference sample $\mathcal Z$.
For reference, the empirical neighbourhood and local moments are
\[
 \mathcal I_{\mathcal Z,h}(\mathbf y)
 =\{j:\|\mathbf z_j-\mathbf y\|\le h\},
 \qquad n_{\mathcal Z,h}(\mathbf y)
 =|\mathcal I_{\mathcal Z,h}(\mathbf y)|,
\]
and, whenever $n_{\mathcal Z,h}(\mathbf y)>0$,
\[
 \widehat{\boldsymbol\mu}_{\mathcal Z,h}(\mathbf y)
 =\frac{1}{n_{\mathcal Z,h}(\mathbf y)}
   \sum_{j\in\mathcal I_{\mathcal Z,h}(\mathbf y)}\mathbf z_j,
\]
\[
 \widehat{\boldsymbol\Sigma}_{\mathcal Z,h}(\mathbf y)
 =\frac{1}{n_{\mathcal Z,h}(\mathbf y)}
   \sum_{j\in\mathcal I_{\mathcal Z,h}(\mathbf y)}
   \{\mathbf z_j-\widehat{\boldsymbol\mu}_{\mathcal Z,h}(\mathbf y)\}
   \{\mathbf z_j-\widehat{\boldsymbol\mu}_{\mathcal Z,h}(\mathbf y)\}^{\top}.
\]
Let $\widehat{\mathbf P}_{\mathcal Z,h}(\mathbf y)$ be its leading-$d$
spectral projector whenever the $d$th and $(d+1)$st eigenvalues are
strictly separated, and define
\[
 \widehat{\mathcal T}_{\mathcal Z,h}(\mathbf y)
 =\widehat{\boldsymbol\mu}_{\mathcal Z,h}(\mathbf y)
  +\widehat{\mathbf P}_{\mathcal Z,h}(\mathbf y)
   \{\mathbf y-\widehat{\boldsymbol\mu}_{\mathcal Z,h}(\mathbf y)\}.
\]

\begin{lemma}[Uniform empirical approximation]
\label{app:lem:empirical-approximation}
Under the conditions of Theorem~\ref{app:thm:contraction}, let $0<\delta<1$ and
\[
 r_N=
 \sqrt{\frac{D\log(1/h)+\log(1/\delta)}{Nh^d}}
 +\frac{D\log(1/h)+\log(1/\delta)}{Nh^d}.
\]
There are constants $C,r_0,c_{\mathrm{gap}}>0$ such that, if
$r_N\le r_0$, then with probability at least $1-\delta$, every empirical
neighbourhood in $\Phi_{A\sigma}$ is nonempty and
\[
 \sup_{\mathbf y\in\Phi_{A\sigma}}
 \|\widehat{\boldsymbol\mu}_{\mathcal Z,h}(\mathbf y)
       -\boldsymbol\mu_h(\mathbf y)\|\le Chr_N,
\]
\[
 \sup_{\mathbf y\in\Phi_{A\sigma}}
 \|\widehat{\boldsymbol\Sigma}_{\mathcal Z,h}(\mathbf y)
       -\boldsymbol\Sigma_h(\mathbf y)\|_{\mathrm{op}}\le Ch^2r_N,
\]
\[
 \sup_{\mathbf y\in\Phi_{A\sigma}}
 \|\widehat{\mathbf P}_{\mathcal Z,h}(\mathbf y)
       -\mathbf P_h(\mathbf y)\|_{\mathrm{op}}\le Cr_N,
 \qquad
 \inf_{\mathbf y\in\Phi_{A\sigma}}
 \{\widehat\lambda_d(\mathbf y)-\widehat\lambda_{d+1}(\mathbf y)\}
 \ge c_{\mathrm{gap}}h^2.
\]
On the same event,
\[
 \sup_{\mathbf y\in\Phi_{A\sigma}}
 \|\widehat{\mathcal T}_{\mathcal Z,h}(\mathbf y)
       -\mathcal T_h(\mathbf y)\|\le Chr_N.
\]
\end{lemma}

\begin{proof}
The ambient-density expansion in Lemma~\ref{app:lem:pop-mean-cov} is uniform
for scaled normal and tangential coordinates in bounded sets. Hence, for
$\mathbf y\in\Phi_{A\sigma}$ and $\mathbf z\in
B_{\mathrm E}(\mathbf y,h)$, it gives
$p_\sigma(\mathbf z)\ge c_1\sigma^{-(D-d)}$ for all sufficiently small
$\sigma$. Therefore
\[
 \inf_{\mathbf y\in\Phi_{A\sigma}}
 \Pr\{\mathbf Y\in B_{\mathrm E}(\mathbf y,h)\}
 \ge c_1\sigma^{-(D-d)}\omega_Dh^D
 =c_{\mathrm{mass}}h^d,
\]
where $h=c\sigma$ and $c_{\mathrm{mass}}>0$ depends only on the fixed
parameters. The same local density expansion gives the upper bound
$\sup_{\Phi_{A\sigma}}p_h(\mathbf y)\le C_{\mathrm{mass}}h^d$ for
the neighbourhood mass $p_h$ defined in the main paper.

Write $\mathbf v_{\mathbf y}(\mathbf z)=(\mathbf z-\mathbf y)/h$ and
$I_{\mathbf y}(\mathbf z)=\mathbf1_{\{\|\mathbf z-\mathbf y\|\le h\}}$.
For $k,\ell=1,\ldots,D$, consider the classes consisting of
$I_{\mathbf y}$, $v_{\mathbf y,k}I_{\mathbf y}$, and
$v_{\mathbf y,k}v_{\mathbf y,\ell}I_{\mathbf y}$ as $\mathbf y$ ranges
over $\Phi_{A\sigma}$. Each has envelope one, and every member has
second moment at most $C_{\mathrm{mass}}h^d$. The ball class is VC, and
the coordinate-weighted classes are bounded semialgebraic
VC-subgraph classes with complexity depending only on $D$.
In particular, each class has a uniform $L_2(Q)$ covering bound of the
form $N(\varepsilon,\mathcal F,L_2(Q))\le(A_D/\varepsilon)^{V_D}$ for
every probability measure $Q$, with finite constants $A_D,V_D$ depending
only on the fixed ambient dimension. Applying the VC maximal inequality
of \citet{vdvwellner1996} to these bounded classes, followed by the
bounded empirical-process tail inequality of \citet{bousquet2002}, uses
their variance bound $C_{\mathrm{mass}}h^d$ and yields, simultaneously for
the finitely many component classes, that
\[
 \sup_{\mathbf y\in\Phi_{A\sigma}}
 \left|\widehat Q_j(\mathbf y)-Q_j(\mathbf y)\right|
 \le C\left\{\sqrt{\frac{h^d u}{N}}+\frac{u}{N}\right\},
 \qquad u=D\log(1/h)+\log(1/\delta),
\]
where $Q_j$ denotes any scalar component of the three unnormalized
moments and $\widehat Q_j$ its empirical counterpart. At the
$h^{d/2}$ variance scale, the covering bound contributes a constant
multiple of $\log(1/h)$; because $D$ is fixed, its dimension-dependent
factor and the finite number of components are absorbed into $C$ in the
displayed bound with $u=D\log(1/h)+\log(1/\delta)$.
In particular, $Q_0=p_h\ge c_{\mathrm{mass}}h^d$, and, when $r_N$ is
small enough, $\widehat Q_0\ge c_{\mathrm{mass}}h^d/2$. Dividing the
first- and second-moment bounds by these denominators and using
$\|Q_1\|,\|Q_2\|\le Q_0$ gives
\[
 \sup_{\Phi_{A\sigma}}
 \left\|\frac{\widehat Q_1}{\widehat Q_0}
       -\frac{Q_1}{Q_0}\right\|\le Cr_N,
 \qquad
 \sup_{\Phi_{A\sigma}}
 \left\|\frac{\widehat Q_2}{\widehat Q_0}
       -\frac{Q_2}{Q_0}\right\|_{\mathrm{op}}\le Cr_N .
\]
Here $Q_1$ and $Q_2$ denote the vector and matrix unnormalized moments,
respectively. Rescaling them by $h$ and $h^2$ gives the conditional
query-centred moment errors. To make the covariance conversion explicit,
define the empirical
query-centred conditional moments
\[
 \widehat m_1(\mathbf y)=n_{\mathcal Z,h}(\mathbf y)^{-1}
 \sum_{j\in\mathcal I_{\mathcal Z,h}(\mathbf y)}(\mathbf z_j-\mathbf y),
\quad
 \widehat M_2(\mathbf y)=n_{\mathcal Z,h}(\mathbf y)^{-1}
 \sum_{j\in\mathcal I_{\mathcal Z,h}(\mathbf y)}
 (\mathbf z_j-\mathbf y)(\mathbf z_j-\mathbf y)^\top.
\]
Then, exactly,
\[
 \widehat{\boldsymbol\mu}_{\mathcal Z,h}(\mathbf y)
 =\mathbf y+\widehat m_1(\mathbf y),\qquad
 \widehat{\boldsymbol\Sigma}_{\mathcal Z,h}(\mathbf y)
 =\widehat M_2(\mathbf y)-\widehat m_1(\mathbf y)
  \widehat m_1(\mathbf y)^\top,
\]
with the same identities for the population conditional moments. Since
$\|\widehat m_1\|$ and its population counterpart are at most $h$, the
uniform query-centred first- and second-moment bounds give the displayed
centred covariance bound. The factors $h$ and $h^2$ arise because every
observation in a ball is within $h$ of its query. The count bound also
makes every empirical count at least a fixed multiple of $Nh^d$.

The population gap is at least
$g_0\sigma^2/2=g_0h^2/(2c^2)$. Decrease $r_0$ so that the covariance
perturbation is at most $g_0h^2/(8c^2)$. Weyl's inequality, applied to
both eigenvalues defining the gap, gives the
empirical gap with $c_{\mathrm{gap}}=g_0/(4c^2)$, and Davis--Kahan gives
$\sup\|\widehat{\mathbf P}_{\mathcal Z,h}-\mathbf P_h\|_{\mathrm{op}}
\le Cr_N$. Since
$\|\mathbf y-\boldsymbol\mu_h(\mathbf y)\|=O(h)$ on the shrinking tube,
\[
 \sup_{\mathbf y\in\Phi_{A\sigma}}
 \|\widehat{\mathcal T}_{\mathcal Z,h}(\mathbf y)-\mathcal T_h(\mathbf y)\|
 \le Ch r_N .
\]
This proves the lemma.
\end{proof}

\begin{corollary}[Localization of empirical fixed points]
\label{app:thm:sample-fixpt-dist}
Under the conditions of Lemma~\ref{app:lem:empirical-approximation}, there
are constants $C,r_{\mathrm{fp}}>0$ such that, if
$r_N\le r_{\mathrm{fp}}$, then with probability at least $1-\delta$,
every fixed point $\widehat{\mathbf y}_N\in\Phi_{A\sigma}$ of
$\widehat{\mathcal T}_{\mathcal Z,h}$, whenever one exists, satisfies
\[
 \dist(\widehat{\mathbf y}_N,\mathcal M)
 \le C\{\sigma^2+h r_N\}.
\]
Consequently, for fixed $d,D,c,A$, $h=c\sigma$ and
\[
 N\ge C\sigma^{-(d+2)}
 \{D\log(1/\sigma)+\log(1/\delta)\}
\]
suffice for the right-hand side to be $O(\sigma^2)$.
\end{corollary}

\begin{proof}
On the event in Lemma~\ref{app:lem:empirical-approximation}, if
$\widehat{\mathbf y}_N=\widehat{\mathcal T}_{\mathcal Z,h}
(\widehat{\mathbf y}_N)$, then
\[
 \begin{split}
 \dist(\widehat{\mathbf y}_N,\mathcal M)
 &\le \dist\bigl(\mathcal T_h(\widehat{\mathbf y}_N),\mathcal M\bigr)
       +\|\widehat{\mathcal T}_{\mathcal Z,h}(\widehat{\mathbf y}_N)
                -\mathcal T_h(\widehat{\mathbf y}_N)\|\\
 &\le \kappa\dist(\widehat{\mathbf y}_N,\mathcal M)+C\sigma^2+Chr_N.
 \end{split}
\]
Division by $1-\kappa$ proves the assertion. The
sample-size statement follows by requiring $r_N=O(\sigma)$.
\end{proof}

{
\begin{theorem}[Finite-sample orbit localization]
\label{app:thm:sample-orbit-localization}
For the same fixed $0<A_0<A$, under the conditions of
Corollary~\ref{app:thm:sample-fixpt-dist}, there are
$\sigma_{\mathrm{orb}},r_{\mathrm{orb}},C>0$ such that, if $\sigma\le\sigma_{\mathrm{orb}}$ and
$r_N\le r_{\mathrm{orb}}$, then with probability at least $1-\delta$, every
empirical orbit
\[
 \widehat{\mathbf y}^{(t+1)}
 =\widehat{\mathcal T}_{\mathcal Z,h}(\widehat{\mathbf y}^{(t)}),
 \qquad \widehat{\mathbf y}^{(0)}\in\Phi_{A_0\sigma},
\]
remains in $\Phi_{A\sigma}$ and satisfies, for every $t\ge0$,
\[
 \dist(\widehat{\mathbf y}^{(t)},\mathcal M)
 \le
 \kappa^t\dist(\widehat{\mathbf y}^{(0)},\mathcal M)
 +\frac{C(\sigma^2+h r_N)}{1-\kappa}.
\]
In particular, after the geometric transient the entire empirical orbit
lies in an $O(\sigma^2+h r_N)$ neighbourhood of $\mathcal M$; if
$r_N=O(\sigma)$ and $h=c\sigma$, this neighbourhood has width
$O(\sigma^2)$.
\end{theorem}

\begin{proof}
On the uniform event in Lemma~\ref{app:lem:empirical-approximation}, for every
$\mathbf y\in\Phi_{A\sigma}$,
\[
 \begin{split}
 \dist\bigl(\widehat{\mathcal T}_{\mathcal Z,h}(\mathbf y),\mathcal M\bigr)
 &\le \dist\bigl(\mathcal T_h(\mathbf y),\mathcal M\bigr)
       +\|\widehat{\mathcal T}_{\mathcal Z,h}(\mathbf y)-\mathcal T_h(\mathbf y)\|\\
 &\le \kappa\dist(\mathbf y,\mathcal M)+C\sigma^2+Ch r_N.
 \end{split}
\]
Let $\beta_N=C(\sigma^2+h r_N)$.  Iterating this affine recursion, as long
as the orbit remains in $\Phi_{A\sigma}$, gives
\[
 \dist(\widehat{\mathbf y}^{(t)},\mathcal M)
 \le\kappa^t\dist(\widehat{\mathbf y}^{(0)},\mathcal M)
     +\beta_N\frac{1-\kappa^t}{1-\kappa}.
\]
Because $A_0<A$ and $h=c\sigma$, decrease $\sigma_{\mathrm{orb}}$ and $r_{\mathrm{orb}}$ so that
$\beta_N/(1-\kappa)\le(A-A_0)\sigma$.  The last display then proves by
induction that every iterate remains in $\Phi_{A\sigma}$, so the
recursion is valid for all $t$.  Dropping the factor
$1-\kappa^t$ proves the stated bound.  If $r_N=O(\sigma)$, then
$\sigma^2+h r_N=O(\sigma^2)$.
The event in Lemma~\ref{app:lem:empirical-approximation} is uniform over the
whole tube. Thus the induction uses the same event for every $t\ge0$ and
requires no derivative or continuity of the hard-radius empirical
operator.
\end{proof}

\subsection{Observed initialization and iteration complexity}

\begin{proposition}[Coverage of observed initializations]
\label{app:prop:observed-initialization}
Fix $0<\eta<1$ and choose constants $A_{\mathrm{init}}<A_0<A$ such that
$\Pr\{\|\boldsymbol\xi\|\le A_{\mathrm{init}}\}\ge1-\eta$ for
$\boldsymbol\xi\sim\mathcal N(\mathbf0,I_D)$. Then a generic observation
satisfies
\[
 \Pr\{\dist(\mathbf Y,\mathcal M)\le A_{\mathrm{init}}\sigma\}
 \ge1-\eta.
\]
Moreover, with probability at least $1-\delta_{\mathrm{init}}$, the
fraction of the $N$ observed initializations in
$\Phi_{A_{\mathrm{init}}\sigma}$ is at least
\[
 1-\eta-\sqrt{\frac{\log(1/\delta_{\mathrm{init}})}{2N}}.
\]
Every such initialization falls within the fixed tube covered uniformly
by Theorem~\ref{app:thm:sample-orbit-localization}.
\end{proposition}

\begin{proof}
Write $\mathbf y_i=\mathbf x_i+\sigma\boldsymbol\xi_i$ with
$\boldsymbol\xi_i\sim\mathcal N(\mathbf0,I_D)$.  Then
$\dist(\mathbf y_i,\mathcal M)\le\sigma\|\boldsymbol\xi_i\|$, which proves
the single-observation bound. The indicators of membership in
$\Phi_{A_{\mathrm{init}}\sigma}$ are independent Bernoulli variables with
mean at least $1-\eta$. Hoeffding's inequality gives the displayed sample
fraction. Since $A_{\mathrm{init}}<A_0$, every counted observation meets
the initialization condition of the finite-sample orbit theorem.
\end{proof}

Theorem~\ref{app:thm:sample-orbit-localization} also gives a finite-iteration
interpretation. Put
$d_0=\dist(\widehat{\mathbf y}^{(0)},\mathcal M)$. If
\[
 L\ge
 \frac{\left[\log\left\{(1-\kappa)d_0/
       [C(\sigma^2+hr_N)]\right\}\right]_+}{|\log\kappa|},
\]
then the $L$th output is within twice the statistical error radius of
$\mathcal M$. If $d_0=O(\sigma)$ and
$\sigma^2+hr_N=O(\sigma^2)$, this requires
$L=O\{\log(1/\sigma)\}$.

}

\section{Theory for MS-STSP}
\label{app:sec:multiscale}

{We first derive the quadratic expansion of the clean local
average that motivates squared-scale extrapolation. We then use stable
geometric aggregation to establish localization for the noisy
fixed-reference MS-STSP output.}

\subsection{Quadratic structure of the clean local average}

For $0<h\le c_0r_{\mathrm{geo}}$, recall
\[
 f_h(x)=
 \frac{1}{\mathrm{Area}\{B_{\mathrm E}(x,h)\cap\mathcal M\}}
 \int_{B_{\mathrm E}(x,h)\cap\mathcal M}y\,dA(y).
\]

\begin{lemma}[Quadratic structure of local averaging bias]
\label{app:lem:local_avg_expansion}
Under conditions (A2)--(A4) of the main paper, uniformly in
$x\in\mathcal M$,
\[
 f_h(x)=x+\tfrac12h^2\mathbf B(x)+O(h^4),
\]
where, for any orthonormal basis $\{e_1,\ldots,e_d\}$ of
$T_x\mathcal M$,
\begin{equation}
\label{app:eq:Bx-def}
 \mathbf B(x)=\frac{1}{d+2}\sum_{k=1}^d
 \mathrm{II}_x(e_k,e_k)\in N_x\mathcal M.
\end{equation}
The remainder is uniform in $x$.
\end{lemma}

\begin{proof}
Fix $x\in\mathcal M$ and use geodesic normal coordinates
$u\in T_x\mathcal M$. Let $\mathcal J_x(u)$ be the area density and
$D_h(x)=\{u:\exp_x(u)\in B_{\mathrm E}(x,h)\cap\mathcal M\}$.
Then
\[
 f_h(x)-x=
 \frac{\int_{D_h(x)}\{\exp_x(u)-x\}\mathcal J_x(u)\,du}
      {\int_{D_h(x)}\mathcal J_x(u)\,du}.
\]
Uniform Taylor expansions give
\begin{align*}
 \exp_x(u)-x
 &=u+\tfrac12\mathrm{II}_x(u,u)+F_{3,x}(u)+F_{4,x}(u)
   +O(\|u\|^5),\\
 \mathcal J_x(u)&=1+\mathcal J_{2,x}(u)+\mathcal J_{3,x}(u)+O(\|u\|^4),
\end{align*}
where the indexed terms are homogeneous of the indicated degree.

It remains to account for the fact that $D_h(x)$ is not exactly a
symmetric tangent ball.  Write $u=r\theta$.  Orthogonality of
$u$ and $\mathrm{II}_x(u,u)$ gives
\[
 \|\exp_x(r\theta)-x\|^2
 =r^2+r^4H_{4,x}(\theta)+r^5H_{5,x}(\theta)+O(r^6),
\]
where $H_{4,x}$ is even in $\theta$.  Solving the boundary equation
$\|\exp_x\{\rho_x(\theta,h)\theta\}-x\|=h$ yields
\[
 \rho_x(\theta,h)=h+h^3\alpha_x(\theta)+h^4\beta_x(\theta)
 +O(h^5),
 \qquad \alpha_x(-\theta)=\alpha_x(\theta).
\]
Thus the order-$h^3$ boundary perturbation has zero pairing with the
odd linear term $u$.  Its first possible nonzero boundary contribution
to that term is $O(h^{d+4})$.  All degree-one and degree-three interior
terms also integrate to zero on the leading symmetric ball; the remaining
numerator corrections are $O(h^{d+4})$.  Moreover,
\[
 \int_{D_h(x)}\mathcal J_x(u)\,du=\omega_dh^d\{1+O(h^2)\}.
\]
Consequently,
\[
 f_h(x)-x=
 \frac{1}{\omega_dh^d}
 \int_{\|u\|\le h}\tfrac12\mathrm{II}_x(u,u)\,du+O(h^4).
\]
Since
\[
 \int_{\|u\|\le h}u_au_b\,du
 =\mathbf 1\{a=b\}\frac{\omega_d}{d+2}h^{d+2},
\]
the leading term is exactly $\tfrac12h^2\mathbf B(x)$.
The trace $\sum_a\mathrm{II}_x(e_a,e_a)$ is invariant under an
orthogonal change of tangent basis, so \eqref{app:eq:Bx-def} is intrinsic.
Compactness and the uniform bounds in (A2)--(A4) make the stated
$O(h^4)$ remainder uniform over chart centres.
\end{proof}

\subsection{Squared-scale intercept weights}

\begin{lemma}[OLS intercept identities]
\label{app:lem:ms-weights}
Let $s_j=h_j^2$, $\bar s=M^{-1}\sum_js_j$,
$\bar{\mathbf v}=M^{-1}\sum_j\mathbf v_j$, and
$S_{ss}=\sum_j(s_j-\bar s)^2>0$. The intercept in the ordinary
least-squares regression of vectors $\mathbf v_j\in\R^D$ on
$(1,s_j)$ is $\sum_jw_j\mathbf v_j$, where
\[
 w_j=\frac1M-\frac{\bar s(s_j-\bar s)}{S_{ss}}.
\]
These weights satisfy
\[
 \sum_jw_j=1,
 \qquad
 \sum_jw_js_j=0.
\]
If $h_j=c_j\sigma$ with fixed positive $c_j$ that are not all equal,
then the weights do not depend on $\sigma$ and
$\Lambda_w:=\sum_j|w_j|<\infty$.
\end{lemma}

\begin{proof}
The OLS slope is
$\widehat{\mathbf b}=S_{ss}^{-1}\sum_j(s_j-\bar s)\mathbf v_j$ and the
intercept is $\bar{\mathbf v}-\bar s\widehat{\mathbf b}$, which gives the
displayed formula. Since
$\sum_j(s_j-\bar s)=0$,
\[
 \sum_jw_j=1.
\]
Furthermore,
\[
 \sum_j(s_j-\bar s)s_j
 =\sum_j(s_j-\bar s)^2=S_{ss},
\]
so $\sum_jw_js_j=\bar s-\bar s=0$. If $s_j=c_j^2\sigma^2$, the common
factor $\sigma^2$ cancels from $\bar s(s_j-\bar s)/S_{ss}$. Thus every
$w_j$ is a fixed finite number and so is its finite $\ell_1$ norm.
\end{proof}

For the clean local-average map in Lemma~\ref{app:lem:local_avg_expansion},
the two weight identities give, uniformly in $x\in\mathcal M$,
\[
 \sum_{j=1}^Mw_j f_{h_j}(x)
 =x+O\!\left(\Lambda_w h_{\max}^4\right),
 \qquad h_{\max}=\max_j h_j.
\]
This follows by cancelling the common $h_j^2\mathbf B(x)$ term and
bounding the weighted remainders by
$C\sum_j|w_j|h_j^4\le C\Lambda_w h_{\max}^4$.
The noisy-output result below instead uses geometric aggregation.

\subsection{A geometric aggregation lemma}

\begin{lemma}[Stable affine aggregation in one manifold patch]
\label{app:lem:local-affine-aggregation}
Fix $a_{\mathrm{patch}}<\infty$ and an a priori weight cap $C_w<\infty$. There are
$\sigma_{\mathrm{agg}},C<\infty$, depending only on $a_{\mathrm{patch}},C_w$ and the uniform geometric
constants, with the following property. Let $0<\sigma\le\sigma_{\mathrm{agg}}$,
$x_0\in\mathcal M$, and $\mathbf v_1,\ldots,\mathbf v_M$ lie in the reach
tube. Put $x_j=\Pi_{\mathcal M}(\mathbf v_j)$ and suppose
\[
 d_g(x_j,x_0)\le a_{\mathrm{patch}}\sigma,
 \qquad
 \dist(\mathbf v_j,\mathcal M)\le\varepsilon_{\mathrm{agg}}.
\]
If $\sum_jw_j=1$ and $\sum_j|w_j|\le C_w$, then
\[
 \dist\left(\sum_jw_j\mathbf v_j,\mathcal M\right)
 \le C C_w(\varepsilon_{\mathrm{agg}}+\sigma^2).
\]
The constant $C$ may depend on the fixed cap $C_w$, but not on
$\sigma,\varepsilon_{\mathrm{agg}},N$, or the particular points.
\end{lemma}

\begin{proof}
Use an orthogonal chart centred at $x_0$ in which the manifold is the
graph
\[
 x_0+u+\psi(u),\qquad u\in T_{x_0}\mathcal M,
\]
with $\psi(\mathbf0)=\mathbf0$, $\mathrm D\psi(\mathbf0)=0$, and
$\|\psi(u)\|\le C\|u\|^2$ on a fixed ball. Positive reach and the derivative
bounds give this chart uniformly in $x_0$. For sufficiently small
$\sigma$, geodesic--Euclidean comparability puts every $x_j$ in this
chart, so write $x_j=x_0+u_j+\psi(u_j)$ with
$\|u_j\|\le C a_{\mathrm{patch}}\sigma$. Also write
$\mathbf v_j=x_j+n_j$, where
$\|n_j\|\le\varepsilon_{\mathrm{agg}}$.

Set $u_w=\sum_jw_ju_j$. Then
$\|u_w\|\le C C_wa_{\mathrm{patch}}\sigma$, so after reducing $\sigma_{\mathrm{agg}}$ the point
$x_w=x_0+u_w+\psi(u_w)$ belongs to $\mathcal M$ in the same chart. Because
$\sum_jw_j=1$,
\[
 \sum_jw_j\mathbf v_j-x_w
 =\sum_jw_j\psi(u_j)-\psi(u_w)+\sum_jw_jn_j.
\]
The normal-error term has norm at most $C_w\varepsilon_{\mathrm{agg}}$. The graph bounds
give
\[
 \left\|\sum_jw_j\psi(u_j)-\psi(u_w)\right\|
 \le C\left\{C_w a_{\mathrm{patch}}^2+C_w^2a_{\mathrm{patch}}^2\right\}\sigma^2.
\]
Since the weight cap is fixed and $C_w\ge|\sum_jw_j|=1$, this is bounded
by $C C_w\sigma^2$ after absorbing the fixed cap into $C$. Taking
distance to the particular manifold point $x_w$ proves the result.
\end{proof}

\subsection{Population MS-STSP}

\begin{theorem}[Population MS-STSP localization]
\label{app:thm:ms-population}
Let $M$ be fixed, let $h_j=c_j\sigma$, and suppose every $(c_j,A)$ is
spectrally admissible. For each scale, let $\mathcal M_{h_j}$ and
$\mathcal U_{h_j}$ be supplied by Lemma~\ref{app:lem:fixed-manifold}, and
write $\rho_{0,j}$, $\bar\kappa_j$, and $C_{{\mathrm{step}},j}$ for its
scale-$h_j$ constants. Assume
\[
 \mathbf y^{(0)}\in\bigcap_{j=1}^M\mathcal U_{h_j}.
\]
Let
\[
 \mathbf y_j^{(\infty)}=\lim_{t\to\infty}\mathcal T_{h_j}^{(t)}(\mathbf y^{(0)}),
 \qquad
 \mathbf y_{\mathrm{MS}}^{(\infty)}=\sum_{j=1}^Mw_j\mathbf y_j^{(\infty)},
\]
where $\sum_jw_j=1$ and $\sum_j|w_j|\le C_w$ for a fixed cap. Then, for
all sufficiently small $\sigma$,
\[
 \dist(\mathbf y_j^{(\infty)},\mathcal M)\le C\sigma^2,
 \qquad
 d_g\{\Pi_{\mathcal M}(\mathbf y_j^{(\infty)}),
       \Pi_{\mathcal M}(\mathbf y^{(0)})\}\le C\sigma,
\]
uniformly over $j$, and
\[
 \dist(\mathbf y_{\mathrm{MS}}^{(\infty)},\mathcal M)\le C C_w\sigma^2.
\]
Here and below $C$ may depend on the fixed finite scale design and its
weight cap, but not on $\sigma$.
\end{theorem}

\begin{proof}
Apply Theorem~\ref{app:thm:orbit-convergence} separately at every scale.
Because $M$ is fixed, the maxima of the finitely many constants are still
finite. Thus every limit exists, belongs to $\mathcal M_{h_j}$, and
$\dist(\mathbf y_j^{(\infty)},\mathcal M)\le C\sigma^2$.

The tail estimate in that theorem at $t=0$ gives
\[
 \|\mathbf y_j^{(\infty)}-\mathbf y^{(0)}\|
 \le\frac{C_{{\mathrm{step}},j}}{1-\bar\kappa_j}
      \dist(\mathbf y^{(0)},\mathcal M_{h_j})
 \le C\sigma,
\]
because $\mathbf y^{(0)}\in\mathcal U_{h_j}$. Let
$x_0=\Pi_{\mathcal M}(\mathbf y^{(0)})$ and
$x_j=\Pi_{\mathcal M}(\mathbf y_j^{(\infty)})$. Nearest-point projection is uniformly
Lipschitz on a smaller reach tube, so
$\|x_j-x_0\|\le C\sigma$. Local geodesic--Euclidean comparability then
gives $d_g(x_j,x_0)\le C\sigma$.
Lemma~\ref{app:lem:local-affine-aggregation}, with
$\varepsilon_{\mathrm{agg}}=C\sigma^2$, proves the
last display.
\end{proof}

\subsection{Finite-sample MS-STSP}

\begin{theorem}[Finite-sample MS-STSP localization]
\label{app:thm:ms-sample}
Retain the finite scale design of Theorem~\ref{app:thm:ms-population}. Let
$0<\delta<1$ and assume
the common initialization has a uniform attraction-tube buffer,
\[
 \dist(\mathbf y^{(0)},\mathcal M_{h_j})\le\tfrac12\rho_{0,j}\sigma,
 \qquad j=1,\ldots,M.
\]
Define
\[
 r_{N,j}=\sqrt{\frac{D\log(1/h_j)+\log(M/\delta)}{Nh_j^d}}
 +\frac{D\log(1/h_j)+\log(M/\delta)}{Nh_j^d}.
\]
Suppose $\max_j r_{N,j}$ is small enough that the joint empirical event
below holds and
\[
 \frac{C h_jr_{N,j}}{1-\bar\kappa_j}
 \le\tfrac12\rho_{0,j}\sigma
 \quad\text{for every }j.
\]
For predetermined finite iteration horizons $L_j$ and a fixed
stopping-horizon bound $C_{\mathrm{hor}}<\infty$, assume
\[
 \bar\kappa_j^{L_j}\dist(\mathbf y^{(0)},\mathcal M_{h_j})
 \le C\{\sigma^2+h_jr_{N,j}\},
 \qquad
 L_j\{\sigma^2+h_jr_{N,j}\}\le C_{\mathrm{hor}}\sigma.
\]
For $t=0,\ldots,L_j$, define the empirical branch orbit and its aggregate by
\[
 \widehat{\mathbf y}_j^{(t)}
 =\widehat{\mathcal T}_{\mathcal Z,h_j}^{(t)}(\mathbf y^{(0)}),
 \qquad
 \widehat{\mathbf y}_{\mathrm{MS}}=\sum_{j=1}^Mw_j\widehat{\mathbf y}_j^{(L_j)},
\]
with $\sum_jw_j=1$ and $\sum_j|w_j|\le C_w$. Then, with probability at
least $1-\delta$,
\[
 \dist(\widehat{\mathbf y}_{\mathrm{MS}},\mathcal M)
 \le C C_w\left\{\sigma^2+\max_j h_jr_{N,j}\right\}.
\]
If $h_j=c_j\sigma$ and
\[
 N\ge C\sigma^{-(d+2)}
 \{D\log(1/\sigma)+\log(M/\delta)\},
\]
then $\max_j r_{N,j}=O(\sigma)$ and the right-hand side is
$O(\sigma^2)$.
\end{theorem}

\begin{proof}
Apply the uniform empirical approximation in
Lemma~\ref{app:lem:empirical-approximation} at scale $j$ with failure
probability $\delta/M$. A union bound gives, with probability at least
$1-\delta$, simultaneously for every $j$,
\[
 \sup_{\mathbf y\in\Phi_{A\sigma}}
 \|\widehat{\mathcal T}_{\mathcal Z,h_j}(\mathbf y)-\mathcal T_{h_j}(\mathbf y)\|
 \le C h_jr_{N,j},
\]
and all empirical projectors are uniquely defined by their eigengaps.

On the attraction tube, Lemma~\ref{app:lem:fixed-manifold}(iii) and the last
display imply
\[
 \dist\bigl(\widehat{\mathcal T}_{\mathcal Z,h_j}(\mathbf y),\mathcal M_{h_j}\bigr)
 \le\bar\kappa_j\dist(\mathbf y,\mathcal M_{h_j})+C h_jr_{N,j}.
\]
The initialization buffer and the stated smallness condition prove by
induction that the empirical branch remains in $\mathcal U_{h_j}$. Hence
\[
 \dist(\widehat{\mathbf y}_j^{(L_j)},\mathcal M_{h_j})
 \le\bar\kappa_j^{L_j}
      \dist(\mathbf y^{(0)},\mathcal M_{h_j})
   +\frac{C h_jr_{N,j}}{1-\bar\kappa_j}.
\]
Since
$\sup_{\mathbf z\in\mathcal M_{h_j}}\dist(\mathbf z,\mathcal M)
\le C\sigma^2$, the transient
condition yields
\[
 \dist(\widehat{\mathbf y}_j^{(L_j)},\mathcal M)
 \le C\{\sigma^2+h_jr_{N,j}\}.
\]

We next establish cross-scale locality. For the empirical iterate
$\widehat{\mathbf y}_j^{(t)}$, the population step bound in
Lemma~\ref{app:lem:fixed-manifold}(iii) and the uniform perturbation bound give
\[
 \|\widehat{\mathbf y}_j^{(t+1)}-\widehat{\mathbf y}_j^{(t)}\|
 \le C_{{\mathrm{step}},j}
 \dist(\widehat{\mathbf y}_j^{(t)},\mathcal M_{h_j})+C h_jr_{N,j}.
\]
Summing the preceding affine recursion for
$t=0,\ldots,L_j-1$ gives
\[
 \|\widehat{\mathbf y}_j^{(L_j)}-\mathbf y^{(0)}\|
 \le C\sigma+C L_j h_jr_{N,j}
 \le C\sigma,
\]
where the last inequality uses the stopping-horizon condition. Projection
Lipschitzness and local geodesic--Euclidean comparability therefore show
that all $\Pi_{\mathcal M}(\widehat{\mathbf y}_j^{(L_j)})$ lie within $C\sigma$ of
$\Pi_{\mathcal M}(\mathbf y^{(0)})$ in geodesic distance.

Apply Lemma~\ref{app:lem:local-affine-aggregation} with
$\varepsilon_{\mathrm{agg}}=C\{\sigma^2+\max_j h_jr_{N,j}\}$ to obtain the asserted
multi-scale bound. Finally, for fixed $M$ and fixed positive $c_j$, the
displayed sample-size lower bound makes every $r_{N,j}=O(\sigma)$, which
proves the last statement.
\end{proof}

\section{Numerical and reproducibility details}

\subsection{Simulation manifolds}

\medskip\noindent\textit{Curves.}
\begin{enumerate}
   \item Tilted Circle ($\vartheta \in [0,2\pi]$):
   \[
       \begin{pmatrix} x \\ y \\ z \end{pmatrix} =
       \begin{pmatrix}
           1 & 0 & 0 \\
           0 & \cos(\pi/6) & -\sin(\pi/6) \\
           0 & \sin(\pi/6) & \cos(\pi/6)
       \end{pmatrix}
       \begin{pmatrix} \cos(\vartheta) \\ \sin(\vartheta) \\ 0 \end{pmatrix}
   \]

   \item Tennis Curve ($\vartheta \in [-\pi,\pi]$):
   \[
       \begin{pmatrix} x \\ y \\ z \end{pmatrix} =
       \begin{pmatrix}
           \cos^3(\vartheta) \\
           \sin^3(\vartheta) \\
           \frac{\sqrt{3}}{2}\sin(2\vartheta)
       \end{pmatrix}
   \]

   \item Trefoil Knot ($\vartheta \in [0,2\pi]$):
   \[
       \begin{pmatrix} x \\ y \\ z \end{pmatrix} =
       \begin{pmatrix}
           \sin(\vartheta)+2\sin(2\vartheta) \\
           \cos(\vartheta)-2\cos(2\vartheta) \\
           -\sin(3\vartheta)
       \end{pmatrix}
   \]

   \item Sixfolds ($\vartheta \in [0,12\pi]$):
   \[
       \begin{pmatrix} x \\ y \\ z \end{pmatrix} =
       \begin{pmatrix}
           \cos(\vartheta)(1 + 0.3\cos(6\vartheta)) \\
           \sin(\vartheta)(1 + 0.3\cos(6\vartheta)) \\
           0.3\sin(6\vartheta)
       \end{pmatrix}
   \]

\end{enumerate}

\medskip\noindent\textit{Surfaces.}
\begin{enumerate}
   \item Sphere ($\theta \in [0,2\pi], \phi \in [0,\pi]$):
   \[
       \begin{pmatrix} x \\ y \\ z \end{pmatrix} =
       \begin{pmatrix}
           \sin(\phi)\cos(\theta) \\
           \sin(\phi)\sin(\theta) \\
           \cos(\phi)
       \end{pmatrix}
   \]

   \item Torus ($\theta,\phi \in [0,2\pi]$):
   \[
       \begin{pmatrix} x \\ y \\ z \end{pmatrix} =
       \begin{pmatrix}
           (1 + 0.3\cos(\phi))\cos(\theta) \\
           (1 + 0.3\cos(\phi))\sin(\theta) \\
           0.3\sin(\phi)
       \end{pmatrix}
   \]

\end{enumerate}

\begin{figure}[H]
\centering
\includegraphics[width=\linewidth]{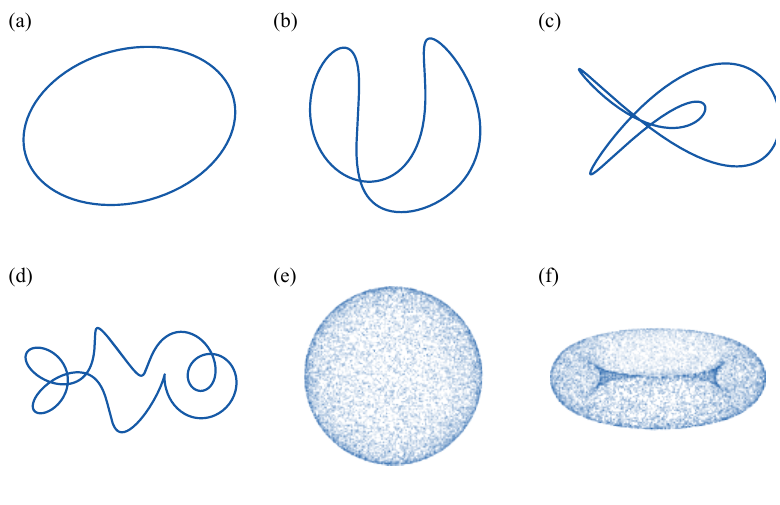}
\caption{Underlying manifolds used in the synthetic study:
(a) circle, (b) tennis curve, (c) trefoil knot, (d) sixfolds,
(e) sphere, and (f) torus.}
\label{app:fig:supp-simulation-manifolds}
\end{figure}

For each independent synthetic repetition, the curve input is sampled by
inverting a numerical cumulative arc-length table with 100,001 parameter
grid points. The Sixfolds interval traverses the same curve six times;
equal-length traversals induce the same arc-length distribution as one
fundamental period. Sphere inputs are sampled using a uniform height
coordinate and a uniform azimuth. Torus inputs are uniform in surface area:
$\theta$ is uniform on $[0,2\pi]$, while $\phi$ has density
$(1+0.3\cos\phi)/(2\pi)$, generated by numerically inverting
$F(\phi)=(\phi+0.3\sin\phi)/(2\pi)$.
Independent Gaussian noise is then added coordinatewise. The evaluation
clouds contain 20,000 points for curves and 30,000 points for surfaces and
are held fixed across repetitions of a given geometry. Distances from
fitted points to the manifold are approximated by nearest-neighbour distances
to the corresponding independent evaluation cloud.

The synthetic studies use a common selection rule. Within each synthetic
repetition, every method receives the same underlying-manifold realisation,
noisy realisation, independent evaluation cloud, sample size, and noise level.
Recall the fitted cloud $\widehat{\mathcal Y}$, independent evaluation
cloud $\mathcal X_{\mathrm{eval}}$, and the two mean errors defined in
Section~5.1 of the main paper:
\[
 e_{\mathrm{fit}}=\frac1N\sum_{i=1}^N
 \dist(\widehat{\mathbf y}_i,\mathcal X_{\mathrm{eval}}),\qquad
 e_{\mathrm{cov}}=\frac1{N_{\mathrm{eval}}}\sum_{\ell=1}^{N_{\mathrm{eval}}}
 \dist(\mathbf x_{{\mathrm{eval}},\ell},\widehat{\mathcal Y}).
\]
Tunable parameters are selected separately for each method on the full
data set by minimizing
\[
 J_{\mathrm{mean}}=\max(e_{\mathrm{fit}},e_{\mathrm{cov}}),
\]
with ties resolved first by the analogous joint 95th-percentile error and
then by the maximum fitting error.
The independent evaluation cloud provides a finite-point approximation to
distance from a fitted point to the manifold. This is method-wise full-data
oracle tuning. The displayed
$e_{\mathrm{fit}}$ and $e_{\mathrm{cov}}$ are means over all fitted and evaluation
points, respectively.

\subsection{Common search spaces}
For STSP, MS-STSP, local PCA, and mean shift, the neighbourhood of each
query consists of all fixed reference observations within radius $h$.
The ordinary synthetic studies use
\[
h/\sigma\in\{1.5,2,2.5,3,3.5,4,5,6,8\}.
\]
Local PCA uses $L=1$. STSP and mean shift stop at the first $t\ge1$ for
which the mean displacement divided by $s$ is below $10^{-3}$, or at
$T_{\max}$, whichever occurs first. In the synthetic studies,
$s=\sigma$ and $T_{\max}=30$. MS-STSP applies the same rule to the largest
normalised mean displacement
among its constituent branches and searches finite scale templates from
the same radius pool, with an $\ell_1$ weight screen to exclude unstable
extrapolation designs.
STSP and every MS-STSP branch use the original noisy cloud as their fixed
reference at every iteration. All branches are evaluated at the common
stopping iteration.
Theorem~S5 applies to finite horizons that satisfy its transient and
accumulated-drift conditions. The displacement rule used in the experiments
is a computational stopping criterion; it does not by itself verify those
theorem conditions for the selected stopping time.
For the ordinary synthetic settings, MS-STSP uses the seven triples
of adjacent values in the ordered $h/\sigma$ grid:
\[
\begin{aligned}
 &(1.5,2,2.5),\ (2,2.5,3),\ (2.5,3,3.5),\ (3,3.5,4),\\
 &(3.5,4,5),\ (4,5,6),\ (5,6,8).
\end{aligned}
\]
Each triple is multiplied by $\sigma$ to obtain its three radii.
The high-noise study uses
\[
 h/\sigma\in\{1,1.25,1.5,2,2.5,3,3.5,4\}
\]
and the six triples of adjacent values in this grid:
\[
\begin{aligned}
 &(1,1.25,1.5),\ (1.25,1.5,2),\ (1.5,2,2.5),\ (2,2.5,3),\\
 &(2.5,3,3.5),\ (3,3.5,4).
\end{aligned}
\]
These triples are also multiplied by $\sigma$. In both studies, triples with
$\sum_j|w_j|>5$ are excluded.

As in the main paper, the estimator of \citet{yao2023manifold} is called the
Yao--Yau (YY) estimator. YY-adapted retains the direction from each query
to its local centre in the
published YY construction. Within the orthogonal complement of that
direction, local PCA estimates the known $d$-dimensional tangent subspace;
the remaining directions form the normal subspace used by the update. The
published averaging option is enabled, and the routine's noise input is set
to the simulation value, that is, \texttt{sig} $=\sigma$. Its internal radii
are therefore the published
$r=5\sigma/\log_{10}N$ and
$R=10\sigma\sqrt{\log(1/\sigma)}/\log_{10}N$.
For a query $\mathbf y$, let $\mathcal I_{\mathrm{YY}}(\mathbf y)$ be the union of
the reference points within radius $2r$ and its nearest
$\min\{N,\max(5,d+2)\}$ reference points. Up to floating-point precision,
the local centre is
\[
 \mathbf m_{\mathrm{YY}}(\mathbf y)=
 |\mathcal I_{\mathrm{YY}}(\mathbf y)|^{-1}
 \sum_{i\in\mathcal I_{\mathrm{YY}}(\mathbf y)}\mathbf z_i,
 \qquad
 \mathbf v_{\mathrm{YY}}(\mathbf y)=\mathbf y-\mathbf m_{\mathrm{YY}}(\mathbf y).
\]
For a nonzero displacement, choose an orthonormal
basis matrix $\mathbf U_{\mathrm{YY}}$ for $\mathbf v_{\mathrm{YY}}^{\perp}$ and form the
matrix with rows $(\mathbf z_i-\mathbf m_{\mathrm{YY}})^{\mathsf T}\mathbf U_{\mathrm{YY}}$,
$i\in\mathcal I_{\mathrm{YY}}(\mathbf y)$. If $\mathbf V_{\mathrm{YY},d}$ contains its first
$d$ right singular vectors, the tangent projector is
\[
 \mathbf P_{\mathrm{YY}}(\mathbf y)=
 \mathbf U_{\mathrm{YY}}\mathbf V_{\mathrm{YY},d}\mathbf V_{\mathrm{YY},d}^{\mathsf T}
 \mathbf U_{\mathrm{YY}}^{\mathsf T}.
\]
The cylindrical averaging set is
\[
 \mathcal C_{\mathrm{YY}}(\mathbf y)=\left\{i:
 \|\mathbf P_{\mathrm{YY}}(\mathbf y)(\mathbf z_i-\mathbf y)\|<r,
 \quad\|(I_D-\mathbf P_{\mathrm{YY}}(\mathbf y))(\mathbf z_i-\mathbf y)\|<R
 \right\}.
\]
The preliminary output is the mean over $\mathcal C_{\mathrm{YY}}(\mathbf y)$
when this set contains more than ten points, and is $\mathbf m_{\mathrm{YY}}(\mathbf y)$
otherwise. The final averaging pass takes the mean of preliminary
outputs within radius $r/4$, augmented by the five nearest preliminary
outputs (or all outputs if fewer than five are available).

The original YY construction uses the full space
$\mathbf v_{\mathrm{YY}}^{\perp}$ as its tangent space, with projector
$I_D-\mathbf v_{\mathrm{YY}}\mathbf v_{\mathrm{YY}}^{\mathsf T}/\|\mathbf v_{\mathrm{YY}}\|^2$.
YY-adapted selects $d$ directions within that space, so that the
normal complement has dimension $D-d$. For $d=D-1$, the published routine
is used. For $d<D-1$, local PCA selects the tangent directions while
retaining the mean-displacement direction in the normal space. The radii,
cylindrical averaging rule, and final averaging pass retain the published
construction.

MF25 uses $L_{\max}=50$,
$\varepsilon_{\mathrm{opt}}=10^{-8}$, initial step size one, and ten backtracking steps.
Only $h=c_{\mathrm{MF25}}\sqrt{\sigma}$ is selected, with candidate multipliers
$c_{\mathrm{MF25}}\in\{.75,1,1.25,1.5,1.75,2,2.5\}$.

MF25 uses the complete noisy cloud as both the reference set and the
initial query set; all $N$ fitted outputs enter the evaluation.

\begin{table}[H]
\centering
\caption{MF25 radius multipliers and radii selected for the first
realisation displayed in the fitted-cloud figures. YY-adapted has
no selected external parameter.}
\label{app:tab:supp-mf-settings}
\small
\begin{tabular}{llrr}
\hline
Study & Data & $c_{\mathrm{MF25}}$ & $h=c_{\mathrm{MF25}}\sqrt{\sigma}$ \\
\hline
Curves & Circle & 1.00 & 0.2828 \\
Curves & Tennis curve & 1.00 & 0.2828 \\
Curves & Trefoil knot & 1.75 & 0.4950 \\
Curves & Sixfolds & 1.00 & 0.2828 \\
\hline
Surfaces & Sphere & 1.00 & 0.2828 \\
Surfaces & Torus & 0.75 & 0.2121 \\
\hline
High noise & $\sigma=0.1$ & 1.25 & 0.3953 \\
High noise & $\sigma=0.2$ & 1.25 & 0.5590 \\
High noise & $\sigma=0.3$ & 1.25 & 0.6847 \\
\hline
\end{tabular}

\end{table}

\subsection{SCMS configurations}
Both SCMS variants use isotropic Gaussian weights over the complete fixed
reference cloud, with all noisy observations as initial queries.
SCMS-KDE takes the eigenvectors associated with the smallest $D-d$
eigenvalues of the density Hessian; SCMS-log-KDE uses the corresponding
eigenvectors of the log-density Hessian. Each update adds the projection
of the Gaussian mean-shift vector onto this subspace.
The intrinsic dimension is supplied in the synthetic experiments.
The ordinary synthetic bandwidth pool is
\[
 h/\sigma\in\{.35,.5,.75,1,1.25,1.5,2,2.5,3,3.5,4,5,6,8\},
\]
and the high-noise pool is
\[
 h/\sigma\in\{.25,.35,.5,.75,1,1.25,1.5,2,2.5,3,3.5,4\}.
\]
The full-data geometric criterion selects the bandwidth separately for
each variant. The reported comparisons use $T_{\max}=30$ and the common
mean-displacement threshold.
Some SCMS-KDE configurations reach the iteration cap.

\subsection{Ordinary simulation repetitions}
For each of the six ordinary manifold settings, ten independent noisy samples were
generated. The evaluation cloud is fixed within each setting and is
independent of the noisy samples. The method-specific search described
above was repeated on each noisy sample. The main-paper tables report
the mean and standard deviation of each error across the ten repetitions.

\begin{table}[H]
\centering
\caption{Configurations selected for the first realisation of the six
synthetic settings. The four curves use $N=5000$ and $\sigma=0.08$;
the sphere and torus use $N=20000$ and $\sigma=0.08$.
A list in the radius column gives the three MS-STSP scales.}
\label{app:tab:supp-synthetic-config}
\begin{tabular}{llrrl}
\hline
Data & Method & $h/\sigma$ & $L$ & Radius or scales \\
\hline
Circle & Local PCA & 4 & 1 & $0.32$ \\
Circle & Mean shift & 1.5 & 28 & $0.12$ \\
Circle & SCMS-KDE & 2 & 30 & $0.16$ \\
Circle & SCMS-log-KDE & 1.25 & 9 & $0.1$ \\
Circle & STSP & 2.5 & 8 & $0.2$ \\
Circle & MS-STSP & -- & 5 & $0.32,0.4,0.48$ \\
\hline
Tennis curve & Local PCA & 4 & 1 & $0.32$ \\
Tennis curve & Mean shift & 1.5 & 25 & $0.12$ \\
Tennis curve & SCMS-KDE & 2 & 30 & $0.16$ \\
Tennis curve & SCMS-log-KDE & 1.25 & 9 & $0.1$ \\
Tennis curve & STSP & 2.5 & 8 & $0.2$ \\
Tennis curve & MS-STSP & -- & 5 & $0.32,0.4,0.48$ \\
\hline
Trefoil knot & Local PCA & 5 & 1 & $0.4$ \\
Trefoil knot & Mean shift & 2 & 22 & $0.16$ \\
Trefoil knot & SCMS-KDE & 2.5 & 14 & $0.2$ \\
Trefoil knot & SCMS-log-KDE & 2 & 6 & $0.16$ \\
Trefoil knot & STSP & 4 & 5 & $0.32$ \\
Trefoil knot & MS-STSP & -- & 4 & $0.4,0.48,0.64$ \\
\hline
Sixfolds & Local PCA & 3.5 & 1 & $0.28$ \\
Sixfolds & Mean shift & 1.5 & 22 & $0.12$ \\
Sixfolds & SCMS-KDE & 1.25 & 30 & $0.1$ \\
Sixfolds & SCMS-log-KDE & 1 & 13 & $0.08$ \\
Sixfolds & STSP & 2 & 12 & $0.16$ \\
Sixfolds & MS-STSP & -- & 6 & $0.24,0.28,0.32$ \\
\hline
Sphere & Local PCA & 3.5 & 1 & $0.28$ \\
Sphere & Mean shift & 1.5 & 25 & $0.12$ \\
Sphere & SCMS-KDE & 1.25 & 30 & $0.1$ \\
Sphere & SCMS-log-KDE & 1 & 11 & $0.08$ \\
Sphere & STSP & 2 & 11 & $0.16$ \\
Sphere & MS-STSP & -- & 5 & $0.4,0.48,0.64$ \\
\hline
Torus & Local PCA & 3 & 1 & $0.24$ \\
Torus & Mean shift & 1.5 & 27 & $0.12$ \\
Torus & SCMS-KDE & 0.75 & 30 & $0.06$ \\
Torus & SCMS-log-KDE & 1 & 13 & $0.08$ \\
Torus & STSP & 2 & 12 & $0.16$ \\
Torus & MS-STSP & -- & 8 & $0.24,0.28,0.32$ \\
\hline
\end{tabular}

\end{table}

\begin{figure}[p]
\centering
\includegraphics[width=\linewidth,height=.82\textheight,keepaspectratio]{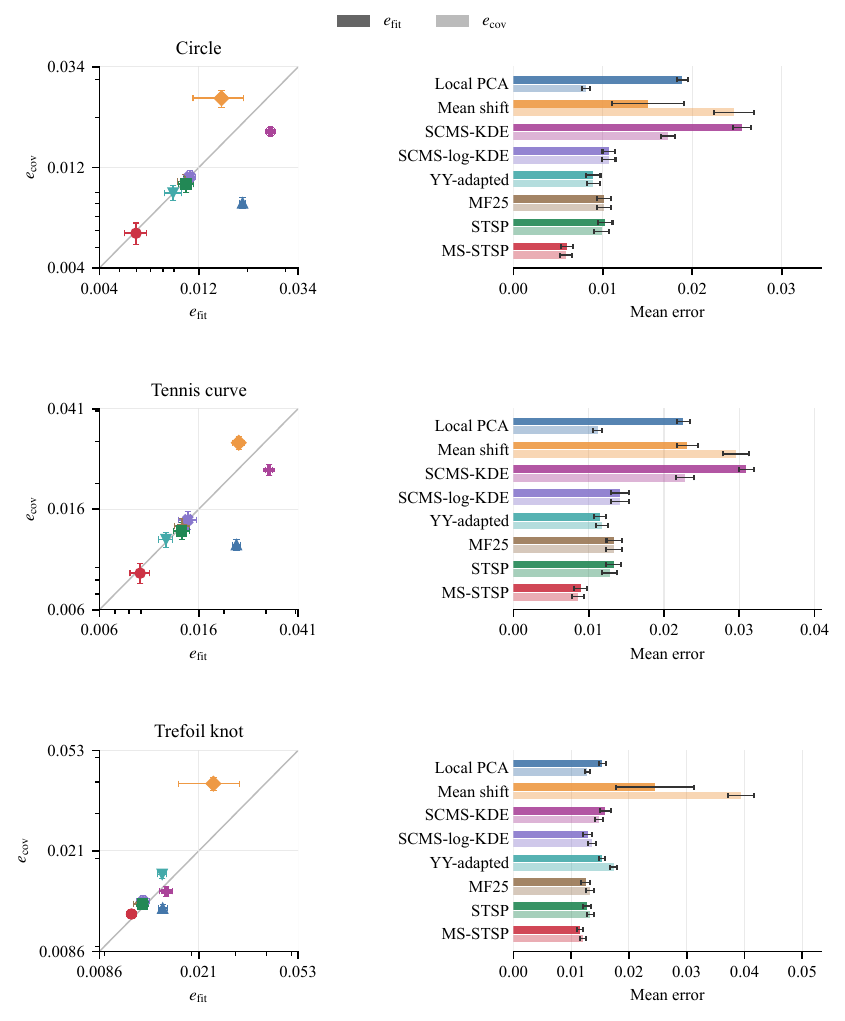}
\caption{Fitting and coverage errors for the circle, tennis curve, and
trefoil knot, from top to bottom, over ten independent realisations.
Each row pairs a scatter plot of method means
on log--log axes with horizontal bars showing the two mean errors.
Error bars show one standard deviation; the grey diagonal marks equal
errors in the scatter plot. Smaller errors are preferable.}
\label{app:fig:supp-synthetic-joint-first}
\end{figure}

\begin{figure}[p]
\centering
\includegraphics[width=\linewidth,height=.82\textheight,keepaspectratio]{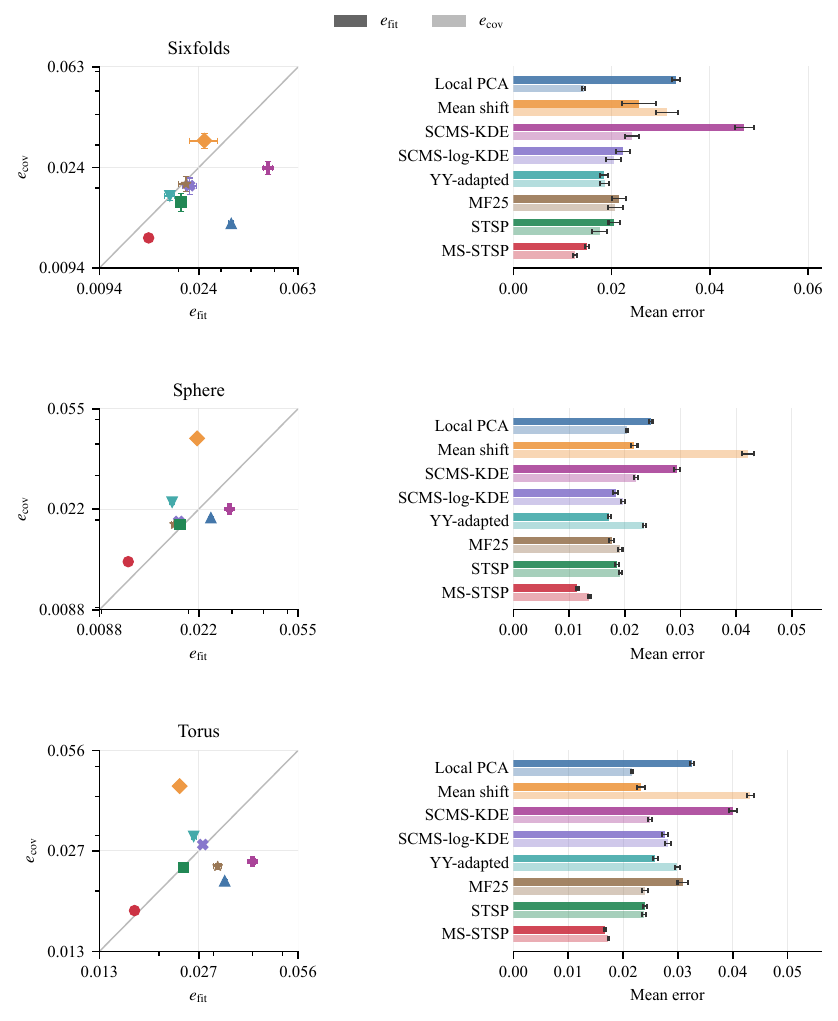}
\caption{Fitting and coverage errors for Sixfolds, the sphere, and the
torus, from top to bottom, continuing Figure~\ref{app:fig:supp-synthetic-joint-first}.
Each setting uses ten independent realisations; graphical conventions
are the same as in Figure~\ref{app:fig:supp-synthetic-joint-first}.}
\label{app:fig:supp-synthetic-joint-continued}
\end{figure}
\clearpage

\begin{figure}[!htbp]
\centering
\includegraphics[width=\linewidth]{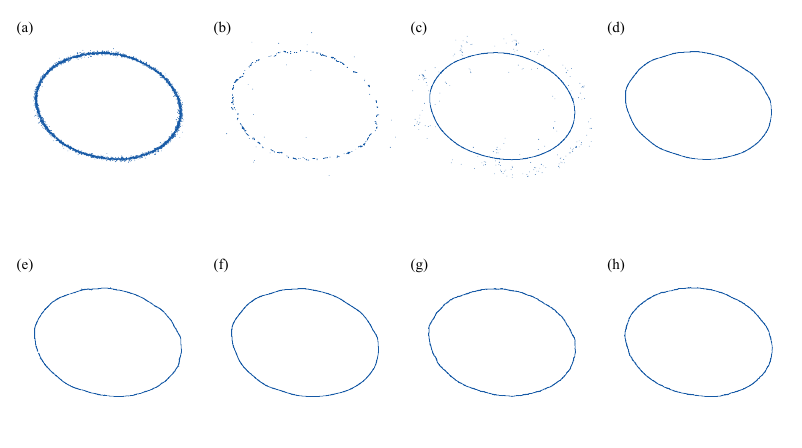}
\caption{Circle fitted clouds. Panels (a)--(h) show local PCA, mean shift,
SCMS-KDE, SCMS-log-KDE, YY-adapted, MF25, STSP, and MS-STSP, respectively.}
\label{app:fig:supp-circle-fits}
\end{figure}

\begin{figure}[!htbp]
\centering
\includegraphics[width=\linewidth]{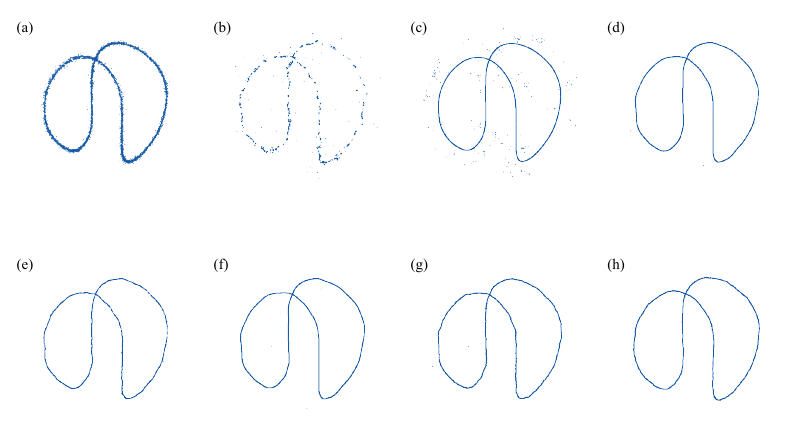}
\caption{Tennis-curve fitted clouds. Panels (a)--(h) show local PCA, mean shift,
SCMS-KDE, SCMS-log-KDE, YY-adapted, MF25, STSP, and MS-STSP, respectively.}
\label{app:fig:supp-tennis-fits}
\end{figure}

\begin{figure}[!htbp]
\centering
\includegraphics[width=\linewidth]{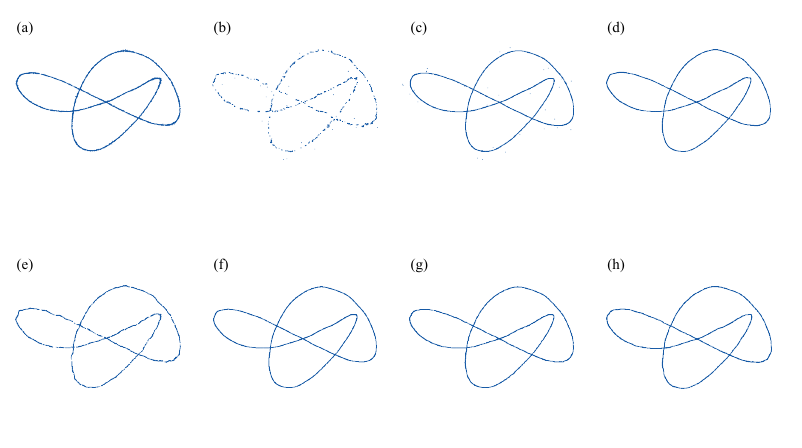}
\caption{Trefoil-knot fitted clouds. Panels (a)--(h) show local PCA, mean shift,
SCMS-KDE, SCMS-log-KDE, YY-adapted, MF25, STSP, and MS-STSP, respectively.}
\label{app:fig:supp-trefoil-fits}
\end{figure}

\begin{figure}[!htbp]
\centering
\includegraphics[width=\linewidth]{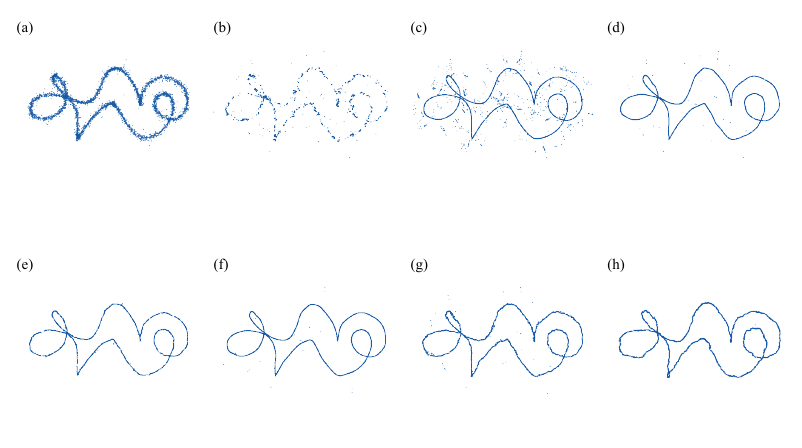}
\caption{Sixfolds fitted clouds. Panels (a)--(h) show local PCA, mean shift,
SCMS-KDE, SCMS-log-KDE, YY-adapted, MF25, STSP, and MS-STSP, respectively.}
\label{app:fig:supp-sixfolds-fits}
\end{figure}

\begin{figure}[!htbp]
\centering
\includegraphics[width=\linewidth]{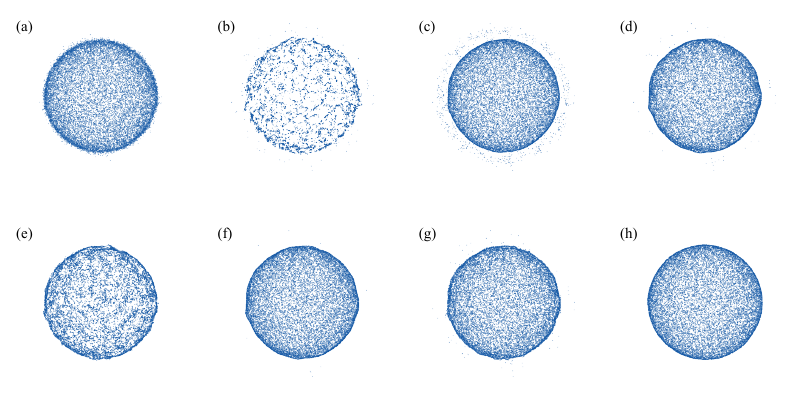}
\caption{Sphere fitted clouds. Panels (a)--(h) show local PCA, mean shift,
SCMS-KDE, SCMS-log-KDE, YY-adapted, MF25, STSP, and MS-STSP, respectively.}
\label{app:fig:supp-sphere-fits}
\end{figure}

\begin{figure}[!htbp]
\centering
\includegraphics[width=\linewidth]{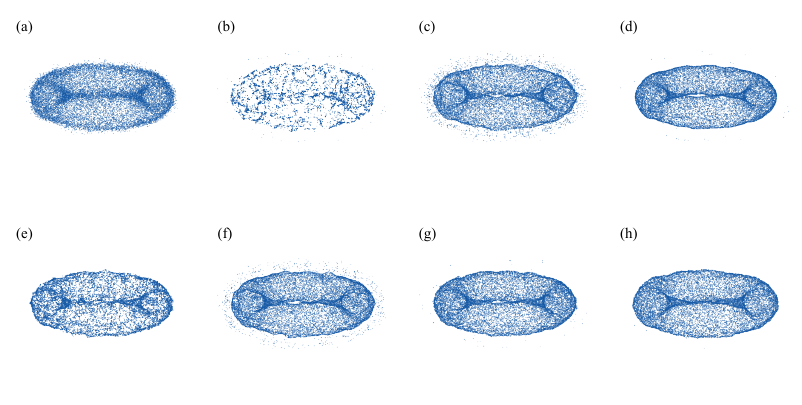}
\caption{Torus fitted clouds. Panels (a)--(h) show local PCA, mean shift,
SCMS-KDE, SCMS-log-KDE, YY-adapted, MF25, STSP, and MS-STSP, respectively.}
\label{app:fig:supp-torus-fits}
\end{figure}

\clearpage
\subsection{High-noise circle experiment}
The high-noise study uses $N=5000$ and
$\sigma\in\{0.1,0.2,0.3\}$. We compare STSP, MS-STSP,
YY-adapted, MF25, local PCA, mean shift, SCMS-KDE, and SCMS-log-KDE.
At each noise level each tunable method is selected by the geometric
criterion using the high-noise candidate sets specified above.

\begin{table}[H]
\centering
\caption{Configurations selected for the first realisation at each
high-noise level.}
\label{app:tab:supp-high-noise-config}
\scriptsize
\begin{tabular}{llrrl}
\hline
Data & Method & $h/\sigma$ & $L$ & Radius or scales \\
\hline
$\sigma=0.1$ & Local PCA & 4 & 1 & $0.4$ \\
$\sigma=0.1$ & Mean shift & 2 & 30 & $0.2$ \\
$\sigma=0.1$ & SCMS-KDE & 2 & 30 & $0.2$ \\
$\sigma=0.1$ & SCMS-log-KDE & 1.25 & 9 & $0.125$ \\
$\sigma=0.1$ & STSP & 2.5 & 8 & $0.25$ \\
$\sigma=0.1$ & MS-STSP & -- & 6 & $0.3,0.35,0.4$ \\
\hline
$\sigma=0.2$ & Local PCA & 3 & 1 & $0.6$ \\
$\sigma=0.2$ & Mean shift & 1.5 & 30 & $0.3$ \\
$\sigma=0.2$ & SCMS-KDE & 1.25 & 30 & $0.25$ \\
$\sigma=0.2$ & SCMS-log-KDE & 0.75 & 17 & $0.15$ \\
$\sigma=0.2$ & STSP & 1.5 & 18 & $0.3$ \\
$\sigma=0.2$ & MS-STSP & -- & 6 & $0.6,0.7,0.8$ \\
\hline
$\sigma=0.3$ & Local PCA & 3 & 1 & $0.9$ \\
$\sigma=0.3$ & Mean shift & 1.25 & 30 & $0.375$ \\
$\sigma=0.3$ & SCMS-KDE & 0.5 & 30 & $0.15$ \\
$\sigma=0.3$ & SCMS-log-KDE & 0.75 & 17 & $0.225$ \\
$\sigma=0.3$ & STSP & 1.5 & 17 & $0.45$ \\
$\sigma=0.3$ & MS-STSP & -- & 9 & $0.75,0.9,1.05$ \\
\hline
\end{tabular}

\end{table}

\begin{figure}[H]
\centering
\includegraphics[width=.72\linewidth,height=.44\textheight,keepaspectratio]{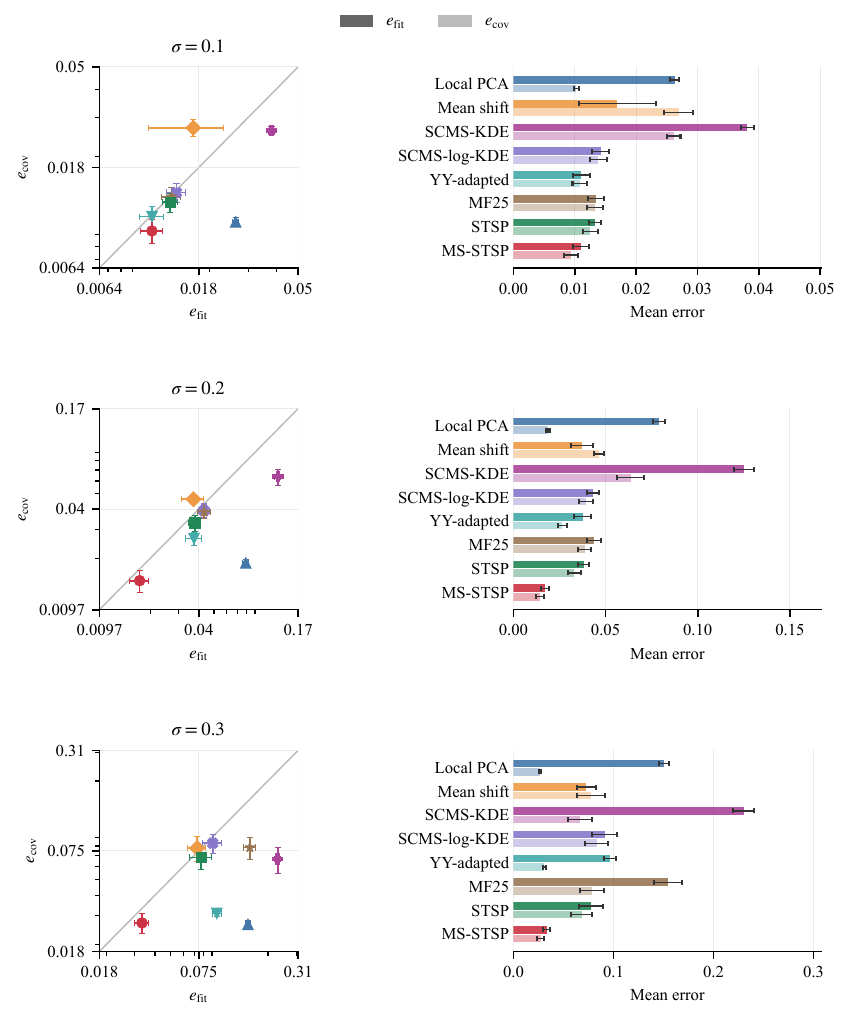}
\caption[Fitting and coverage errors for the high-noise circle]{Fitting and coverage errors over ten independent realisations
of the high-noise circle. Each row gives a scatter plot of method means
and paired horizontal bars for $e_{\mathrm{fit}}$ and $e_{\mathrm{cov}}$.
Error bars show one standard deviation; the grey diagonal marks equal
errors in the scatter plot.}
\label{app:fig:supp-high-noise-joint}
\end{figure}

\begin{figure}[H]
\centering
\includegraphics[width=.92\linewidth]{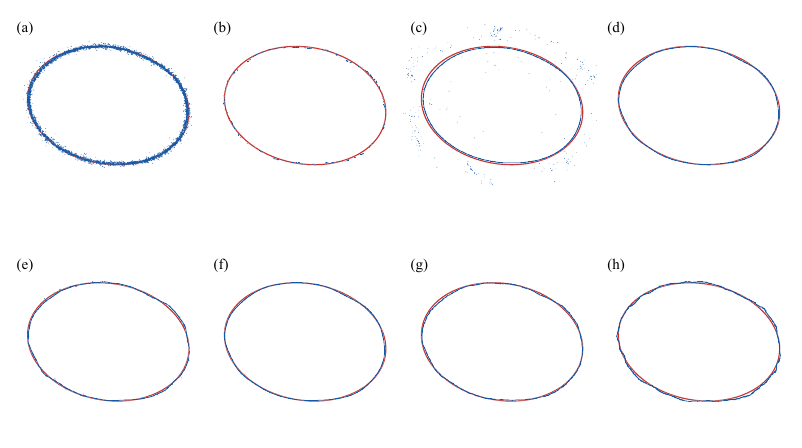}
\caption{Circle fitted clouds at $\sigma=0.1$. Panels (a)--(h) show local PCA,
mean shift, SCMS-KDE, SCMS-log-KDE, YY-adapted, MF25, STSP, and MS-STSP,
respectively.}
\label{app:fig:supp-high-circle-01}
\end{figure}

\begin{figure}[H]
\centering
\includegraphics[width=.92\linewidth]{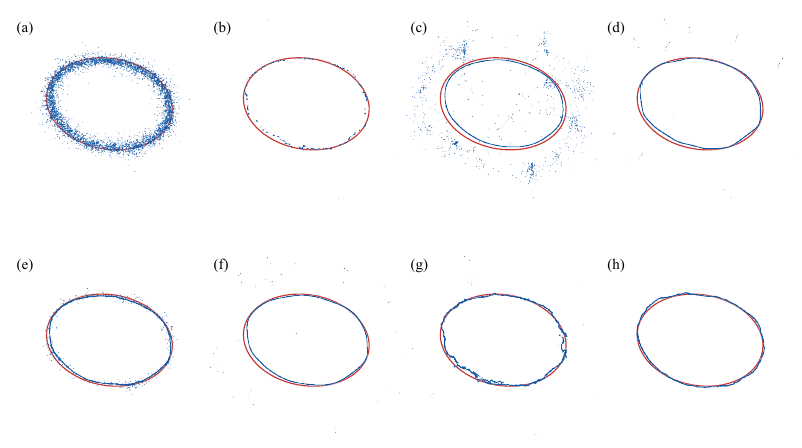}
\caption{Circle fitted clouds at $\sigma=0.2$. Panels (a)--(h) show local PCA,
mean shift, SCMS-KDE, SCMS-log-KDE, YY-adapted, MF25, STSP, and MS-STSP,
respectively.}
\label{app:fig:supp-high-circle-02}
\end{figure}

\begin{figure}[H]
\centering
\includegraphics[width=.92\linewidth]{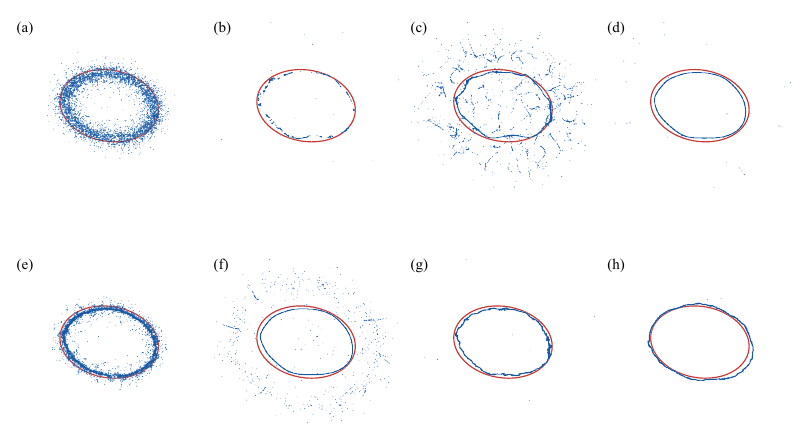}
\caption{Circle fitted clouds at $\sigma=0.3$. Panels (a)--(h) show local PCA,
mean shift, SCMS-KDE, SCMS-log-KDE, YY-adapted, MF25, STSP, and MS-STSP,
respectively.}
\label{app:fig:supp-high-circle-03}
\end{figure}

\subsection{Rate, convergence, and radius diagnostics}
The expanded rate diagnostic uses hard-radius STSP with $N=5000$,
$h=3\sigma$, $T_{\max}=30$, the displacement stopping rule above,
10 replicates, and
\[
 \sigma\in\{.01,.02,.04,.08,.10,.15,.20,.25,.30,.35\}
\]
on the circle and sphere. The iteration diagnostic uses the same two
manifolds, $N=5000$, $\sigma=.08$, $h=3\sigma$, 30 iterations, and 10
replicates. The radius-sensitivity study uses the circle, 15 geometrically
spaced radii from $.2\sigma$ to $.5$, 10 iterations, and ten replicates
at each of the two displayed noise levels.

\subsection{Dentate gyrus configuration}

The supplied count matrix contains 2,930 cells and 14 annotations:
astrocytes (120), Cajal--Retzius cells (34), Cck--Tox cells (27),
endothelial cells (87), GABAergic cells (61), immature granule cells
(785), mature granule cells (1,070), microglia (81), mossy cells (75),
neuroblasts (417), oligodendrocytes (50), oligodendrocyte precursor cells
(53), radial glia-like cells (51), and neuronal intermediate progenitor
cells (19). Counts are normalised to a library size of $10^4$ and
log-transformed. The 2,000 genes with the largest variance are retained,
standardised coordinatewise, and clipped to $[-10,10]$. A randomised PCA
with seed 20260925 produces one nested 100-component representation.
For each displayed value of $p$, the first $p$ components are centred and
divided by their root-mean-square cloud radius.

\begin{samepage}
The principal analysis uses $p=30$. The dimension comparison uses
$p\in\{20,30,40,50\}$. Mean shift, local PCA, STSP, MS-STSP, both SCMS
variants, YY-adapted, and MF25 use the common candidate scale grid
formed by the upper 20 values of a 30-point logarithmically spaced
sequence from $0.025$ to $2$. Equivalently, its unrounded values are
$h_j=0.025(2/0.025)^{(j+10)/29}$ for $j=0,\ldots,19$; to six decimal
places, this gives
\[
\begin{split}
 h\in\{&0.113286,0.131766,0.153259,0.178258,0.207335,\\
        &0.241155,0.280492,0.326245,0.379462,0.441359,\\
        &0.513353,0.597090,0.694486,0.807769,0.939531,\\
        &1.092785,1.271038,1.478368,1.719516,2\}.
\end{split}
\]
\end{samepage}
Mean shift, local PCA, and STSP use the reference observations within
radius $h$ of each query.
The SCMS variants use Gaussian weights with bandwidth $h$;
YY-adapted and MF25 retain their method-specific neighbourhoods.
MS-STSP uses three distinct radii; all $\binom{20}{3}=1140$ triples from
the 20-scale grid satisfy $\sum_j|w_j|\le5$ and are evaluated at each $p$.
The triple with the largest full-cell average silhouette is selected.
Each branch applies the fixed-reference
STSP update, with multi-scale weights as in Algorithm~2.
Methods requiring an intrinsic dimension use $d=1$. Mean shift and STSP
retain the original PCA cloud as the reference throughout their
iterations. With the normalised input radius $s=1$, mean shift, STSP,
and both SCMS variants stop at the first iteration satisfying
\[
 \frac{1}{N}\sum_{i=1}^N
 \|\mathbf y_i^{(t)}-\mathbf y_i^{(t-1)}\|<10^{-3},
\]
or at $T_{\max}=30$. MS-STSP uses the same displacement criterion for
each of its three branches, stopping when all branches satisfy it or
at $T_{\max}=30$. Local PCA applies one projection. YY-adapted uses
its published stopping rule. MF25 uses 10 initialization iterations,
50 maximum outer iterations and tolerance $10^{-8}$. At each $p$, the
full-cell average Euclidean silhouette with respect to the 14 supplied
annotations selects the radius or, for MS-STSP, the radius triple.

At $p=30$, the selected radii are 0.513353 for mean shift, 0.939531 for
local PCA, 0.441359 for SCMS-KDE, 0.326245 for SCMS-log-KDE, 0.326245 for
YY-adapted, 0.939531 for MF25, and 0.694486 for STSP. Mean shift stops
after 18 updates and STSP after 11. At $p=(20,30,40,50)$, the selected
MS-STSP radius triples are $(0.694486,1.478368,1.719516)$,
$(0.694486,1.271038,1.719516)$, $(0.807769,1.271038,1.719516)$, and
$(0.807769,1.271038,1.719516)$, respectively. Their stopping times are
$(23,16,14,19)$. At the same dimensions, the selected
STSP radii are $(0.694486,\allowbreak 0.694486,\allowbreak 0.807769,\allowbreak 0.807769)$ and the
corresponding stopping times are $(10,11,9,11)$. The selected mean-shift
radii are $(0.513353,\allowbreak 0.513353,\allowbreak 0.597090,\allowbreak 0.694486)$, with stopping times
$(17,18,19,16)$. Both STSP and mean shift reach the displacement
threshold before the iteration cap at these dimensions.

For each selected $p=30$ output, $k$-means uses 14 clusters, random seeds
0--9, five $k$-means++ initializations per seed, and at most 500 iterations.
ARI and geometric NMI compare each partition with the supplied annotations;
accuracy uses an optimal one-to-one matching between cluster and annotation
labels. The main-paper table reports averages over the 10 seeds.

For the main-paper t-SNE figure, the selected STSP and MS-STSP outputs
are separately centred and divided by their root-mean-square radii.
Each is embedded with perplexity 80, early exaggeration 24, PCA
initialisation, automatic learning rate, 1,500 optimisation iterations,
and random seed 20260927. Both embeddings are aligned by a global
orthogonal transformation to the independently computed input embedding;
this allows rotation and reflection but no local deformation. The
annotations determine colours in the t-SNE display and are not supplied
to the embedding algorithm; they also enter the SIL criterion used to
select the fitted representation.

\bibliographystyle{plainnat}
\bibliography{paper-ref}
\end{document}